\documentclass[
a4paper,reqno]{amsart}
\usepackage[T1]{fontenc}
\usepackage[english]{babel}
\usepackage{amssymb,bbm,enumerate}
\usepackage[linktocpage=true,colorlinks=true, linkcolor=blue, citecolor=red, urlcolor=green]{hyperref}
\usepackage[all]{xy}
\usepackage{color}

\newcommand{\mc}[1]{\mathcal{#1}}
\newcommand{\mf}[1]{\mathfrak{#1}}
\newcommand{\mb}[1]{\mathbb{#1}}

\newcommand{\quot}[2] {\ensuremath{\raisebox{.40ex}{\ensuremath{#1}}
\! \big / \! \raisebox{-.40ex}{\ensuremath{#2}}}}

\newcommand{\tint}{{\textstyle\int}}

\newcommand{\ldb}{\{\!\!\{}
\newcommand{\rdb}{\}\!\!\}}

\DeclareMathOperator{\Mat}{Mat}

\DeclareMathOperator{\tr}{tr}
\DeclareMathOperator{\res}{Res}
\DeclareMathOperator{\Vect}{Vect}

\DeclareMathOperator{\Der}{Der}

\DeclareMathOperator{\mult}{m}

\theoremstyle{plain}
\newtheorem{theorem}{Theorem}[section]
\newtheorem{lemma}[theorem]{Lemma}
\newtheorem{proposition}[theorem]{Proposition}
\newtheorem{corollary}[theorem]{Corollary}

\theoremstyle{definition}
\newtheorem{definition}[theorem]{Definition}
\newtheorem{example}[theorem]{Example}

\theoremstyle{remark}
\newtheorem{remark}[theorem]{Remark}

\numberwithin{equation}{section}

\definecolor{light}{gray}{.9}

\title{Double Poisson vertex algebras and non-commutative Hamiltonian equations}

\author{Alberto De Sole, Victor G. Kac,
Daniele Valeri}
\address{Dipartimento di Matematica, Sapienza Universit\`a di Roma,
P.le Aldo Moro 2, 00185 Rome, Italy.}
\email{desole@mat.uniroma1.it}
\address{Department of Mathematics, MIT,
77 Massachusetts Ave, Cambridge, MA 02139, USA.}
\email{kac@math.mit.edu}
\address{Yau Mathematical Sciences Center, Tsinghua University,
100084 Beijing, China.}
\email{daniele@math.tsinghua.edu.cn}

\begin{document}

\begin{abstract}
We develop the formalism of double Poisson vertex algebras (local and non-local)
aimed at the study of non-commutative Hamiltionan PDEs.
This is a generalization of the
theory of double Poisson algebras, developed by Van den Bergh, which is used in the study of
Hamiltonian ODEs.
We apply our theory of double Poisson vertex algebras to non-commutative KP and Gelfand-Dickey
hierarchies.
We also construct the related non-commutative de Rham and variational complexes. 
\end{abstract}

\pagestyle{plain}
\maketitle
\tableofcontents

\section{Introduction}
First, we recall the notion of a Poisson vertex algebra.
Given a unital commutative associative algebra $\mc V$
with a derivation $\partial$, a $\lambda$-bracket on $\mc V$
is a linear map $\mc V\otimes\mc V\to\mc V[\lambda]$, $a\otimes b\to\{a_\lambda b\}$,
satisfying the following axioms:
\begin{align}
\label{eq:0.1}
\{\partial a_\lambda b\}=-\lambda\{a_\lambda b\}\,,
\quad
\{a_\lambda \partial b\}=(\lambda+\partial)\{a_\lambda b\}
\qquad
&
\text{(sesquilinearity)}\,,
\\
\label{eq:0.2}
\{a_\lambda bc\}=\{a_\lambda b\}c+\{a_\lambda c\}b
\qquad\qquad\qquad
&
\text{(left Leibniz rule)}\,,
\\
\label{eq:0.3}
\{ab_\lambda c\}=\{a_{\lambda+\partial} c\}_\to b+\{b_{\lambda+\partial} c\}_\to a
\qquad\qquad
&
\text{(right Leibniz rule)}\,.
\end{align}
The algebra $\mc V$, endowed with a $\lambda$-bracket satisfying
\begin{align}
\label{eq:0.4}
\{a_\lambda b\}=-\{b_{-\lambda-\partial}a\}
\qquad\qquad\qquad
&
\text{(skewsymmetry)}\,,
\\
\label{eq:0.5}
\{a_\lambda\{b_\mu c\}\}-\{b_\mu\{a_\lambda c\}\}=\{\{a_\lambda b\}_{\lambda+\mu}c\}
\qquad
&
\text{(Jacobi identity)}\,,
\end{align}
is called a \emph{Poisson vertex algebra} (PVA).
(Of course, if $\partial=0$, then \eqref{eq:0.2}--\eqref{eq:0.5} turn into axioms of a
Poisson algebra if we let $\lambda=\mu=0$.)

In the case when $\mc V$ is the algebra of functions on the space of jets of a finite-dimensional
manifold, the notion of a PVA is equivalent to that of a local Poisson bracket \cite{BDSK09}.
In order to accommodate non-local Poisson brackets one needs to introduce the non-local
PVA, for which the $\lambda$-bracket takes values in $\mc V((\lambda^{-1}))$,
the formal Laurent series in $\lambda^{-1}$ (instead of polynomials in $\lambda$).
The axioms \eqref{eq:0.1}--\eqref{eq:0.4} can be easily interpreted, but for the Jacobi identity
\eqref{eq:0.5} to make sense
in the non-local case, one needs, in addition, the \emph{admissibility} property of the $\lambda$-bracket:
\begin{equation}\label{eq:0.6}
\{a_\lambda\{b_\mu c\}\}\in\mc V[[\lambda^{-1},\mu^{-1},(\lambda+\mu)^{-1}]][[\lambda,\mu]]\,,
\text{ for all } a,b,c\in\mc V
\,,
\end{equation}
 (see \cite{DSK13} for details).

Local and non-local Poisson brackets play a fundamental role in the theory of integrable systems
(see e.g. \cite{FT86}).
In recent years there have been attempts to develop a theory of integrable  systems on non-commutative
associative algebras (see e.g. \cite{DF92,OS98,Kup00,MS00,ORS13}).
An important advance in this direction was made by Van den Bergh, who in \cite{VdB08} introduced the
notion of a double Poisson algebra structure
in a non-commutative associative algebra $V$.
His basic idea was to consider a $2$-fold bracket $\ldb-,-\rdb$ on $V$,
 with values in $V\otimes V$.
The Leibniz rules of a $2$-fold bracket are almost identical to the usual Leibniz rules:
\begin{align}
&
\ldb a,bc\rdb=\ldb a,b\rdb c+b\ldb a,c\rdb\,,
\label{eq:0.7}
\\
&\ldb ab,c\rdb=\ldb a,c\rdb\star_1 b+a\star_1\ldb b,c\rdb\,,
\label{eq:0.8}
\end{align}
where $(a_1\otimes a_2)\star_1 a_3=a_1a_3\otimes a_2$,
$a_1\star_1(a_2\otimes a_3)=a_2\otimes a_1a_3$.
The skewsymmetry axiom is
\begin{equation}\label{eq:0.9}
\ldb a,b\rdb=-\ldb b,a\rdb^\sigma\,,
\end{equation}
where $\sigma$ is the permutation of factors in $V\otimes V$,
and the Jacobi identity is
\begin{equation}\label{eq:0.10}
\ldb a,\ldb b,c\rdb\rdb_L-\ldb b,\ldb a,c\rdb\rdb_R=\ldb\ldb a,b\rdb,c\rdb_L\,,
\end{equation}
where we denote $\ldb a_1,a_2\otimes a_3\rdb_L=\ldb a_1,a_2\rdb\otimes a_3$,
$\ldb a_1,a_2\otimes a_3\rdb_R=a_2\otimes\ldb a_1,a_3\rdb$,
$\ldb a_1\otimes a_2,a_3\rdb_L=\ldb a_1,a_3\rdb\otimes_1a_2$,
and $(a\otimes b)\otimes_1c=a\otimes c\otimes b$.
(Formula \eqref{eq:0.10} is slightly different from
but equivalent to the Jacobi identity in \cite{VdB08}.)

An associative algebra $V$, endowed with a $2$-fold bracket $\ldb-,-\rdb$ satisfying
axioms \eqref{eq:0.9} and \eqref{eq:0.10}, is called a \emph{double Poisson algebra}.

Given a double Poisson algebra $V$, Van den Bergh defines the following bracket on $V$ with values
in $V$:
\begin{equation}\label{eq:0.11}
\{a,b\}=\mult\ldb a,b\rdb\,,
\end{equation}
where $\mult:V\otimes V\to V$ is the multiplication map.
The bracket \eqref{eq:0.11} still satisfies the left Leibniz rule, but does not satisfy, in general,
other axioms of a Poisson bracket. However, this bracket induces well-defined
linear maps
$$
\quot{V}{[V,V]}\otimes V\to V
\qquad
\text{and}
\qquad
\quot{V}{[V,V]}\otimes \quot{V}{[V,V]}\to \quot{V}{[V,V]}\,,
$$
given by
\begin{equation}\label{eq:0.12}
\{\tr(a),b\}=\{a,b\}
\qquad\text{and}\qquad
\{\tr(a),\tr(b)\}=\tr\{a,b\}\,,
\end{equation}
where $\tr:V\to\quot{V}{[V,V]}$ is the quotient map and $[V,V]$ is the linear span of commutators
$ab-ba$.
These maps have the following properties, important for the theory of non-commutative Hamiltonian
ODEs: the vector space $\quot{V}{[V,V]}$ is a Lie algebra and one has
its representation on $V$ by derivations,
both defined by \eqref{eq:0.12} (see Proposition \ref{20131205:prop1}).

Brackets \eqref{eq:0.12} allow one to define the basic notions of a Hamiltonian theory.
Given a \emph{Hamiltonian function} $\tr(h)$, $h\in V$, one defines the associated
\emph{Hamiltonian equation}
\begin{equation}\label{eq:0.13}
\frac{dx}{dt}=\{\tr(h),x\}\,,
\quad x\in V\,.
\end{equation}
Two Hamiltonian functions $\tr(f)$ and $\tr(g)$ are said to be in involution if
\begin{equation}\label{eq:0.14}
\{\tr(f),\tr(g)\}=0\,.
\end{equation}

In the case of $V=R_\ell$, the algebra of non-commutative polynomials in $\ell$ variables
$x_1,\dots,x_\ell$, any $2$-fold bracket can be written in a traditional form ($f,g\in R_\ell$):
\begin{equation}\label{eq:0.15}
\ldb f,g\rdb =\nabla g\bullet H\bullet (\nabla f)^{\sigma}\,,
\end{equation}
where $\nabla f$, the \emph{gradient} of $f$, is the vector of $2$-fold derivatives
$\frac{\partial f}{\partial x_i}:R_\ell\to R_\ell\otimes R_\ell$, defined by
$\frac{\partial f}{\partial x_i}(x_j)=\delta_{ij}(1\otimes1)$ (see \cite{CBEG07}),
$H=(H_{ij})_{i,j=1}^\ell$, where $H_{ij}=\ldb x_j,x_i\rdb\in V\otimes V$, and
$\bullet$ denotes the multiplication in $V\otimes V^{op}$.
Moreover, the $2$-fold bracket \eqref{eq:0.15} is skewsymmetric if and only if skewsymmetry
holds on any pair of generators $x_i,x_j$
(equivalently, if the matrix $H$ is skewadjoint: $(H^t)^\sigma=-H$),
and, provided that $H$ is skewadjoint, the Jacobi identity \eqref{eq:0.10} holds
if and only if it holds on any triple of generators $x_i$, $x_j$, $x_k$
(see Theorem \ref{thm:master-finite}).

In the case of $V=R_\ell$ the Hamiltonian equation \eqref{eq:0.13} becomes the following
evolution ODE, where $x=(x_i)_{i=1}^\ell$:
\begin{equation}\label{eq:0.16}
\frac{dx}{dt}
=\mult\left(H(\nabla h)^\sigma\right)
\,,
\end{equation}
and the bracket $\{-,-\}$ on $\quot{R_\ell}{[R_\ell,R_\ell]}$ becomes:
\begin{equation}\label{eq:0.17}
\{\tr(f),\tr(g)\}
=\tr\left(\mult\left(\nabla g\bullet H\bullet (\nabla h)^\sigma\right)\right)
\,.
\end{equation}

Developing the ideas of \cite{MS00}, we study in detail the Euler hierarchy,
obtained via the Lenard-Magri scheme from the following pair of compatible
double Poisson brackets on $R_2=\mb F\langle x,y\rangle$, where $y$ is
central for both $2$-fold brackets:
\begin{equation}\label{eq:0.18}
\ldb x,x\rdb_0=1\otimes y-y\otimes1\,,
\quad
\ldb x,x\rdb_1=x\otimes y-y\otimes x\,.
\end{equation}
Letting
\begin{equation}\label{eq:0.19}
h_0=1
\quad\text{and}\quad
h_n=\frac{1}{n}(x+y)^n
\quad\text{for }n\geq1\,,
\end{equation}
it is easy to establish the following recursion relations:
$$
\{\tr(h_0),x\}_0=0
\quad\text{and}\quad
\{\tr(h_n),x\}_1=\{\tr(h_{n+1}),x\}_0
\quad\text{for }n\in\mb Z_+\,.
$$
Hence, by the Lenard-Magri scheme, all $\tr(h_n)$ are in involution with respect to both brackets,
and we obtain, by \eqref{eq:0.16}, the following compatible set of Hamiltonian ODEs ($m\in\mb Z_+$):
\begin{equation}\label{eq:0.20}
\frac{dx}{dt_m}=x(x+y)^my-y(x+y)^mx\,,
\quad
\frac{dy}{dt_m}=0\,,
\end{equation}
for which all $h_n$, $n\in\mb Z_+$, are conserved densities.
Thus, by definition in \cite{MS00}, the hierarchy of Hamiltonian ODEs \eqref{eq:0.20}
is integrable (see Section \ref{sec:2.6} for details).

We develop this theory in the natural generality of a (non-commutative)
\emph{algebra of ordinary differential functions}. This is a unital associative
algebra $V$ endowed with $\ell$ strongly commuting $2$-fold derivations
$\frac{\partial}{\partial x_i}$. We say that the $2$-fold derivations $D$ and $E$
\emph{strongly commute} if $D_L\circ E=E_R\circ D$, where
$D_L(a_1\otimes a_2)=(D_La_1)\otimes a_2$ and
$D_R(a_1\otimes a_2)=a_1\otimes(D_Ra_2)$
(then $D$ and $E$ commute).

Formula \eqref{eq:0.15} defines a $2$-fold bracket on any algebra of ordinary differential
functions $V$, and $V$ is a double Poisson algebra if and only if the matrix $H$ is skewadjoint
and \eqref{eq:0.10} holds on any ``triple of generators''
(see Theorem \ref{thm:master-finite} for details).

The assumption of strong commutativity of the $\frac{\partial}{\partial x_i}$ is used in the proof
of the Jacobi identity for the $2$-fold bracket \eqref{eq:0.15} and the proof that $d^2=0$
for the de Rham complex $\widetilde{\Omega}(V)$, discussed below.

For an algebra of differential functions $V$ we construct in Section \ref{sec:derham}
the de Rham complex $\widetilde{\Omega}(V)$ as a free product of $V$
and the algebra of non-commutative polynomials in the indeterminates $dx_1,\dots,dx_\ell$.
This is a superalgebra with the $\mb Z_+$-grading given by $\deg(V)=0$,
$\deg(dx_i)=1$, compatible with the parity,
and with the de Rham differential $d$ being
the odd derivation such that $d(dx_i)=0$ and
$df=\sum_{i=1}^\ell \mult\left(\frac{\partial f}{\partial x_i}\otimes_1dx_i\right)$
for $f\in\widetilde{\Omega}^0(V)=V$.

The relevant to the theory of non-commutative evolution ODEs is the reduced complex
\begin{equation}\label{eq:0.21}
\Omega(V)=\quot{\widetilde{\Omega}(V)}{[\widetilde{\Omega}(V),\widetilde{\Omega}(V)]}\,.
\end{equation}
In Section \ref{sec:red-com} we give an explicit description of the space of $k$-forms
$\Omega^k(V)$. In particular: $\Omega^0(V)=\quot{V}{[V,V]}$, $\Omega^1(V)=V^{\oplus\ell}$,
$\Omega^2(V)=\left(\Mat_{\ell\times\ell}(V\otimes V)\right)_-$,
where the subscript $\_$ means that $(A_{ij})_{i,j=1}^{\ell}=-(A_{ji}^\sigma)_{i,j=1}^\ell$.
Under these identifications the formula for $d$ becomes:
\begin{equation}\label{eq:0.22}
d(\tr f)=\mult(\nabla f)^{\sigma}\,,
\quad
dF=\frac12(J_F^t-J_F^\sigma)\,,
\end{equation}
where for $F=(F_i)_{i=1}^\ell\in V^{\oplus\ell}$ we define the $2$-fold Jacobian matrix as
$J_F=\left(\frac{\partial F_i}{\partial x_j}\right)_{i,j=1}^\ell$.
We also prove that both complexes $(\widetilde{\Omega}(R_\ell),d)$ and
$(\Omega(R_\ell),d)$ are acyclic.

The theory of double Poisson algebras and non-commutative Hamiltonian ODEs, over an algebra
of ordinary differential functions $V$, developed in Section \ref{sec:2} sets the stage for the
theory of double PVAs and non-commutative Hamiltonian PDEs over an algebra of differential
functions $\mc V$, which we develop in Section \ref{sec:3}.

In a nutshell a double PVA is a cross between a double Poisson algebra and a PVA. Let $\mc V$ be
unital associative algebra with a derivation $\partial$.
We define a \emph{$2$-fold $\lambda$-bracket} on $\mc V$ as a map
$\ldb-_\lambda-\rdb:\mc V\otimes\mc V\to(\mc V\otimes\mc V)[\lambda]$,
satisfying sesquilinearity \eqref{20140702:eq4b} and left and right Leibniz rules
\eqref{20140702:eq6b}, similar to \eqref{eq:0.1} and \eqref{eq:0.2},\eqref{eq:0.3} for PVA.
A \emph{double PVA} is the algebra $\mc V$, endowed with a $2$-fold $\lambda$-bracket
satisfying skewsymmetry \eqref{eq:skew2} and Jacobi identity \eqref{eq:jacobi2}
($=$ cross between \eqref{eq:0.4} and \eqref{eq:0.8},
and between \eqref{eq:0.5} and \eqref{eq:0.10},
respectively).
If $\partial =0$, the axioms of a double PVA turn into the axioms of a double Poisson algebra if we let $\lambda=\mu=0$.

A non-local double PVA is defined in the same way as a non-local PVA
(but the one cannot put $\lambda=0$).

For a (local) double PVA $\mc V$, we denote by $\tint$ the quotient map
$\mc V\to\overline{\mc V}:=\quot{\mc V}{([\mc V,\mc V]+\partial\mc V)}$.
Similarly to \eqref{eq:0.11}, define the map $\{-_\lambda-\}:\mc V\otimes\mc V\to \mc V[\lambda]$ by
\begin{equation}\label{eq:0.23}
\{a_\lambda b\}=\mult\ldb a_\lambda b\rdb
\,.
\end{equation}
As for the (local) PVA (see \cite{BDSK09}), \eqref{eq:0.23} induces well-defined maps
$$
\overline{\mc V}\otimes\mc V\to\mc V
\quad\text{and}\quad
\overline{\mc V}\otimes\overline{\mc V}\to\overline{\mc V}
\,,
$$
given by formulas, similar to \eqref{eq:0.12}:
\begin{equation}\label{eq:0.24}
\{\tint a, b\}=\{a_\lambda b\}\big|_{\lambda=0}
\quad\text{and}\quad
\{\tint a,\tint b\}=\tint\{a_\lambda b\}\big|_{\lambda=0}
\,.
\end{equation}
As before, these maps have properties important for the theory of non-commutative
Hamiltonian PDEs: $\overline{\mc V}$ is a Lie algebra and one has its representation on $\mc V$
by derivations, both defined by \eqref{eq:0.24} (see Theorem \ref{20140707:thm}).
We have definitions of a Hamiltonian function $\tint h\in\overline{\mc V}$, the corresponding
Hamiltonian equation in $\mc V$, similar to \eqref{eq:0.13},
and involutiveness, similar to \eqref{eq:0.14}, where $\tr$ is replaced by $\tint$.

We define an \emph{algebra of (non-commutative) differential functions} as a unital associative
algebra $\mc V$ with a derivation $\partial$ and strongly commuting $2$-fold derivations
$\frac{\partial}{\partial u_i^{(n)}}$, $i=1,\dots,\ell$, $n\in\mb Z_+$, such that the following
two properties hold (cf. \cite{BDSK09}):
\begin{equation}\label{eq:0.25}
\text{for each }f\in\mc V\,,
\frac{\partial f}{\partial u_i^{(n)}}=0
\text{ for all but finitely many }(i,n)\,, 
\end{equation}
\begin{equation}\label{eq:0.26}
\left[\frac{\partial}{\partial u_i^{(n)}},\partial\right]=\frac{\partial}{\partial u_i^{(n-1)}}\,.
\end{equation}
The most important example of an algebra of non-commutative
differential functions is the algebra of non-commutative differential polynomials
$\mc R_\ell$ in the indeterminates $u_i^{(n)}$, $i=1,\dots,\ell$, $n\in\mb Z_+$, with the
derivation $\partial$ defined by $\partial u_i^{(n)}=u_i^{(n+1)}$.

First, we prove a PVA analogue of Theorem \ref{thm:master-finite}. Namely,
any $2$-fold $\lambda$-bracket on $\mc R_\ell$ has the form (cf. \eqref{eq:0.15}):
\begin{equation}\label{eq:0.27}
\ldb f_{\lambda}g\rdb
=\sum_{\substack{i,j\in I\\m,n\in\mb Z_+}}
\frac{\partial g}{\partial u_j^{(n)}}
\bullet
(\lambda+\partial)^n
H_{ij}(\lambda+\partial)
(-\lambda-\partial)^m
\bullet
\left(\frac{\partial f}{\partial u_i^{(m)}}\right)^\sigma\,,
\end{equation}
where $H_{ij}(\lambda)=\ldb u_j{}_\lambda u_i\rdb\in(\mc R_\ell\otimes\mc R_\ell)[\lambda]$. 
Formula \eqref{eq:0.27} defines a $2$-fold $\lambda$-bracket on any algebra of differential functions
$\mc V$. This $2$-fold $\lambda$-bracket is skewsymmetric if and only if
$H(\partial)=\left(H_{ij}(\partial)\right)_{i,j=1}^\ell$ is a skewadjoint differential operator over
$\mc V\otimes\mc V$ and, provided that $H(\partial)$ is skewadjoint, the Jacobi identity holds
if and only if it holds on any triple of generators (see Theorem \ref{20130921:prop1}).

Next, the non-commutative Hamiltonian PDE, associated to the matrix differential operator
$H(\partial)$, defining the $2$-fold $\lambda$-bracket, and to the Hamiltionan functional
$\tint h\in\overline{\mc V}$ is the following evolution PDE, where $u=(u_i)_{i=1}^\ell$
(cf. \eqref{eq:0.16}):
\begin{equation}\label{eq:0.28}
\frac{du}{dt}
=\mult(H(\partial)\bullet(\delta h)^\sigma)
\,,
\end{equation}
where $\delta h
=\Big(\frac{\delta h}{\delta u_{i}}\Big)_{i=1}^\ell\in(\mc V\otimes\mc V)^{\oplus\ell}$
is the vector of $2$-fold variational derivatives
$$
\frac{\delta h}{\delta u_{i}}
=\sum_{n\in\mb Z_+}(-\partial)^n\frac{\partial h}{\partial u_{i}^{(n)}}\,.
$$
The bracket $\{-,-\}$ on $\overline{\mc V}$ is given by a formula, similar to \eqref{eq:0.17}:
\begin{equation}\label{eq:0.29}
\{\tint f,\tint g\}
=\tint \mult\left(\delta g\bullet H(\partial)\bullet (\delta f)^\sigma
\right)\,.
\end{equation}

Furthermore the de Rham complex $\widetilde{\Omega}(\mc V)$ over an algebra of differential
functions $\mc V$ is defined in the same way as the de Rham complex $\widetilde{\Omega}(V)$ over an
algebra of ordinary differential functions $V$.
Moreover, the action of $\partial$ on $\mc V$ naturally extends
to an action on $\widetilde{\Omega}(\mc V)$, commuting with the de Rham differential $\delta$
(see Section \ref{sec:3.5}).
This allows us to define the \emph{variational complex}
$$
\Omega(\mc V)
=\quot{\widetilde{\Omega}(\mc V)}
{(\partial\widetilde{\Omega}(\mc V)+[\widetilde{\Omega}(\mc V),\widetilde{\Omega}(\mc V)])}
\,.
$$
In Section \ref{sec:3.6} we give an explicit description of the space of the variational $k$-forms
$\Omega^k(\mc V)$.
In particular,
$$
\Omega^0(\mc V)=\quot{\mc V}{(\partial\mc V+[\mc V,\mc V])}\,,
\quad
\Omega^1(\mc V)=\mc V^{\oplus\ell}\,,
\quad
\Omega^2(\mc V)=\left(\Mat_{\ell\times\ell}(\mc V\otimes\mc V)[\lambda]
\right)_-\,,
$$
where the subscript $\_$ means that $(A_{ij}(\lambda))_{i,j=1}^\ell
=-(A_{ji}(-\lambda-\partial)^\sigma)_{i,j=1}^\ell$.
Under these identifications, the formula for $\delta$ is similar to \eqref{eq:0.22}:
\begin{equation}\label{eq:0.30}
\delta(\tint f)=\mult(\delta f)^\sigma\,,
\quad
(\delta F)(\partial)=\frac12(D_F(\partial)^t-D_F(\partial)^{*\sigma})
\,,
\end{equation}
where $(D_F(\partial))_{ij}=\sum_{n\in\mb Z_+}\frac{\partial F_i}{\partial u_j^{(n)}}\partial^n$
is the $2$-fold Frechet derivative.

We also prove that both complexes $(\widetilde{\Omega}(\mc R_\ell),\delta)$ and
$(\Omega(\mc R_\ell),\delta)$ are acyclic.
Hence, in particular, we obtain the description of the kernel and the image of the $2$-fold
variational derivative $\delta$.

The main motivation for introducing double Poisson algebras was the observation that, given a unital
finitely generated associative algebra $V$, each double Poisson algebra structure on $V$ associates for
each positive integer $m$ a Poisson algebra structure on the algebra of polynomial functions $V_m$ on
the affine algebraic variety of $m$-dimensional representations of the algebra $V$ \cite{VdB08}.

Note that $V_m$ is isomorphic to the commutative
associative algebra with generators
$\{a_{ij}\mid a\in V, i,j=1,\dots,m\}$, subject to the relations
\eqref{20130917:eq2} in Section \ref{sec:V_m}.

If $\mc V$ is a (non-commutative) differential associative algebra
with derivation $\partial$, then $\mc V_m$ is a
commutative differential associative algebra with derivation defined
on generators by $\partial(a_{ij})=(\partial a)_{ij}$.
Similarly to \cite{VdB08}, we show that, given a double PVA structure on $\mc V$, one can associate
for each positive integer $m$ a PVA structure on $\mc V_m$
(see Theorem \ref{20130921:cor1}).

The simplest example of a double PVA is the following family of compatible double PVA structures
on $\mc R_1=\mb F\langle u, u', u'',\dots\rangle$:
$$
\ldb u_\lambda u\rdb=1\otimes u-u\otimes1+c(1\otimes1)\lambda\,,
$$
where $c\in\mb F$ is a parameter. In this case the associated PVA structure on
$(\mc R_1)_m=\mb F[u_{ij}^{(n)}\mid i,j=1,\dots,m,n\in\mb Z_+]$
is just the affine PVA for $\mf{gl}_m$ (see Section \ref{sec:affine}).
The extension of this double PVA structure on $\mc R_1$ to the double
PVA structure on $\mc R_2=\mb F\langle u,v,u',v',\dots\rangle$, where $v$ is central, 
produces the following hierarchy of non-commutative Hamiltonian PDE's ($n\in\mb Z_+$):
$$
\frac{du}{dt_n}
= 
v(u+v)^n u-u(u+v)^n v +c\partial (u+v)^{n+1}
\,\,,\,\,\,\, 
\frac{dv}{dt_n}
=0
\,.
$$

Our next example is a compatible pair of double PVA structures on the algebra $\mc R_N$
(resp. $\mc R_\infty$), associated to the ``generic'' differential (resp. pseudodifferential)
operator with coefficients in $\mc R_N$ (resp. $\mc R_\infty$), similar
to the PVA structures in the commutative case, considered in \cite{DSKV14a}
(see Section \ref{sec:AGD}).

We apply for the latter example the theory of Dirac reduction for (non-local) double PVA
developed in Section \ref{sec:dirac}, similar to that for (non-local) PVA,
developed in \cite{DSKV14b}.
As a result, we obtain compatible non-local double Poisson vertex algebra structures on 
$\mc V=\mc R_N$ and $\mc R_\infty$, such that the corresponding non-local PVA structures on
$\mc V_m$ coincide with those in \cite{DSKV14a} (see Example \ref{exa:5.9}).

Finally, in Section \ref{sec:hierarchies} we apply the Lenard-Magri scheme to the double PVA
constructed in Section \ref{sec:AGD} to prove integrability of the non-commutative KP hierarchy
and of non-commutative Gelfand-Dickey hierarchies. In particular, we prove integrability of the
non-commutative KdV and Boussinesq hierarchies.

Throughout the paper we let $\mb F$ be a field of characteristic $0$,
and, unless otherwise specified, we consider all vector spaces,
tensor products, etc., over the field $\mb F$.

\bigskip
\noindent\textit{Acknowledgments.}
We are grateful to Vladimir Sokolov for introducing us to the fascinating subject of non-commutative integrable equations.

Part of the work was completed during 
the visits of the second and third author to the University of Rome La Sapienza,
during the visit of the first author to SISSA in Trieste,
and during the visits of all three authors to IHES, France.
We thank all these institutions for their kind hospitality.
The first author is supported by the national FIRB grant RBFR12RA9W ``Perspectives in Lie Theory'',
and the national PRIN grant number  2012KNL88Y$\_$001.
The second author is supported by a national NSF grant.
The third author was supported
by the ERC grant ``FroM-PDE: Frobenius Manifolds and Hamiltonian Partial Differential Equations''.

\section{Preliminaries and notation for associative algebras}

Throughout this and the next section, we let $V$ be an associative algebra
over the field $\mb F$.

\subsection{``Double'' linear algebra}

Consider the associative product $\bullet$ on $V\otimes V$
given, in Sweedler's notation, by the formula
\begin{equation}\label{20140609:eqc1}
A\bullet B=A^\prime B^\prime\otimes B^{\prime\prime}A^{\prime\prime}
\,.
\end{equation}
(It is the product in $V\otimes V^{op}$.)
The endomorphism $\sigma$ of $V\otimes V$
obtained by exchanging the two factors:
\begin{equation}\label{20140710:eq1}
(a\otimes b)^\sigma=b\otimes a
\,,
\end{equation}
is an antiautomorphism of the $\bullet$-product:
\begin{equation}\label{20140710:eq4}
(A\bullet B)^\sigma=B^\sigma\bullet A^\sigma
\,.
\end{equation}
We define an inner product $(\cdot\,|\,\cdot)$ on the vector space $(V\otimes V)^{\oplus\ell}$
with values in the algebra $V\otimes V$,
letting for $F=(F_i)_{i=1}^\ell$ and $G=(G_i)_{i=1}^\ell$ ($F_i,G_i\in V\otimes V$):
\begin{equation}\label{20140710:eq2}
(F|G)=\sum_{i=1}^\ell F_i\bullet G_i^\sigma
\,.
\end{equation}
By \eqref{20140710:eq4}, it satisfies the following symmetry property:
\begin{equation}\label{20140710:eq5}
(G|F)^\sigma=(F|G)
\,.
\end{equation}
Let $H=(H_{ij})_{i,j=1}^\ell\in\Mat_{\ell\times\ell}(V\otimes V)$.
We define its action on $(V\otimes V)^{\oplus\ell}$ by the following formula:
\begin{equation}\label{20140710:eq3}
(HF)_i
=
\sum_{j=1}^\ell (H_{ij}\bullet F_j^\sigma)^\sigma
\,.
\end{equation}
In Sweedler's notation, it is
$(HF)_i=\sum_{j=1}^\ell  F_j^\prime H_{ij}^{\prime\prime} \otimes H_{ij}^\prime F_j^{\prime\prime}$.
The \emph{adjoint} of the matrix $H$
is defined by
\begin{equation}\label{20140710:eq6}
(H^\dagger)_{ij}=(H_{ji})^\sigma
\,.
\end{equation}
It is adjoint with respect to the inner product \eqref{20140710:eq2} in the usual sense:
\begin{equation}\label{20140710:eq7}
(F|HG)=(H^\dagger F|G)
\,\,,\,\,\,\,
F,G\in(V\otimes V)^{\oplus\ell}
\,.
\end{equation}

\subsection{Trace map and reduction}

Denote by $\mult:\, V\otimes V\to V$ (or, in general, $\mult:\,V^{\otimes n}\to V$)
the multiplication map:
$$
\mult(a\otimes b)=ab\,.
$$
We define the following ``reduction'' map $(V\otimes V)^{\oplus\ell}\to V^{\oplus\ell}$ by
\begin{equation}\label{20140710:eq8}
F
\mapsto 
\mult(F^\sigma)
\,\,\Big(=\big(\mult(F_i^\sigma)\big)_{i=1}^\ell\Big)
\,.
\end{equation}
Furthermore,
let $[V,V]$ be the subspace of $V$ spanned by commutators $ab-ba$, for $a,b\in V$,
and denote by 
\begin{equation}\label{eq:trace}
\tr:\, V\to\quot{V}{[V,V]}
\end{equation}
the canonical quotient map.

The inner product \eqref{20140710:eq2}
induces a well defined inner product $(\cdot\,|\,\cdot):\,V^{\oplus\ell}\times V^{\oplus\ell}\to\quot{V}{[V,V]}$
so that the following diagram is commutative:
$$
\UseTips
\xymatrix{
(V\otimes V)^{\oplus\ell}\times(V\otimes V)^{\oplus\ell}
\ar[d]_{\mult\circ\sigma}
&
\ar[r]^-{(\cdot\,|\,\cdot)}
&
&
V\otimes V
\ar[d]^{\tr\circ\mult} 
\\
V^{\oplus\ell}\times V^{\oplus\ell}
&
\ar@{.>}[r]^-{(\cdot\,|\,\cdot)}
&
&
\quot{V}{[V,V]}
}
$$
It is given explicitly by
\begin{equation}\label{20140710:eq9}
(F|G)
=
\sum_{i=1}^\ell \tr(F_iG_i)
\,\,,\,\,\,\,
F,G\in V^{\oplus\ell}
\,.
\end{equation}
This bilinear form is obviously symmetric.

Furthermore, the action \eqref{20140710:eq3} of the matrix $H\in\Mat_{\ell\times\ell}(V\otimes V)$
induces an action on $V^{\oplus\ell}$ via the commutative diagram
$$
\UseTips
\xymatrix{
(V\otimes V)^{\oplus\ell}
\ar[d]_{\mult\circ\sigma}
&
\ar[r]^-{H}
&
&
(V\otimes V)^{\oplus\ell}
\ar[d]^{\mult\circ\sigma} 
\\
V^{\oplus\ell}
&
\ar@{.>}[r]^-{H}
&
&
V^{\oplus\ell}
}
$$
Explicitly, it is given, in Sweedler's notation, by
\begin{equation}\label{20140710:eq10}
(HF)_i
=
\sum_{j=1}^\ell
H_{ij}^\prime F_j H_{ij}^{\prime\prime}
\,\,,\,\,\,\,
F\in V^{\oplus\ell}
\,.
\end{equation}

\subsection{\texorpdfstring{$V$}{V}-module structures of
\texorpdfstring{$V^{\otimes n}$}{Vxn}}

The space $V^{\otimes n}$ has the following outer $V$-bimodule structure:
\begin{equation}\label{20140604:eq2}
a(b_1\otimes\dots\otimes b_n)c=(ab_1)\otimes b_2\dots\otimes b_{n-1}\otimes(b_nc)
\,.
\end{equation}
More generally, for every $i=0,\dots,n-1$, we define the $i$-th left and right $V$-module structures
of $V^{\otimes n}$ by
\begin{equation}\label{20140605:eq1}
\begin{array}{l}
\displaystyle{
\vphantom{\Big(}
a\star_i (b_1\otimes\dots\otimes b_n)
=
b_1\otimes\dots\otimes b_i\otimes ab_{i+1}\otimes\dots\otimes b_n
\,,} \\
\displaystyle{
\vphantom{\Big(}
(a_1\otimes\dots\otimes a_n)\star_i b
=
a_1\otimes\dots\otimes a_{n-i}b\otimes\dots\otimes a_n
\,.}
\end{array}
\end{equation}
(The index denotes the number of ``jumps''.)
In particular, the bimodule structure \eqref{20140604:eq2}
is given by $aBc=a\star_0 B\star_0 c$.

We use a similar notation for $\otimes$-product
of an element of $V$ and an element of $V^{\otimes n}$:
\begin{equation}\label{20140605:eq1b}
\begin{array}{l}
\displaystyle{
\vphantom{\Big(}
a\otimes_i (b_1\otimes\dots\otimes b_n)
=
b_1\otimes\dots\otimes b_i\otimes a\otimes b_{i+1}\otimes\dots\otimes b_n
\,,} \\
\displaystyle{
\vphantom{\Big(}
(a_1\otimes\dots\otimes a_n)\otimes_i b
=
a_1\otimes\dots\otimes a_{n-i}\otimes b\otimes\dots\otimes a_n
\,.}
\end{array}
\end{equation}
For example, the action \eqref{20140710:eq10}
of $H\in\Mat_{\ell\times\ell}(V\otimes V)$ on $V^{\oplus\ell}$
can be rewritten, using this notation, as
$(HF)_i
=\sum_{j=1}^\ell \mult(H_{ij}\otimes_1 F_j)
\,\big(
=\sum_{j=1}^\ell \mult(F_j\otimes_1 H_{ij})
=\sum_{j=1}^\ell \mult(H_{ij}\star_1 F_j)
=\sum_{j=1}^\ell \mult(F_j\star_1 H_{ij})
\big)$.

Note that, for $f\in V$ and $A,B\in V^{\otimes2}$, we have
\begin{equation}\label{lemma:bullet}
\begin{array}{l}
\displaystyle{
\vphantom{\Big)}
(fA)\bullet B=f(A\bullet B)
\,\,,\,\,\,\,
(Af)\bullet B=(A\bullet B)f
\,,} \\
\displaystyle{
\vphantom{\Big)}
A\bullet(f\star_1B)=f\star_1(A\bullet B)
\,\,,\,\,\,\,
A\bullet(B\star_1f)=(A\bullet B)\star_1f
\,.}
\end{array}
\end{equation}

We also will be using the \emph{multiplication map}
$V^{\otimes m}\times V^{\otimes n}\to V^{\otimes(m+n-1)}$,
given by
\begin{equation}\label{20140604:eq3}
(a_1\otimes\dots\otimes a_m)(b_1\otimes\dots\otimes b_n)
=
a_1\otimes\dots\otimes a_{m-1}\otimes (a_m b_1)\otimes b_2\dots\otimes b_n
\,.
\end{equation}

\subsection{Some further notation}

We define three possible left and right actions of the algebra $V^{\otimes 2}$
on $V^{\otimes3}$, denoted by $\bullet_i$, $i=1,2,3$, as follows
(here the subscript indicates the place where no multiplication occurs)
\begin{equation}\label{20140609:eqc2}
\begin{array}{l}
\displaystyle{
\vphantom{\Big)}
(a\otimes b)\bullet_1(x\otimes y\otimes z)
=
x\otimes ay\otimes zb
\,,\,\,
(x\otimes y\otimes z)\bullet_1(a\otimes b)
=
x\otimes ya\otimes bz
\,,}\\
\displaystyle{
\vphantom{\Big)}
(a\otimes b)\bullet_2(x\otimes y\otimes z)
=
ax\otimes y\otimes zb
\,,\,\,
(x\otimes y\otimes z)\bullet_2(a\otimes b)
=
xa\otimes y\otimes bz
\,,}\\
\displaystyle{
\vphantom{\Big)}
(a\otimes b)\bullet_3(x\otimes y\otimes z)
=
ax\otimes yb\otimes z
\,,\,\,
(x\otimes y\otimes z)\bullet_3(a\otimes b)
=
xa\otimes by\otimes z
\,.}
\end{array}
\end{equation}
\begin{lemma}\label{lemma:bullet-i}
\begin{enumerate}[(a)]
\item
The $\bullet_i$ left (and right) actions of $V^{\otimes2}$ on $V^{\otimes3}$ are 
indeed actions, i.e. they are associative
with respect to the $\bullet$-product of $V^{\otimes2}$:
$$
A\bullet_i(B\bullet_i X)=(A\bullet B)\bullet_i X
\,\text{ and }\,
(X\bullet_i A)\bullet_i B=X\bullet_i(A\bullet B)
\,,
$$
for every $A,B\in V^{\otimes2}$ and $X\in V^{\otimes3}$.
\item
The left $\bullet_i$ and the right $\bullet_j$ actions commute
for every $i,j=1,2,3$:
$$
A\bullet_i(X\bullet_j B)=(A\bullet_i X)\bullet_j B
$$
for every $A,B\in V^{\otimes2}$ and $X\in V^{\otimes3}$.
\item
The $\bullet_1$ and $\bullet_3$ left (resp. right) actions of $V^{\otimes2}$ on $V^{\otimes3}$ 
commute:
$$
A\bullet_1(B\bullet_3 X)=
B\bullet_3(A\bullet_1 X)
\,\text{ and }\,
(X\bullet_1 A)\bullet_3 B=(X\bullet_3 B)\bullet_1 A
\,,
$$
for every $A,B\in V^{\otimes2}$ and $X\in V^{\otimes3}$.
(In general, the $\bullet_i$ and $\bullet_j$ left (resp. right) actions do NOT commute
if $|i-j|=1$.)
\end{enumerate}
\end{lemma}
\begin{proof}
Straightforward.
\end{proof}
\begin{lemma}\label{lemma:bullet2}
For every $A\in V^{\otimes2}$ 
and $X\in V^{\otimes3}$ we have
\begin{equation}\label{eq:bullet-i-sigma}
\begin{array}{lll}
\displaystyle{
\vphantom{\Big(}
(A\bullet_1 X)^\sigma=X^\sigma\bullet_2 A^\sigma
\,\,,}
&
\displaystyle{
\vphantom{\Big(}
(A\bullet_2 X)^\sigma=X^\sigma\bullet_3 A^\sigma
\,\,,}
&
\displaystyle{
\vphantom{\Big(}
(A\bullet_3 X)^\sigma=A\bullet_1 X^\sigma
\,,} \\
\displaystyle{
\vphantom{\Big(}
(A\bullet_1 X)^{\sigma^2}=A\bullet_3 X^{\sigma^2}
\,\,,}
&
\displaystyle{
\vphantom{\Big(}
(A\bullet_2 X)^{\sigma^2}=X^{\sigma^2}\bullet_1 A^\sigma
\,\,,}
&
\displaystyle{
\vphantom{\Big(}
(A\bullet_3 X)^{\sigma^2}=X^{\sigma^2}\bullet_2A^\sigma
\,,} \\
\displaystyle{
\vphantom{\Big(}
(X\bullet_1 A)^\sigma=A^\sigma\bullet_2 X^\sigma
\,\,,}
&
\displaystyle{
\vphantom{\Big(}
(X\bullet_2 A)^\sigma=A^\sigma\bullet_3 X^\sigma
\,\,,}
&
\displaystyle{
\vphantom{\Big(}
(X\bullet_3 A)^\sigma=X^\sigma\bullet_1 A
\,,} \\
\displaystyle{
\vphantom{\Big(}
(X\bullet_1 A)^{\sigma^2}=X^{\sigma^2}\bullet_3A 
\,\,,}
&
\displaystyle{
\vphantom{\Big(}
(X\bullet_2 A)^{\sigma^2}=A^\sigma \bullet_1 X^{\sigma^2}
\,\,,}
&
\displaystyle{
\vphantom{\Big(}
(X\bullet_3 A)^{\sigma^2}=A^\sigma \bullet_2 X^{\sigma^2}
\,.}
\end{array}
\end{equation}
\end{lemma}
\begin{proof}
Straightforward.
\end{proof}

\subsection{\texorpdfstring{$n$}{n}-fold derivations}

\begin{definition}\label{n-derivation}
An $n$-\emph{fold derivation} of $V$ is a linear map
$D:\,V\to V^{\otimes n}$ such that 
\begin{equation}\label{20140604:eq4}
D(ab)=(Da)b+a(Db)
\,\,,\,\,\,\,
a,b\in V\,,
\end{equation}
where, in the RHS, we use the bimodule structure \eqref{20140604:eq2}.
\end{definition}
We can extend the $n$-fold derivation $D:\,V\to V^{\otimes n}$
to a map $D:\,V^{\otimes m}\to V^{\otimes(m+n-1)}$
for every $m\in\mb Z_+$, by
\begin{equation}\label{20140604:eq5}
D(a_1\otimes\dots\otimes a_m)
=
\sum_{i=1}^m
a_1\otimes\dots\otimes(Da_i)\otimes\dots\otimes a_m
\,,
\end{equation}
where, in the RHS, we use the concatenation product in the tensor algebra over $V$.
We can also extend the $n$-fold derivation $D$ to a map $V^{\otimes m}\to V^{\otimes m+n-1}$
by applying $D$ to only the $i$-th factor, and we denote the corresponding map by $D_{(i)}$:
\begin{equation}\label{20140606:eq1a}
D_{(i)}(a_1\otimes\dots\otimes a_m)
=
a_1\otimes\dots\otimes D(a_i)\otimes\dots\otimes a_m
\,.
\end{equation}
We will use the special notation for the derivation on 
the leftmost and the rightmost factors: 
\begin{equation}\label{20140606:eq1b}
D_L:=D_{(1)}
\,\text{ and }\,
D_R:=D_{(m)}
\,.
\end{equation}

The symmetric group $S_n$ acts on $V^{\otimes n}$ in the usual way:
for $\tau\in S_n$, we have
\begin{equation}\label{20140606:eq3a}
(a_1\otimes\dots\otimes a_n)^\tau=a_{\tau^{-1}(1)}\otimes\dots\otimes a_{\tau^{-1}(n)}
\,.
\end{equation}
For every $n$, we shall denote by $\sigma$ the cyclic permutation $(1,\dots,n)$,
so that
\begin{equation}\label{20140606:eq3}
(a_1\otimes\dots\otimes a_n)^\sigma=a_n\otimes a_1\otimes\dots\otimes a_{n-1}
\,.
\end{equation}
For convenience, we write here the formulas for the action of the whole
cyclic group $C_n$ on $A^{\otimes n}$ ($s=1,\dots,n$):
\begin{equation}\label{20140606:eq3b}
(a_1\otimes\dots\otimes a_n)^{\sigma^s}=a_{n+1-s}\otimes\dots\otimes 
a_n\otimes a_1\otimes\dots\otimes a_{n-s}
\,,
\end{equation}
and on an $n$-tuple of indices $(i_1,\dots,i_n)\in\{1,\dots,\ell\}^n$:
\begin{equation}\label{20140606:eq3c}
(i_{\sigma^{-s}(1)},\dots,i_{\sigma^{-s}(n)})
=
(i_{n+1-s},\dots,i_n,i_1,\dots,i_{n-s})
\,.
\end{equation}
\begin{lemma}\label{20140606:lem}
Let $D:\, V\to V^{\otimes n}$ be an $n$-fold derivation of $V$.
For $A\in V^{\otimes m}$, we have
\begin{equation}\label{20140606:eq4}
\begin{array}{l}
\displaystyle{
\vphantom{\Big(}
(D_{(i)}A)^\sigma=D_{(i+1)}(A^\sigma)
\,\,\text{ for } 1\leq i\leq m-1
\,,} \\
\displaystyle{
\vphantom{\Big(}
(D_{(m)}A)^{\sigma^n}=D_{(1)}(A^\sigma)
\,.}
\end{array}
\end{equation}
\end{lemma}
\begin{proof}
Straightforward.
\end{proof}
\begin{lemma}\label{20140604:lem}
If $D:\,V\to V^{\otimes n}$ is an $n$-fold derivation of $V$,
then, for every $A\in V^{\otimes h}$ and $B\in V^{\otimes k}$ ($h,k\in\mb Z_+$),
we have
\begin{equation}\label{20140604:eq6}
D(AB)=(DA)B+A(DB)
\,,
\end{equation}
where, in the RHS, we use the multiplication map \eqref{20140604:eq3}.
\end{lemma}
\begin{proof}
Straightforward.
\end{proof}
Given 
an $m$-fold derivation $D_1:\,V\to V^{\otimes m}$
and an $n$-fold derivation $D_2:\,V\to V^{\otimes n}$,
we can compose them, using their extension \eqref{20140604:eq5},
to get a map 
\begin{equation}\label{20140702:eq3}
D_1\circ D_2:\, V\to V^{\otimes m+n-1}
\,,\,\,\,\,
(D_1\circ D_2)(a)=D_1(D_2 a)
\,.
\end{equation}
\begin{proposition}\label{20140604:prop}
Let $D_1$ and $D_2$ be respectively an $m$-fold and an $n$-fold derivation of $V$.
Then:
\begin{enumerate}[(a)]
\item
The commutator $[D_1,D_2]=D_1\circ D_2-D_2\circ D_1$
is an $m+n-1$-fold derivation of $V$.
\item
$(D_1)_L\circ D_2-(D_2)_R\circ D_1$
is an $m+n-1$-fold derivation of $V$.
\end{enumerate}
\end{proposition}
\begin{proof}
Part (a) follows from Lemma \ref{20140604:lem}.
Using the fact that $D_2$ is an $n$-fold derivation,
and by the definition \eqref{20140606:eq1b} of $(D_1)_L$, we have ($a,b\in V$)
\begin{equation}\label{20140626:eq1b}
\begin{array}{l}
\displaystyle{
\vphantom{\Big)}
\left((D_1)_L\circ D_2\right)(ab)
=(D_1)_L(D_2(a)b)+(D_1)_L(aD_2(b))
} \\
\displaystyle{
\vphantom{\Big)}
=((D_1)_L(D_2 a))b+(D_1 a)(D_2 b)+a((D_1)_L (D_2 b))
\,,}
\end{array}
\end{equation}
where in the second term we are using the multiplication 
in the tensor algebra given by \eqref{20140604:eq3}.
Similarly we get
\begin{equation}\label{20140626:eq2b}
\left((D_2)_R\circ D_1\right)(ab)
=((D_2)_R(D_1a))b+(D_1a)(D_2b)+a((D_2)_R(D_1b))\,.
\end{equation}
Combining equations \eqref{20140626:eq1b} and \eqref{20140626:eq2b} we get part (b).
\end{proof}
\begin{remark}\label{20140604:rem}
Suppose that the algebra $V$ is 
counital (with counity $\epsilon\in V^*$),
so that $\mb F$ is an $A$-bimodule (using the counit).
Then, by Proposition \ref{20140604:prop},
we have a $\mb Z$-graded Lie algebra $\mf g=\bigoplus_{n=-1}^\infty\mf g_n$,
where $\mf g_n$ is the space of $n+1$-fold derivations of $V$.
\end{remark}

\subsection{\texorpdfstring{$n$}{n}-fold brackets}

\begin{definition}\label{n-bracket}
An $n$-\emph{fold bracket} on $V$ is a linear map
$\ldb-,\cdots,-\rdb:A^{\otimes n}\to A^{\otimes n}$
satisfying the following Leibniz rules (using the notation \eqref{20140605:eq1}):
\begin{equation}\label{20140605:eq2}
\ldb a_1,\dots,bc,\dots,a_n\rdb
=
b\star_{i}\ldb a_1,\dots, c,\dots,a_n\rdb
+\ldb a_1,\dots, b,\dots,a_n\rdb \star_{n-i}c
\,,
\end{equation}
for all $i=1\dots,n$
(in the first term of the RHS we let $\star_{n}=\star_0$).
\end{definition}
Note that, if $\ldb-,\cdots,-\rdb$ is an $n$-fold bracket,
then, for every $a_1,\dots,a_{n-1}\in V$,
the map $\ldb a_1,\dots,a_{n-1},-\rdb:\,V\to V^{\otimes n}$
is an $n$-fold derivation of $V$.
Then, according to the notation \eqref{20140606:eq1b},
we let, for $a_1,\dots,a_{n-1},b_1,\dots,b_m\in V$,
\begin{equation}\label{20140606:eq2}
\begin{array}{l}
\displaystyle{
\vphantom{\Big(}
\ldb a_1,\dots,a_{n-1},b_1\otimes\dots\otimes b_m\rdb_L
=
\ldb a_1,\dots,a_{n-1},b_1\rdb\otimes b_2\otimes\dots\otimes b_m
\,,} \\
\displaystyle{
\vphantom{\Big(}
\ldb a_1,\dots,a_{n-1},b_1\otimes\dots\otimes b_m\rdb_R
=
b_1\otimes\dots\otimes b_{m-1}\otimes\ldb a_1,\dots,a_{n-1},b_m\rdb
\,.}
\end{array}
\end{equation}
As a special case of Lemma \ref{20140606:lem},
we get the following result, which will be used later 
\begin{corollary}\label{20140606:cor}
Given a $2$-fold bracket $\ldb-,-\rdb$ on $V$,
we have
\begin{equation}\label{20140606:eq5}
\ldb a,\ldb b,c\rdb\rdb_R
=
\ldb a,\ldb b,c\rdb^\sigma\rdb_L^\sigma
\,.
\end{equation}
\end{corollary}
\begin{proof}
Equation \eqref{20140606:eq5} is a special case of equation \eqref{20140606:eq4}
for $D=\ldb a,-\rdb:\, V\to V^{\otimes2}$.
\end{proof}

The following result will be used later as well.
\begin{lemma}\label{20140609:lem2}
Let $\ldb-,-\rdb$ be a $2$-fold bracket on the associative algebra $V$.
For every $a\in V$ and $B,C\in V^{\otimes2}$, we have
$$
\ldb a,B\bullet C\rdb_L
=B\bullet_2 \ldb a,C\rdb_L
+\ldb a,B\rdb_{L}\bullet_1 C\,,
$$
where in the LHS we use the product \eqref{20140609:eqc1}
while in the RHS we use notation \eqref{20140609:eqc2}.
\end{lemma}
\begin{proof}
Straightforward.
\end{proof}

\section{Double Poisson algebras and non-commutative Hamiltonian ODEs}\label{sec:2}

\subsection{Definition of double Poisson algebras}

\begin{definition}[\cite{VdB08}]\label{20140606:def}
A \emph{double Poisson algebra} is an associative algebra $V$ endowed 
with a $2$-fold bracket $\ldb-,-\rdb:\,V\otimes V\to V\otimes V$
satisfying the following axioms:
\begin{enumerate}[(i)]
\item
skewsymmetry: $\ldb a,b\rdb=-\ldb b,a\rdb^\sigma$,
\item
Jacobi identity:
\begin{equation}\label{eq:jacobi}
\ldb a, \ldb b,c\rdb \rdb_L
+ \ldb b, \ldb c, a\rdb \rdb_L^\sigma
+ \ldb c, \ldb a, b\rdb \rdb_L^{\sigma^2}=0\,.
\end{equation}
(Here we are using notation \eqref{20140606:eq2},
i.e. we let $\ldb a,b\otimes c\rdb_L=\ldb a,b\rdb\otimes c\in V\otimes V\otimes V$.)
\end{enumerate}
\end{definition}
We can write explicitly the Leibniz rules for a double Poisson algebra bracket:
equation \eqref{20140605:eq2} with $i=2$ and $1$ reads, respectively,
\begin{equation}\label{20140605:eq2b}
\ldb a,bc\rdb
=
\ldb a,b\rdb c+b\ldb a,c\rdb
\,,
\end{equation}
and
\begin{equation}\label{20140605:eq2c}
\ldb ab,c\rdb
=
\ldb a,c\rdb\star_1 b+a\star_1\ldb b,c\rdb
\,.
\end{equation}
Note that equation \eqref{20140605:eq2c} follows from \eqref{20140605:eq2b}
and the skewsymmetry axiom.

\begin{remark}\label{20140703:rem}
We can write the Jacobi identity \eqref{eq:jacobi} in an alternative way, 
analogous to the usual ``representation theoretic form'' 
of the Jacobi identity for Lie algebras.
For this we extend the notation \eqref{20140606:eq2} of $\ldb-,-\rdb_L$ and $\ldb-,-\rdb_R$
to the case when we have an element of $V\otimes V$ in the first entry.
We do this in analogy to the two terms in the RHS of \eqref{20140605:eq2c}:
using notation \eqref{20140605:eq1b}, we let
\begin{equation}\label{20140703:eq1}
\ldb a\otimes b,c\rdb_L := \ldb a,c\rdb \otimes_1 b
\,\,,\,\,\,\,
\ldb a\otimes b,c\rdb_R := a\otimes_1\ldb b,c\rdb
\,.
\end{equation}
With this notation, it is not hard to check, using skewsymmetry, 
that the Jacobi identity \eqref{eq:jacobi}
can be equivalently written as follows:
\begin{equation}\label{eq:jacobi-b}
\ldb a, \ldb b,c\rdb \rdb_L
- \ldb b, \ldb a, c\rdb \rdb_R
=
\ldb \ldb a, b\rdb,c \rdb_L
\,.
\end{equation}
\end{remark}

\begin{lemma}[\cite{VdB08}]\label{20140605:lem}
If $\ldb-,-\rdb:\,V\to V^{\otimes2}$
is a $2$-fold bracket satisfying the skewsymmetry axiom,
then the LHS of the Jacobi identity \eqref{eq:jacobi}
\begin{equation}\label{20140606:eq10}
\ldb a,b,c\rdb
:=
\ldb a, \ldb b,c\rdb \rdb_L
+ \ldb b, \ldb c, a\rdb \rdb_L^\sigma
+ \ldb c, \ldb a, b\rdb \rdb_L^{\sigma^2}
\end{equation}
is a $3$-fold bracket on $V$.
\end{lemma}
\begin{proof}
Using the Leibniz rules for the $2$-fold bracket $\ldb-,-\rdb$,
we have
\begin{equation}\label{20140606:eq6a}
\begin{array}{l}
\displaystyle{
\vphantom{\Big(}
\ldb a,\ldb b,cd\rdb\rdb_L
=
\ldb a,\ldb b,c\rdb d\rdb_L
+
\ldb a,c\ldb b,d\rdb\rdb_L
} \\
\displaystyle{
\vphantom{\Big(}
=
\ldb a,\ldb b,c\rdb\rdb_L d
+
\ldb a,c\rdb \ldb b,d\rdb
+
c\ldb a,\ldb b,d\rdb\rdb_L
\,,}
\end{array}
\end{equation}
\begin{equation}\label{20140606:eq6b}
\begin{array}{l}
\displaystyle{
\vphantom{\Big(}
\ldb b,\ldb cd,a\rdb\rdb_L^\sigma
=
\big(
\ldb b,c\star_1\ldb d,a\rdb\rdb_L
+
\ldb b,\ldb c,a\rdb\star_1d\rdb_L
\big)^\sigma
} \\
\displaystyle{
\vphantom{\Big(}
=
\big(
c\star_2\ldb b,\ldb d,a\rdb\rdb_L
+
\ldb b,\ldb c,a\rdb\rdb_L\star_1d
+
\ldb c,a\rdb'
\ldb b,d\rdb
\ldb c,a\rdb''
\big)^\sigma
} \\
\displaystyle{
\vphantom{\Big(}
=
c\ldb b,\ldb d,a\rdb\rdb_L^\sigma
+
\ldb b,\ldb c,a\rdb\rdb_L^\sigma d
+
\ldb c,a\rdb^\sigma
\ldb b,d\rdb
\,.}
\end{array}
\end{equation}
Hereafter we use 
Sweedler's notation:
for $A\in V^{\otimes2}$, we let $A=A'\otimes A''$
(omitting the sign of summation).
Similarly,
\begin{equation}\label{20140606:eq6c}
\begin{array}{l}
\displaystyle{
\vphantom{\Big(}
\ldb cd,\ldb a,b\rdb\rdb_L^{\sigma^2}
=
\big(
c\star_1
\ldb d,\ldb a,b\rdb\rdb_L
+
\ldb c,\ldb a,b\rdb\rdb_L\star_2d
\big)^{\sigma^2}
} \\
\displaystyle{
\vphantom{\Big(}
=
c
\ldb d,\ldb a,b\rdb\rdb_L^{\sigma^2}
+
\ldb c,\ldb a,b\rdb\rdb_L^{\sigma^2}d
\,.}
\end{array}
\end{equation}
Combining equations \eqref{20140606:eq6a}, \eqref{20140606:eq6b} and \eqref{20140606:eq6c},
and using the skewsymmetry assumption on $\ldb-,-\rdb$,
we get
\begin{equation}\label{20140606:eq7}
\ldb a,b,cd\rdb
=
c\ldb a,b,d\rdb+\ldb a,b,c\rdb d\,,
\end{equation}
i.e. the Leibniz rule holds on the third entry.
The Leibniz rules \eqref{20140605:eq2} on the first and second entry 
follow from \eqref{20140606:eq7}
and the following obvious identities ($a,b,c\in V$)
\begin{equation}\label{20140606:eq8}
\ldb a,b,c\rdb
=
\ldb b,c,a\rdb^\sigma
=
\ldb c,a,b\rdb^{\sigma^2}
\,.
\end{equation}
\end{proof}

\subsection{The trace map and connection to Lie algebras}

Let $\ldb-,-\rdb$ be a $2$-fold bracket on the associative algebra $V$.
By composing it with the multiplication map $\mult:\,V\otimes V\to V$,
we obtain the corresponding map $\{-,-\}:\,V\otimes V\to V$:
\begin{equation}\label{20140606:eq8b}
\{a,b\}=\mult\ldb a,b\rdb
\,.
\end{equation}
It is immediate to check that this bracket
satisfies the Leibniz rule on the second entry:
\begin{equation}\label{20140606:eq9}
\{a,bc\}=b\{a,c\}+\{a,b\}c\,.
\end{equation}
(While, in general, it does not satisfies the Leibniz rule on the first entry.)
Hence, for every $a\in V$, the map $\{a,-\}:\,V\to V$ is a ($1$-fold) derivation,
and we extend it to a map $\{a,-\}:\,V^{\otimes n}\to V^{\otimes n}$, for every $n\in\mb Z_+$,
via \eqref{20140604:eq5}.
\begin{lemma}[\cite{VdB08}]\label{prop:1}
If $\ldb-,-\rdb$ is a skewsymmetric $2$-fold bracket on $V$, 
then the following identity holds in $V^{\otimes2}$ ($a,b,c\in V$)
\begin{equation}\label{20131205:eq3}
\{a,\ldb b,c\rdb\}-\ldb\{a,b\},c\rdb-\ldb b,\{a,c\}\rdb
=(\mult\otimes1)\ldb a,b,c\rdb-(1\otimes \mult)\ldb b,a,c\rdb\,,
\end{equation}
where $\ldb a,b,c\rdb$ is as in \eqref{20140606:eq10}.
\end{lemma}
\begin{proof}
Let us compute the three terms appearing in the LHS of equation \eqref{20131205:eq3}.
Using the Leibniz rules we get, after a straightforward computation, 
\begin{equation}\label{20140226:eq1}
\begin{array}{l}
\displaystyle{
\vphantom{\Big(}
\{a,\ldb b,c\rdb \}
=(\mult\otimes 1)\ldb a,\ldb b,c\rdb \rdb_L+(1\otimes\mult) \ldb a,\ldb b,c\rdb \rdb_R
} \\
\displaystyle{
\vphantom{\Big(}
=(\mult\otimes 1)\ldb a,\ldb b,c\rdb \rdb_L-(1\otimes\mult) \ldb a,\ldb c,b\rdb \rdb_L^\sigma
\,.}
\end{array}
\end{equation}
For the last equality we used Corollary \ref{20140606:cor} and the skewsymmetry assumption.
With similar computations, we get
\begin{equation}\label{20140226:eq3}
\ldb b,\{a,c\}\rdb 
=(1\otimes \mult) \ldb b, \ldb a,c\rdb\rdb_L
-(\mult\otimes 1) \ldb b,\ldb c,a\rdb \rdb_L^\sigma
\,,
\end{equation}
and 
\begin{equation}\label{20140226:eq2}
\ldb \{a,b\},c\rdb
=-\ldb c,\{a,b\}\rdb^\sigma
=-(\mult\otimes 1)\ldb c,\ldb a,b\rdb \rdb_L^{\sigma^2}
+(1\otimes \mult)\ldb c,\ldb b,a\rdb \rdb_L^{\sigma^2}\,.
\end{equation}
Combining equations \eqref{20140226:eq1}, \eqref{20140226:eq2}
and \eqref{20140226:eq3}, and using the definition \eqref{20140606:eq10}
of $\ldb a,b,c\rdb$, equation \eqref{20131205:eq3} follows,
thus concluding the proof.
\end{proof}

\begin{proposition}[\cite{VdB08}]\label{20131205:prop1}
\begin{enumerate}[(a)]
\item
If $\ldb-,-\rdb$ is a $2$-fold bracket on $V$,
then the associated bracket $\{-,-\}$ defined in \eqref{20140606:eq8b}
induces well defined maps 
(which we denote by the same symbol)
$\{-,-\}:\,\quot{V}{[V,V]}\otimes V\to V$, given by
(recalling \eqref{eq:trace})
\begin{equation}\label{20131205:eq1}
\{\tr(a),b\}:=\{a,b\}
\end{equation}
and $\{-,-\}:\,\quot{V}{[V,V]}\otimes\quot{V}{[V,V]}\to\quot{V}{[V,V]}$, given by
\begin{equation}\label{20131205:eq2}
\{\tr(a),\tr(b)\}:=\tr\{a,b\}
\,.
\end{equation}
\item 
If $V$ is a double Poisson algebra, then $\quot{V}{[V,V]}$ is a Lie algebra
with the bracket defined by \eqref{20131205:eq2},
and \eqref{20131205:eq1} defines a representation
of this Lie algebra by derivations of $V$.
\end{enumerate}
\end{proposition}
\begin{proof}
First, we have, using Sweedler's notation,
\begin{equation}\label{20131205:toprove1}
\{ab,c\}
=
\ldb a,c\rdb'b\ldb a,c\rdb''
+
\ldb b,c\rdb'a\ldb b,c\rdb''
=
\{ ba,c\}
\,.
\end{equation}
Moreover, by \eqref{20140606:eq9}, we have
\begin{equation}\label{20131205:toprove1b}
\{a,bc\}
-
\{a,cb\}
=
[\{a,b\},c]+[b,\{a,c\}]\,\in[V,V]
\,.
\end{equation}
Equations \eqref{20131205:toprove1} and \eqref{20131205:toprove1b}
imply that both brackets \eqref{20131205:eq1} and \eqref{20131205:eq2}
are well defined, proving (a).

The skewsymmetry for the bracket $\{-,-\}$ on $\quot{V}{[V,V]}$
follows immediately by the skewsymmetry axiom for the $2$-fold bracket $\ldb-,-\rdb$.
By the Jacobi identity axiom (in Definition \ref{20140606:def}),
the RHS of equation \eqref{20131205:eq3} is identically zero.
Applying the multiplication map $\mult$ to its LHS, we obtain the Jacobi identity
for the bracket $\{-,-\}$ of $\quot{V}{[V,V]}$,
and, at the same time, the claim that \eqref{20131205:eq1}
defines a representation of the Lie algebra $\quot{V}{[V,V]}$ on $V$.
Finally, this Lie algebra action is by derivations of the associative algebra $V$,
thanks to equation \eqref{20140606:eq9}.
\end{proof}

\subsection{
Double Poisson brackets on an algebra of (non-commutative) ordinary differential functions
}
\label{sec:3.3}

Consider the algebra
of non-commutative polynomials
$R_\ell=\mb F\langle x_1,\dots,x_\ell\rangle$,
i.e. the free unital associative algebra in the indeterminates $x_1,\dots,x_\ell$.
Let $\frac{\partial}{\partial x_i}:\,R_\ell\to R_\ell\otimes R_\ell$, $i=1,\dots,\ell$,
be the linear maps defined on monomials by
\begin{equation}\label{20140604:eq1}
\frac{\partial}{\partial x_i} (x_{i_1}\dots x_{i_s})
=
\sum_{k=1}^s
\delta_{i_k,i}
\,
x_{i_1}\dots x_{i_{k-1}}
\otimes
x_{i_{k+1}}\dots x_{i_s}
\,.
\end{equation}
The partial derivatives $\frac{\partial}{\partial x_i}$ are clearly $2$-fold derivations of $R_\ell$.
Furthermore, 
by Proposition \ref{20140604:prop},
$\Big[\frac{\partial}{\partial x_i},\frac{\partial}{\partial x_j}\Big]$
is a 3-fold derivation.
Since it is zero on generators $x_k$, $k=1,\dots,\ell$,
we obtain
\begin{equation}\label{20140609:eq4}
\frac{\partial}{\partial x_i}\circ\frac{\partial}{\partial x_j}
=
\frac{\partial}{\partial x_j}\circ\frac{\partial}{\partial x_i}
\,.
\end{equation}
\begin{lemma}\label{20140609:lem1}
We have the following equality of maps $R_\ell\to R_\ell^{\otimes3}$
($i,j\in\mb Z_+$):
\begin{equation}\label{20140702:eq2}
\left(\frac{\partial}{\partial x_i}\right)_L\circ\frac{\partial}{\partial x_j}
=\left(\frac{\partial}{\partial x_j}\right)_R\circ\frac{\partial}{\partial x_i}
\,.
\end{equation}
\end{lemma}
\begin{proof}
By Proposition \ref{20140604:prop}(b) it follows that
$$
\left(\frac{\partial}{\partial x_i}\right)_L\circ\frac{\partial}{\partial x_j}
-\left(\frac{\partial}{\partial x_j}\right)_R\circ\frac{\partial}{\partial x_i}
$$
is a $3$-fold derivation. The claim follows since this $3$-fold
derivation is obviously zero on the generators $x_k$, $k=1,\dots,\ell$.
\end{proof}
Note that equation \eqref{20140702:eq2} implies \eqref{20140609:eq4}, 
by the definition \eqref{20140702:eq3} of the composition of $2$-fold derivations.
Hence, we will say that the $2$-fold derivations $\frac{\partial}{\partial x_i}$
and $\frac{\partial}{\partial x_j}$ \emph{strongly commute}.
\begin{definition}\label{def:algebra}
An \emph{algebra of ordinary differential functions}
is a unital associative algebra $V$,
endowed with $\ell$ 
strongly commuting $2$-fold derivations $\frac{\partial}{\partial x_i}:\, V\to V\otimes V$,
i.e. satisfying \eqref{20140702:eq2}.
\end{definition}
An example is the algebra $R_\ell=\mb F\langle x_1,\dots,x_\ell\rangle$
of non-commutative polynomials in the variables $x_1,\dots,x_\ell$,
endowed with the $2$-fold derivations \eqref{20140604:eq1}.
Other examples can be obtained as extensions of $R_\ell$
by localization by one or more non-zero elements.

\begin{theorem}\label{thm:master-finite}
\begin{enumerate}[(a)]
\item
Any $2$-fold bracket on $R_\ell$ has the following form ($f,g\in V$)
\begin{equation}\label{master-finite}
\ldb f,g\rdb
=
\sum_{i,j=1}^\ell
\frac{\partial g}{\partial x_j}
\bullet 
\ldb x_i,x_j\rdb
\bullet
\Big(\frac{\partial f}{\partial x_i}\Big)^\sigma
\,,
\end{equation}
where $\bullet$ is as in \eqref{20140609:eqc1}.
\item
Let $V$ be an algebra of ordinary differential functions,
and let $H$ be an $\ell\times\ell$
matrix with entries in $V\otimes V$.
Denoting its entries by $H_{ji}=\ldb x_i,x_j\rdb$, $i,j=1,\dots,\ell$,
formula \eqref{master-finite} 
defines a $2$-fold bracket on $V$.
\item
Equation \eqref{master-finite} defines a structure of double Poisson algebra on $V$
if and only if the skewsymmetry axiom (i) and the Jacobi identity (ii) hold on the $x_i$'s,
namely, for all $i,j,k=1,\dots,\ell$, we have
\begin{equation}\label{skew-gen}
\ldb x_i,x_j\rdb=-\ldb x_j,x_i\rdb^\sigma
\,\,
\Big(\text{equivalently, recalling \eqref{20140710:eq6}, }\,
H^\dagger=-H
\Big)
\,,
\end{equation}
and
\begin{equation}\label{jacobi-gen}
\ldb x_i,\ldb x_j,x_k\rdb\rdb_L
+
\ldb x_j,\ldb x_k,x_i\rdb\rdb_L^\sigma
+
\ldb x_k,\ldb x_i,x_j\rdb\rdb_L^{\sigma^2}
=0
\,,
\end{equation}
where each term is understood, in terms of the matrix $H$, via \eqref{master-finite},
for example
\begin{equation}\label{20140709:eq4}
\ldb x_i,\ldb x_j,x_k\rdb\rdb_L=\sum_{h=1}^\ell
\Big(\Big(\frac{\partial}{\partial x_h}\Big)_LH_{kj}\Big)\bullet_3 H_{hi}
\,.
\end{equation}
\end{enumerate}
\end{theorem}
\begin{proof}
We prove that the RHS of
\eqref{master-finite} satisfies the Leibniz rules \eqref{20140605:eq2b} and \eqref{20140605:eq2c}.
Using the fact that partial derivatives are $2$-fold derivations, the
LHS of \eqref{20140605:eq2b} is
\begin{equation}\label{20140609:eq5}
\sum_{i,j=1}^\ell
\Big(\frac{\partial g}{\partial x_j}h\Big)\bullet\ldb x_i,x_j\rdb\bullet
\left(\frac{\partial f}{\partial x_i}\right)^\sigma
+
\sum_{i,j=1}^\ell
\Big(g\frac{\partial h}{\partial x_j}\Big)\bullet\ldb x_i,x_j\rdb\bullet
\left(\frac{\partial f}{\partial x_i}\right)^\sigma
\,.
\end{equation}
By \eqref{lemma:bullet} we have that
\eqref{20140609:eq5} is the same as RHS of \eqref{20140605:eq2b}.
The Leibniz rule \eqref{20140605:eq2c} is proved in the same way using 
again \eqref{lemma:bullet}.
Since any $2$-fold bracket is extended uniquely 
from the generators to $R_\ell$ via the Leibniz rules,
this concludes the proof of (a) and (b).

The fact that the bracket given by equation \eqref{master-finite} satisfies 
the skewsymmetry axiom is immediate from \eqref{20140710:eq4}
and the assumption that $\ldb x_i,x_j\rdb=-\ldb x_j,x_i\rdb^\sigma$
for every $i,j=1,\dots,\ell$.

To conclude the proof of part (c) we are left to prove Jacobi identity
\eqref{eq:jacobi}.
First, for $f\in V$ and $G\in V^{\otimes2}$ we have, by \eqref{master-finite}
and the definition of the $\bullet_i$ actions,
\begin{equation}\label{master-finite-L}
\ldb f,G\rdb_L
=
\sum_{i,j=1}^\ell
\Big(\Big(\frac{\partial}{\partial x_j}\Big)_LG\Big)
\bullet_3
\Big(
\ldb x_i,x_j\rdb
\bullet
\Big(\frac{\partial f}{\partial x_i}\Big)^\sigma
\Big)
\,.
\end{equation}
Then, by equation \eqref{master-finite},
Lemma \ref{20140609:lem2},
and equation \eqref{master-finite-L}, we get
\begin{align}
&\ldb f,\ldb g,h\rdb\rdb_L\notag\\
\label{eq:jac1}
&=\sum_{i,j,h,k=1}^\ell
\left(\frac{\partial}{\partial x_h}\right)_L\left(\frac{\partial h}{\partial x_k}\right)
\bullet_3
\left(\ldb x_i,x_h\rdb\bullet\left(\frac{\partial f}{\partial x_i}\right)^\sigma\right)
\bullet_1
\left(\ldb x_j,x_k\rdb\bullet\left(\frac{\partial g}{\partial x_j}\right)^{\sigma}\right)
\\
\label{eq:jac2}
&+\sum_{i,j,k=1}^\ell
\frac{\partial h}{\partial x_k}\bullet_2
\ldb x_i,\ldb x_j,x_k\rdb\rdb_L\bullet_3
\left(\frac{\partial f}{\partial x_i}\right)^\sigma\bullet_1
\left(\frac{\partial g}{\partial x_j}\right)^\sigma
\\
\label{eq:jac3}
&+\sum_{i,j,h,k=1}^\ell
\left(\frac{\partial h}{\partial x_k}\bullet\ldb x_j,x_k\rdb\right)
\bullet_2
\left(\frac{\partial}{\partial x_h}\right)_L\left(\frac{\partial g}{\partial x_j}\right)^\sigma
\bullet_3
\left(\ldb x_i,x_h\rdb\bullet\left(\frac{\partial f}{\partial x_i}\right)^{\sigma}\right)
\,.
\end{align}
Note that in the RHS above we avoided using unnecessary parentheses,
in view of Lemma \ref{lemma:bullet-i}.
From the above equation, permuting $f,g$ and $h$, 
applying the permutation $\sigma$,
and using Lemma \ref{lemma:bullet2} several times, we get
\begin{align}
&\ldb g,\ldb h,f\rdb\rdb_L^\sigma\notag\\
\label{eq:jac4}
&=\sum_{i,j,h,k=1}^\ell
\left(\frac{\partial h}{\partial x_k}\bullet\ldb x_k,x_i\rdb^\sigma\right)
\bullet_2
\left(
\left(\frac{\partial}{\partial x_h}\right)_L\left(\frac{\partial f}{\partial x_i}\right)
\right)^\sigma
\bullet_1
\left(\ldb x_j,x_h\rdb\bullet\left(\frac{\partial g}{\partial x_j}\right)^\sigma\right)
\\
\label{eq:jac5}
&+\sum_{i,j,k=1}^\ell
\frac{\partial h}{\partial x_k}
\bullet_2
\ldb x_j,\ldb x_k,x_i\rdb\rdb_L^\sigma
\bullet_1
\left(\frac{\partial g}{\partial x_j}\right)^\sigma
\bullet_3
\left(\frac{\partial f}{\partial x_i}\right)^\sigma
\\
\label{eq:jac6}
&+\!\!\sum_{i,j,h,k=1}^\ell
\!\!\!
\left(\left(\frac{\partial}{\partial x_h}\right)_L
\!\!
\left(\frac{\partial h}{\partial x_k}\right)^\sigma\right)^\sigma
\!\bullet_1\!
\left(\ldb x_j,x_h\rdb
\!\bullet\!
\left(\frac{\partial g}{\partial x_j}\right)^{\sigma}\right)
\!\bullet_3\!
\left(\ldb x_k,x_i\rdb^\sigma
\!\bullet\!
\left(\frac{\partial f}{\partial x_i}\right)^\sigma\right)
.
\end{align}
and, similarly,
\begin{align}
&\ldb h,\ldb f,g\rdb\rdb_L^{\sigma^2}\notag\\
\label{eq:jac7}
&=\sum_{i,j,h,k=1}^\ell
\left(\frac{\partial h}{\partial x_k}\bullet\ldb x_k,x_h\rdb^\sigma\right)
\bullet_2
\left(\left(\frac{\partial}{\partial x_h}\right)_L\left(\frac{\partial g}{\partial x_j}\right)\right)^{\sigma^2}
\bullet_3
\left(\ldb x_i,x_j\rdb\bullet\left(\frac{\partial f}{\partial x_i}\right)^{\sigma}\right)
\\
\label{eq:jac8}
&+\sum_{i,j,k=1}^\ell
\frac{\partial h}{\partial x_k}
\bullet_2
\ldb x_k,\ldb x_i,x_j\rdb\rdb_L^{\sigma^2}
\bullet_3
\left(\frac{\partial f}{\partial x_i}\right)^\sigma
\bullet_1
\left(\frac{\partial g}{\partial x_j}\right)^\sigma
\\
\label{eq:jac9}
&+\!\!\sum_{i,j,h,k=1}^\ell
\!\!
\left(\frac{\partial h}{\partial x_k}
\bullet
\ldb x_k,x_h\rdb^\sigma\right)
\!\bullet_2\!
\left(\left(\frac{\partial}{\partial x_h}\right)_L
\!
\left(\frac{\partial f}{\partial x_i}\right)^\sigma\right)^{\sigma^2}
\!\bullet_1\!
\left(\ldb x_i,x_j\rdb^\sigma\bullet\left(\frac{\partial g}{\partial x_j}\right)^\sigma\right)
\,.
\end{align}
Since, by assumption, the Jacobi identity \eqref{eq:jacobi} holds on
generators, we have that \eqref{eq:jac2}, \eqref{eq:jac5} and
\eqref{eq:jac8} sum up to zero.
Note that here we are using the fact that the $\bullet_1$ and $\bullet_3$ actions
commute, by Lemma \ref{lemma:bullet-i}(c).
Furthermore, 
by Lemma \ref{20140606:lem} and assumption \eqref{20140702:eq2}, we have 
$$
\left(
\left(\frac{\partial}{\partial x_k}\right)_L\left(\frac{\partial h}{\partial x_h}\right)^\sigma\right)^\sigma
=\left(\frac{\partial}{\partial x_k}\right)_R\left(\frac{\partial h}{\partial x_h}\right)
=\left(\frac{\partial}{\partial x_h}\right)_L\left(\frac{\partial h}{\partial x_k}\right)
\,.
$$
It follows that
\eqref{eq:jac1} and \eqref{eq:jac6}
are opposite to each other by the skewsymmetry on generators.
Similarly, 
\eqref{eq:jac3} and \eqref{eq:jac7}
(respectively \eqref{eq:jac4} and \eqref{eq:jac9}) are opposite to each other.
This concludes the proof.
\end{proof}
\begin{definition}\label{poisson-structure}
The matrix $H=\big( H_{ij}\big)_{i,j=1}^\ell\in\Mat_{\ell\times\ell}(V\otimes V)$
satisfying \eqref{skew-gen} and \eqref{jacobi-gen}
is called a \emph{Poisson structure} on the algebra of ordinary differential functions $V$.
\end{definition}
\begin{remark}\label{20140710:rem}
Formula \eqref{master-finite} for the $2$-fold bracket of $f,g\in V$
can be written in a ``traditional'' form as
\begin{equation}\label{master-tradition}
\ldb f,g\rdb
=
(\nabla g | H \nabla f)
\,,
\end{equation}
where
$\nabla f=\Big(\frac{\partial f}{\partial x_i}\Big)_{i=1}^\ell\in(V\otimes V)^{\oplus\ell}$,
$H=(H_{ji})_{i,j=1}^\ell\in\Mat_{\ell\times\ell}(V\otimes V)$
is the matrix with entries $H_{ji}=\ldb x_i,x_j\rdb$,
its action on $(V\otimes V)^{\oplus\ell}$ is defined by \eqref{20140710:eq3},
and $(\cdot\,|\,\cdot)$ is the inner product \eqref{20140710:eq2}.
\end{remark}
\begin{remark}\label{20140627:rem}
One can define a \emph{double Lie algebra}
as a vector space $\mf g$ endowed with a double bracket
$\ldb-,-\rdb:\,\mf g\times\mf g\to\mf g\otimes\mf g$
satisfying the skewsymmetry and Jacobi identity axioms
from Definition \ref{20140606:def}.
Then, by Theorem \ref{thm:master-finite}
the tensor algebra $\mc T(\mf g)$ is a double Poisson algebra,
with the $2$-fold bracket given by \eqref{master-finite}.
\end{remark}

\subsection{
Evolution ODE and Hamiltonian ODE
}

Let $V$ be an algebra of ordinary differential functions.
An \emph{evolution ODE} over $V$
is an equation of the form
\begin{equation}\label{evol-eq}
\frac{d x_i}{dt}=P_i\,\in V
\,\,,\,\,\,\,
i=1,\dots,\ell
\,.
\end{equation}
By the chain rule, a function $f\in V$ evolves, in virtue of \eqref{evol-eq},
according to the equation
\begin{equation}\label{evol-eq2}
\frac{d f}{dt}
=
\sum_{i=1}^\ell
\mult\Big(
P_i\otimes_1
\frac{\partial f}{\partial x_i}
\Big)
\,.
\end{equation}
It is immediate to check that the time derivative $\frac{d}{dt}:\, V\to V$
as given by \eqref{evol-eq2} is a derivation of $V$.
In particular, it induces a well defined map on $\quot{V}{[V,V]}$.
An \emph{integral of motion} for equation \eqref{evol-eq}
is an element $\tr(h)\in\quot{V}{[V,V]}$
constant in time, i.e. such that
\begin{equation}\label{integral-of-motion}
\frac{d \tr h}{dt}
=
\tr \sum_{i=1}^\ell
\mult\Big(
P_i\otimes_1\frac{\partial h}{\partial x_i}
\Big)
=0
\,.
\end{equation}

Given $P\in V^{\ell}$, we define the corresponding 
\emph{evolutionary vector field} $X_P$ as the derivation of $V$
given by the RHS of \eqref{evol-eq2}:
\begin{equation}\label{evol-vect-field}
X_P(f)
=
\sum_{i=1}^\ell
\mult\Big(
P_i\otimes_1\frac{\partial f}{\partial x_i}
\Big)
\,.
\end{equation}
Evolutionary vector fields form the Lie algebra $\Vect(V)$,
with commutator given by the usual formula:
\begin{equation}\label{commut-evf}
[X_P,X_Q]=X_{[P,Q]}
\,\,,\,\,\text{ where }\,
[P,Q]_i
=
X_P(Q_i)-X_Q(P_i)
\,\,,\,\,\,\, i=1,\dots,\ell
\,.
\end{equation}
This can be easily checked using condition \eqref{20140702:eq2} on $V$.

Equation \eqref{evol-eq} is called \emph{compatible} with another evolution ODE 
$\frac{dx_i}{d\tau}=Q_i$, $i=1,\dots,\ell$,
if the corresponding evolutionary vector fields commute: $[X_P,X_Q]=0$.

A \emph{Hamiltonian ODE} on an arbitrary double Poisson algebra $V$,
associated to the \emph{Hamiltonian function} $\tr h\in \quot{V}{[V,V]}$ is
an equation of the form
\begin{equation}\label{ord-ham-eq}
\frac{du}{dt}
=
\{\tr h,u\}
\,,\,\,
u\in V\,.
\end{equation}
An \emph{integral of motion} for the Hamiltonian ODE \eqref{ord-ham-eq}
is an element $\tr f\in \quot{V}{[V,V]}$ such that $\{\tr h,\tr f\}=0$.
In this case, by Proposition \ref{20131205:prop1}, 
the derivations $X^h=\{\tr h,-\}$ and $X^f=\{\tr f,-\}$ 
(called Hamiltonian vector fields) commute,
and equations $\frac{du}{dt}=\{\tr f,u\}$ and \eqref{ord-ham-eq} are \emph{compatible}.

In the special case when $V$ is an algebra of ordinary differential functions
and the $2$-fold bracket $\ldb-,-\rdb$ is given by 
a Poisson structure $H$ on $V$ via \eqref{master-finite},
where $\ldb x_i,x_j\rdb=H_{ji}$,
the Hamiltonian ODE \eqref{ord-ham-eq} becomes the following evolution ODE
\begin{equation}\label{20140702:eq1}
\frac{dx_i}{dt}
=
\mult\sum_{j=1}^\ell
H_{ij}\bullet \Big(\frac{\partial h}{\partial x_j}\Big)^\sigma
\,.
\end{equation}
Furthermore the Lie algebra bracket $\{-,-\}$ on $\quot{V}{[V,V]}$
associated to the double Poisson bracket \eqref{master-finite},
defined by Proposition \ref{20131205:prop1},
is:
\begin{equation}\label{20140709:eq1}
\{\tr(f),\tr(g)\}
=\tr\Big(\mult\sum_{i,j=1}^\ell
\frac{\partial g}{\partial x_j}
\bullet
H_{ji}\bullet
\left(\frac{\partial f}{\partial x_i}\right)^\sigma\Big)
\,\,,\,\,\,\,
f,g\in V
\,.
\end{equation}
In this case the notions of compatibility and of integrals of motion are consistent with those
for general evolution ODEs, due to Theorem \ref{thm:master-finite}(c).
\begin{remark}\label{20140709:rem1}
Recalling the definition \eqref{20140710:eq10} of the action of the matrix
$H\in\Mat_{\ell\times\ell}(V\otimes V)$ on $V^{\oplus\ell}$
formula \eqref{20140702:eq1} can be written in the ``traditional'' form
\begin{equation}\label{20140709:eq7}
\frac{dx}{dt}=Hdh
\,,
\end{equation}
and similarly formula \eqref{20140709:eq1} can be written as 
(see \eqref{master-tradition})
\begin{equation}\label{20140709:eq8}
\{\tr(f),\tr(g)\}
=
(dg|Hdf)
\,,
\end{equation}
where $(\cdot\,|\,\cdot)$ is the inner product \eqref{20140710:eq9} on $V^{\oplus\ell}$,
and $df=\mult(\nabla f)^\sigma$ (cf. \eqref{20140710:eq8}).
The latter notation is compatible with the theory of the reduced de Rham complex, 
cf. \eqref{20140626:eq4a}.
\end{remark}

\subsection{
Lenard-Magri scheme for Hamiltonian ODE
}\label{sec:lenard-magri}

Let $H\in\Mat_{\ell\times\ell}(V\otimes V)$. We define a bilinear form associated to $H$,
$\langle-,-\rangle_H:V^{\oplus\ell}\otimes V^{\oplus\ell}\to\quot{V}{[V,V]}$ as follows
\begin{equation}\label{20140709:eq5}
\langle F, G\rangle_H=(F|HG)
\,\,,\,\,\,\,
F,G\in V^{\oplus\ell}
\,.
\end{equation}
Here we are using notation \eqref{20140710:eq9} and \eqref{20140710:eq10}.
Recalling \eqref{20140710:eq7},
we obtain that the form \eqref{20140709:eq5} is symmetric (resp. skewsymmetric)
if $H$ is selfadjoint (resp. skewadjoint)
(see \eqref{20140710:eq6}).
\begin{lemma}\label{20140709:lem2}
Let $H_0$, $H_1\in\Mat_{\ell\times\ell}(V\otimes V)$ be skewadjoint, and let
$\{F_n\}_{n=0}^N\subset V^{\oplus\ell}$ be such that
$$
H_0F^{n+1}=H_1F^n\,\,\,\,
\text{for }n=0,\dots,N-1
\,.
$$
Then, the following orthogonality relations hold:
$$
\langle F^n, F^m\rangle_{H_0}
=\langle F^n, F^m\rangle_{H_1}=0
\,\,\,\,
\text{for }n,m=0,\dots,N
\,.
$$
\end{lemma}
\begin{proof}
Same proof as in \cite{Mag78} or Lemma 2.6 in \cite{BDSK09}.
\end{proof}
\begin{corollary}\label{cor:lenard_magri}
Let $\ldb-,-\rdb_0$ and $\ldb-,-\rdb_1$
be two skewsymmetric $2$-fold brackets on an algebra of ordinary differential functions $V$
given by the equation \eqref{master-finite} for $H_0,H_1\in\Mat_{\ell\times\ell}(V\otimes V)$
respectively. 
Let $\{h_n\}_{n=0}^N\subset V$ be such that 
\begin{equation}\label{eq:len_rec-a}
\{\tr(h_n), u\}_0=\{ \tr(h_{n+1}), u\}_1
\quad\text{for }u\in V\,,n=0,\dots,N-1\,.
\end{equation}
Then we have
$$
\{\tr(h_n),\tr(h_m)\}_0=
\{\tr(h_n),\tr(h_m)\}_1=0
\quad\text{for }n,m=0,\dots,N
\,.
$$ 
(See \eqref{20140606:eq8b}, \eqref{20131205:eq1} and \eqref{20131205:eq2} for notation).
\end{corollary}
\begin{proof}
It follows from Lemma \ref{20140709:lem2} and equation \eqref{20140709:eq8}.
\end{proof}
\begin{remark}\label{20140711:rem1}
The analogues of Propositions 2.9 and 2.10 in \cite{BDSK09} still hold
in the present setup with the same proof.
\end{remark}

\begin{theorem}
Let $V=R_\ell$ be the algebra of non-commutative polyomials in $\ell$ variables.
Let $H_0,H_1$ be compatible Poisson structures on  $V$,
and assume that $K$ is a non-degenerate matrix.
Let $h_0,h_1\in V$ and $F\in V^\ell$ be such that
\begin{equation}\label{eq:thmfond}
H_1 dh_0=H_0 dh_1
\,\,,\,\,\,\,
H_1 dh_1=H_0 F
\,,
\end{equation}
where, as before, $dh=m(\nabla h)^\sigma$, and we are using the notation \eqref{20140710:eq10}.
Then, there exists $h_2\in V$ such that $F=dh_2$.
\end{theorem}
\begin{proof}
First, one checks that the Jacobian of $F_2$ (defined after equation \eqref{eq:0.22}),
is self-adjoint, i.e. $J_{F}=J_{F}^\dagger$.
The proof of this claim is long but straightforward,
and the argument is the same as in the commutative case, cf. e.g. \cite[Lemma 7.25]{Olv93}. 
The proof uses non-degeneracy of the bilinear form (\ref{20140710:eq9}), 
which follows from the well-known property of the algebra $V=R_\ell$
that for $a\in V$, $\tr\,aV=0$ implies $a=0$.
Then, one uses the exactness of the reduced de Rham complex,
proved in Theorem \ref{20140623:thm} below,
to conclude that the $1$-form $F$ is exact.
\end{proof}

Using this theorem and Remark \ref{20140711:rem1}, one shows that the Lenard-Magri scheme works for $V=R_\ell$, i.e. the Lenard-Magri recursive relation
in Corollary \ref{cor:lenard_magri} for $N=1$ can be extended to $N=\infty$,
provided that the orthogonality condition, similar to that in \cite[Proposition 2.9]{BDSK09}, holds.

\subsection{
Examples
}\label{sec:2.6}

It is claimed in \cite{VdB08} that the only non-zero double Poisson brackets on $R_1=\mb F[x]$
(up to automorphisms of $R_1$) are given by
$$
\ldb x,x\rdb=x\otimes1-1\otimes x
\,\,\,\,\text{and}\,\,\,\,
\ldb x,x\rdb=x^2\otimes x-x\otimes x^2\,.
$$
The corresponding Lie algebra brackets on $\quot{R_1}{[R_1,R_1]}$ given by \eqref{20140709:eq1}
are trivial. Hence, we do not get non-trivial Hamiltonian equations.

The simplest double Poisson bracket on $R_2=\mb F\langle x,y\rangle$ has the following form:
\begin{equation}\label{20140708:eq1}
\ldb x,x\rdb=p\otimes q-q\otimes p
\,\text{ where }p,q\in R_2\,,
\,\,\,\,
\ldb x,y\rdb=\ldb y,y\rdb=0\,.
\end{equation}
The double Poisson bracket given by \eqref{master-finite} and
\eqref{20140708:eq1} is obviously skewsymmetric
and it satisfies Jacobi identity if and only if equation \eqref{jacobi-gen} holds for
$x_i=x_j=x_k=x$ (see Theorem \ref{thm:master-finite}). Using equation \eqref{20140709:eq4},
this condition can be rewritten as follows
(using Sweedler's notation)
\begin{equation}\label{20140708:eq2}
\begin{array}{l}
\displaystyle{
\left(\frac{\partial p}{\partial x}\right)^\prime p
\otimes q\left(\frac{\partial p}{\partial x}\right)^{\prime\prime}
\otimes q
-\left(\frac{\partial p}{\partial x}\right)^\prime q
\otimes p\left(\frac{\partial p}{\partial x}\right)^{\prime\prime}
\otimes q
-\left(\frac{\partial q}{\partial x}\right)^\prime p
\otimes q\left(\frac{\partial q}{\partial x}\right)^{\prime\prime}
\otimes p
}\\
\displaystyle{
+\left(\frac{\partial q}{\partial x}\right)^\prime q
\otimes p\left(\frac{\partial q}{\partial x}\right)^{\prime\prime}
\otimes p
+\text{cyclic permutations}=0\,.
}
\end{array}
\end{equation}
We found the following three types of solutions:
\begin{enumerate}[(a)]
\item
$p=p(y)$, $q=q(y)$;
\item
$p=x$, $q=q(y)$;
\item
$p=x$, $q=xf(y)x$.
\end{enumerate}
In fact, one can show that, up to automorphisms of $R_2$, these are all solutions.

Replacing $x$ by $x+1$ in (b) we get the following compatible
double Poisson structures (where $y$ is central for both structures):
$$
\ldb x,x\rdb_0=1\otimes y-y\otimes1
\,\,\,\,\text{and}\,\,\,\,
\ldb x,x\rdb_1=x\otimes y-y\otimes x\,.
$$
Let 
\begin{equation}\label{20140708:eq7}
h_0=1
\,\,\,\,
\text{and} 
\,\,\,\,h_n=\frac1n(x+y)^n
\,\,\,\,
\text{for every integer }n\geq1
\,.
\end{equation}
We have the following recursion relations
\begin{equation}\label{20140708:eq5}
\{\tr(h_0),x\}_0=0
\,\,\,\,\text{and}\,\,\,\,
\{ \tr(h_n),x\}_1=\{ \tr(h_{n+1}),x\}_0\,,\,\,
n\in\mb Z_+\,.
\end{equation}
which can be proved using
equation \eqref{20140709:eq1} and the fact that for $n\geq1$
$$
\frac{\partial h_n}{\partial x}
=\frac1n\sum_{k=0}^{n-1}(x+y)^k\otimes(x+y)^{n-1-k}
\,.
$$
Hence, by Corollary \ref{cor:lenard_magri},
we get that all $\tr(h_n)$ are in involution with respect to both brackets $\{-,-\}_{0,1}$
on $\quot{R_2}{[R_2,R_2]}$, given by equation \eqref{20140709:eq1}.
Furthermore all the $\tr(h_n)$ are integrals of motion for the corresponding
compatible Hamiltonian ODEs, which are computed using \eqref{20140702:eq1}:
\begin{equation}\label{20140708:eq6}
\frac{dx}{dt_n}=x(x+y)^ny-y(x+y)^nx
\,,\,\,\,\,
\frac{dy}{dt_n}=0
\,,\,\,\,\,
n\in\mb Z_+
\,.
\end{equation}

Note that the algebra $R_2$ is $\mb Z_+^2$-graded in the obvious way. The homogeneous components
of the $h_n$'s are still in involution with respect to both brackets and the corresponding
Hamiltonian vector fields still commute.
Thus, for $n\in\mb Z_+$, equation \eqref{20140708:eq6} produces $n+1$ compatible
Hamiltonian ODEs,
and they are compatible for all $n$.
For example, for $n=0$, equation \eqref{20140708:eq6} is
$$
\frac{dx}{dt_0}=xy-yx
\,,\,\,\,\,
\frac{dy}{dt_0}=0
\,.
$$
The homogeneous components of equation \eqref{20140708:eq6} for $n=1$
produce the following two Hamiltonian ODEs:
$$
\begin{array}{l}
\displaystyle{
\frac{dx}{dt_1}=x^2y-yx^2
\,,\,\,\,\,
\frac{dy}{dt_1}=0
\,,
}\\
\displaystyle{
\frac{dx}{dt_{\tilde1}}=xy^2-y^2x
\,,\,\,\,\,
\frac{dy}{dt_{\tilde1}}=0
\,.
}
\end{array}
$$
As explained in \cite{MS00} the first of the above equations is a higher symmetry of the
\emph{Euler equation} and they find infinitely many its conserved densities.
For $n=2$, the homogeneous $(2,2)$-component of equation \eqref{20140708:eq6} produces the
following equation which also was found in \cite{MS00}:
$$
\frac{dx}{dt_0}=x^2y^2+xyxy-yxyx-y^2x^2
\,,\,\,\,\,
\frac{dy}{dt_0}=0
\,.
$$

Furthermore, let us consider on $R_2$ the following double Poisson structure (of type (c)):
$$
\ldb x,x\rdb=x\otimes xyx-xyx\otimes x\,,
\qquad
\ldb x,y\rdb=\ldb y,y\rdb=0
\,.
$$
Let $h_0=1$ and
\begin{equation}\label{20141012:eq1}
h_n=\frac{x^n}n
\,,\,\,
\tilde h_n=x^ny
\,,\,\,
\bar h_n=x(yx)^n
\,\,\,\,
\text{for every integer }n\geq1
\,.
\end{equation}
Using the fact that
$$
\begin{array}{c}
\displaystyle{
\frac{\partial h_n}{\partial x}
=\frac1n\sum_{k=0}^{n-1}x^k\otimes x^{n-1-k}
\,,
\,\,\,\,
\frac{\partial \tilde h_n}{\partial x}
=\sum_{k=0}^{n-1}x^k\otimes x^{n-1-k}y
\,,
}
\\
\displaystyle{
\frac{\partial \bar h_n}{\partial x}
=1\otimes(yx)^n
+\sum_{k=0}^{n-1}x(yx)^ky\otimes (yx)^{n-1-k}
\,.
}
\end{array}
$$
and equation \eqref{20140709:eq1},
it is straightforward to show that 
$$
\{\tr(X_n),\tr(Y_m)\}=0
\,\,\,\,
\text{for every integers }n,m\geq1
\,,
$$
where $X$ and $Y$ can be $h$, $\tilde h$ or $\bar h$.
Hence, we get the following integrable hierarchy of Hamiltonian ODEs:
\begin{equation}\label{20141012:eq2}
\begin{array}{l}
\displaystyle{
\frac{dx}{dt_n}=x^{n+1}yx-xyx^{n+1}
\,,\,\,\,\,
\frac{dy}{dt_n}=0
\,,\,\,\,\,
n\geq1
\,,
}
\\
\displaystyle{
\frac{dx}{dt_{\tilde n}}=
\sum_{k=0}^{\tilde n-2}x^{\tilde n-k}yx^{k+1}yx-xyx^{k+1}yx^{\tilde n-k}
\,,\,\,\,\,
\frac{dy}{dt_{\tilde n}}=0
\,,\,\,\,\,
\tilde n\geq1
\,,
}
\\
\displaystyle{
\frac{dx}{dt_{\bar n}}=
x(xy)^{\bar n+1}x-x(yx)^{\bar n+1}x
\,,\,\,\,\,
\frac{dy}{dt_{\bar n}}=0
\,,\,\,\,\,
\bar n\geq1
\,.}
\end{array}
\end{equation}
The first equation of the hierarchy, corresponding to $n=1$ (and the Hamiltonian density
$h_1=x$), is
$$
\frac{dx}{dt_1}=x^{2}yx-xyx^{2}
\,,\,\,\,\,
\frac{dy}{dt_1}=0
\,.
$$
This equation already appeared in \cite{MS00},
where its higher symmetries of degree $n+2$ in the variable
$x$ and degree $1$ in the variable $y$ were found (which correspond to the first line in equation
\eqref{20141012:eq2}).

\subsection
{Basic de Rham complex over an algebra of ordinary differential functions
}\label{sec:derham}

Let $V$ be an algebra of ordinary differential functions.
We define the \emph{basic de Rham complex} $\widetilde{\Omega}(V)$
as the free product of the algebra $V$ and the algebra $\mb F\langle dx_1,\dots,dx_\ell\rangle$
of non-commutative polynomials in $\ell$ variables $dx_1,\dots,dx_\ell$.
It is a $\mb Z_+$-graded unital associative algebra,
where $f\in V$ has degree $0$ and the $dx_i$'s have degree $1$.
We consider it as a superalgebra,
with superstructure compatible with the $\mb Z_+$-grading.
The subspace $\widetilde{\Omega}^n(V)$ of degree $n$ consists of linear combinations
of elements of the form
\begin{equation}\label{20140623:eq1}
\widetilde\omega
=
f_1 dx_{i_1}f_2dx_{i_2}\dots f_ndx_{i_n}f_{n+1}
\,\,,\,\text{ where }\,
f_1,\dots,f_{n+1}\in V
\,.
\end{equation}
In particular, $\widetilde{\Omega}^0(V)=V$
and $\widetilde{\Omega}^1(V)\simeq\oplus_{i=1}^\ell V dx_i V$.

Define the de Rham differential $d$ on $\widetilde{\Omega}(V)$ as the odd derivation of degree $1$
on the superalgebra $\widetilde{\Omega}(V)$
by letting
\begin{equation}\label{20140623:eq2}
df
=
\sum_{i=1}^\ell
\Big(\frac{\partial f}{\partial x_i}\Big)^\prime
dx_i
\Big(\frac{\partial f}{\partial x_i}\Big)^{\prime\prime}
\in\widetilde{\Omega}^1(V)
\,\text{ for }\,
f\in\widetilde{\Omega}^0(V)
\,\,,\,\,\,\,
d(dx_i)=0
\,.
\end{equation}
In order to check that $d$ is a differential on $\widetilde{\Omega}(V)$, i.e. $d^2=0$,
it suffices to check that $d^2f=0$ for $f\in V$.
Using Sweedler's notation and equation \eqref{20140623:eq2} we have
$$
\begin{array}{l}
\displaystyle{
d^2f=\sum_{i,j=1}^\ell
\Big(\frac{\partial}{\partial x_j}
\Big(\frac{\partial f}{\partial x_i}\Big)^\prime\Big)^{\prime}
dx_j
\Big(\frac{\partial}{\partial x_j}
\Big(\frac{\partial f}{\partial x_i}\Big)^\prime\Big)^{\prime\prime}
dx_i
\Big(\frac{\partial f}{\partial x_i}\Big)^{\prime\prime}
}
\\
\displaystyle{
-\sum_{i,j=1}^\ell
\Big(\frac{\partial f}{\partial x_i}\Big)^{\prime}
dx_i
\Big(\frac{\partial}{\partial x_j}
\Big(\frac{\partial f}{\partial x_i}\Big)^{\prime\prime}\Big)^{\prime}
dx_j
\Big(\frac{\partial}{\partial x_j}
\Big(\frac{\partial f}{\partial x_i}\Big)^{\prime\prime}\Big)^{\prime\prime}
}
\\
\displaystyle{
=\sum_{i,j=1}^\ell
\mult\Big(
dx_j
\star_1
\Big(
\Big(\frac{\partial}{\partial x_j}\Big)_L
\Big(\frac{\partial f}{\partial x_i}\Big)
\Big)
\star_1
dx_i
\Big)
}
\\
\displaystyle{
-\sum_{i,j=1}^\ell
\mult\Big(
dx_i
\star_1
\Big(
\Big(\frac{\partial}{\partial x_j}\Big)_R
\Big(\frac{\partial f}{\partial x_i}\Big)
\Big)
\star_1
dx_j
\Big)
=0\,.
}
\end{array}
$$
The last equality follows by condition \eqref{20140702:eq2} on $V$.
So we can consider the corresponding cohomology complex 
$(\widetilde{\Omega}(V),d)$.

Given an evolutionary vector field $X_P\in\Vect(V)$ (cf. \eqref{evol-vect-field}),
we define the corresponding \emph{Lie derivative} 
$L_P:\,\widetilde{\Omega}(V)\to\widetilde{\Omega}(V)$
as the even derivation of degree $0$ extending $X_P$ from $V$,
and such that $L_P(dx_i)=dP_i$, $i=1,\dots,\ell$.
In other words, if $\tilde{\omega}$ is as in \eqref{20140623:eq1},
then, recalling \eqref{20140623:eq2}, we have
\begin{equation}\label{20140623:eq3}
\begin{array}{l}
\displaystyle{
L_P(\widetilde\omega)
=
\sum_{s=1}^{n+1} f_1 dx_{i_1}f_2\dots dx_{i_{s-1}}X_P(f_s)dx_{i_{s}}\dots dx_{i_n}f_{n+1}
} \\
\displaystyle{
+\sum_{s=1}^n \sum_{i=1}^\ell f_1 dx_{i_1}f_2\dots f_s
\Big(\frac{\partial P_{i_s}}{\partial x_i}\Big)^\prime
dx_i
\Big(\frac{\partial P_{i_s}}{\partial x_i}\Big)^{\prime\prime}
f_{s+1}\dots dx_{i_n}f_{n+1}
\,.}
\end{array}
\end{equation}
Next, we define the corresponding \emph{contraction operator} 
$\iota_P:\,\widetilde{\Omega}(V)\to\widetilde{\Omega}(V)$
as the odd derivation of degree $-1$ defined on generators
by $\iota_P(f)=0$, for $f\in V$, and $\iota_P(dx_i)=P_i$.
In other words, if $\tilde{\omega}$ is as in \eqref{20140623:eq1},
then
\begin{equation}\label{20140623:eq4}
\iota_X(\widetilde\omega)
=
\sum_{s=1}^n (-1)^{s+1}
f_1 dx_{i_1}f_2\dots dx_{i_{s-1}}f_{s}
P_{i_s}
f_{s+1} dx_{i_{s+1}}f_{s+2}\dots dx_{i_n}f_{n+1}
\,.
\end{equation}
\begin{proposition}\label{20140623:prop}
The Lie derivatives \eqref{20140623:eq3}
and the contraction operators \eqref{20140623:eq4}
define a structure of a $\Vect(V)$-complex on $\widetilde{\Omega}(V)$.
In other words the following relations hold for $P,Q\in V^\ell$
(recalling \eqref{commut-evf}):
\begin{enumerate}[(a)]
\item
$[\iota_P,\iota_Q]=0$;
\item
$[L_P,\iota_Q]=\iota_{[P,Q]}$;
\item
$[L_P,L_Q]=L_{[P,Q]}$;
\item
$L_P=[\iota_P,d]\,(=\iota_P d+d\iota_P)$ (Cartan's formula).
\end{enumerate}
\end{proposition}
\begin{proof}
Both sides in each formula are derivations of the superalgebra $\widetilde{\Omega}(V)$,
and they are obviously equal on generators $dx_i$ and $f\in V$.
\end{proof}
\begin{theorem}\label{20140623:thm}
The complex $(\widetilde{\Omega}(R_\ell),d)$ is acyclic,
i.e. $H^n(\widetilde{\Omega}(R_\ell),d)=\delta_{n,0}\mb F$.
\end{theorem}
\begin{proof}
Let $\Delta=\big(x_i\big)_{i=1}^\ell\in R_\ell^\ell$
(it corresponds to the degree derivation $X_\Delta$ of $R_\ell$).
Then $L_\Delta$ is a diagonalizable operator on $\widetilde{\Omega}(R_\ell)$,
with eigenvalues given by the total degree in the $x_i$'s and the $dx_i$'s.
Moreover, 
$h=\iota_\Delta$ is a homotopy operator for the complex $(\widetilde{\Omega}(R_\ell),d)$:
by Cartan's formula we have
$h(d\widetilde{\omega})+d(h\widetilde{\omega})=k\widetilde{\omega}$,
for every $\widetilde{\omega}\in\widetilde{\Omega}(R_\ell)$
of total degree $k$.
In particular, if $k\neq0$ and $d\widetilde{\omega}=0$,
then $\widetilde{\omega}=\frac1k d(h\widetilde{\omega})\in d\widetilde{\Omega}(R_\ell)$.
The claim follows.
\end{proof}
\begin{remark}
Let $\Delta_i=(\delta_{ij}x_i)_{j=1}^\ell\in V^\ell$.
Consider the filtration $V_0\subset V_1\subset\dots\subset V_\ell=V$, where
$$
V_i=\Big\{f\in V\,\Big|\,\frac{\partial f}{\partial x_j}=0\,\text{ for all } j>i\Big\}
\,,
$$
and extend it to a filtration of $\widetilde{\Omega}(V)$, letting
$\widetilde{\Omega}^n_i(V)$ contain elements of the form \eqref{20140623:eq1},
with $f_1,\dots,f_{n+1}\in V_i$ and $i_1,\dots,i_n\leq i$.
Then, 
$L_{\Delta_i}$ preserves the space $\widetilde{\Omega}^k_{i}(V)$,
and $\iota_{\Delta_i}(\widetilde{\Omega}^k_{i}(V))\subset\widetilde{\Omega}^{k-1}_{i}(V)$.
Assume that $L_{\Delta_i}:\,\widetilde{\Omega}^k_{i}(V)\to\widetilde{\Omega}^k_{i}(V)$ 
is surjective and its kernel is $\widetilde{\Omega}^k_{i-1}(V)$.
This is true, for example, for the algebra $R_\ell$ of non-commutative polynomials
in $\ell$ variables.
Then, it is easy to check that the operator
$h_i=L_{\Delta_i}^{-1}\circ\iota_{\Delta_i}:\, \widetilde{\Omega}^k_i(V)
\to \widetilde{\Omega}^{k-1}_i(V)$,
defined up to elements of $\widetilde{\Omega}^{k-1}_{i-1}(V)$,
is a local homotopy operator,
in the sense that
$$
(d \circ h_i + h_i \circ d)(\widetilde\omega) - \widetilde\omega \in \widetilde{\Omega}^k_{i-1}(V)
\,\text{ for every }\, \widetilde\omega \in \widetilde{\Omega}^k_i(V)
\,.
$$
As a consequence, under the above assumptions the basic de Rham complex is acyclic.
\end{remark}

\subsection{
Reduced de Rham complex
}\label{sec:red-com}

Obviously the subspace 
$[\widetilde{\Omega}(V),\widetilde{\Omega}(V)]\subset\widetilde{\Omega}(V)$
is compatible with the $\mb Z_+$-grading
and it is is preserved by $d$.
Hence, we can consider the \emph{reduced de Rham complex}
\begin{equation}\label{20140623:eq6a}
\Omega(V)
=\quot{\widetilde{\Omega}(V)}{[\widetilde{\Omega}(V),\widetilde{\Omega}(V)]}
=\oplus_{n\in\mb Z_+}\Omega^n(V)
\,,
\end{equation}
with the induced action of $d$.
Since Lie derivatives $L_P$ and contraction operators $\iota_P$, for $P\in V^\ell$, 
are derivations of the superalgebra $\widetilde{\Omega}(V)$,
they preserve the commutator subspace $[\widetilde{\Omega}(V),\widetilde{\Omega}(V)]$.
Hence, they induce well defined maps on the reduced de Rham complex $\Omega(V)$
satisfying all relations of Proposition \ref{20140623:prop}.

It follows by the same argument as in the proof of Theorem \ref{20140623:thm}
that the complex $(\Omega(R_\ell),d)$ is acyclic as well:
\begin{equation}\label{20140623:eq6b}
H^n(\Omega(R_\ell),d)=\delta_{n,0}\mb F
\,.
\end{equation}

We obviously have $\Omega^0(V)=\quot{V}{[V,V]}$.
Elements of $\Omega^1(V)$ can be understood as $1$-forms
$\omega=\sum_{i=1}^\ell f_idx_i$,
and elements of $\Omega^n(V)$
as sort of cyclically skewsymmetric ``products'' of $1$-forms.
We provide now an alternative explicit description of the spaces $\Omega^n(V)$
for $n\geq1$ as follows.
We identify the space $\Omega^n(V)$ with the space $\Sigma^n(V)$
of arrays
$\big(A_{i_1\dots i_n}\big)_{i_1,\dots,i_n=1}^\ell$
with entries $A_{i_1\dots i_n}\in V^{\otimes n}$,
satisfying the following skewsymmetry condition:
\begin{equation}\label{20140626:eq1}
A_{i_1\dots i_n}
=
-(-1)^{n}(A_{i_2\dots i_n i_1})^\sigma
\,\,\,\,
\text{ for all } i_1,\dots,i_n\in\{1,\dots,\ell\}
\,,
\end{equation}
where $\sigma$ denotes the action of the cyclic permutation on $V^{\otimes n}$ 
as in \eqref{20140606:eq3}.
To prove the isomorphism $\Omega^n(V)\simeq\Sigma^n(V)$
we write explicitly the maps in both directions.
Given the coset $\omega=[\widetilde{\omega}]\in\Omega^n(V)$,
with $\widetilde{\omega}$ as in \eqref{20140623:eq1},
we map it to the array (recall \eqref{20140606:eq3b} and \eqref{20140606:eq3c})
$\big(A_{i_1\dots i_n}\big)_{i_1,\dots,i_n=1}^\ell\in\Sigma^n(V)$,
with entries
$A_{j_1\dots j_n}=0$, unless $(j_1,\dots,j_n)$ is a cyclic permutation of $(i_1,\dots,i_n)$,
and
\begin{equation}\label{20140626:eq2}
A_{j_1\dots j_n}
=
\frac1n (-1)^{s(n-s)}
a^{s+1}\otimes \dots\otimes a^n\otimes a^{n+1}a^1\otimes a^2\otimes\dots\otimes a^s
\,,
\end{equation}
for $(j_1,\dots,j_s)=(i_{\sigma^s(1)},\dots,i_{\sigma^s(n)})$,
The inverse map $\Sigma^n(V)\to\Omega^n(V)$
is given by (in Sweedler's notation):
\begin{equation}\label{20140626:eq3}
\big(A_{i_1\dots i_n}\big)_{i_1,\dots,i_n=1}^\ell
\mapsto
\sum_{i_1,\dots,i_n=1}^\ell
[
A_{i_1\dots i_n}^\prime dx_{i_1}
A_{i_1\dots i_n}^{\prime\prime}dx_{i_2}
\dots
A_{i_1\dots i_n}^{\prime\dots\prime}dx_{i_n}
]
\,.
\end{equation}
It is easy to check that the maps \eqref{20140626:eq2} and \eqref{20140626:eq3}
are well defined,
and that they are inverse to each other,
thus proving that the space $\Omega^n(V)$ and the space of arrays $\Sigma^n(V)$
can be identified using these maps.

It is also not hard to find the formula for the
differential $d:\,\Sigma^n(V)\to\Sigma^{n+1}(V)$
corresponding to the differential $d$ of the reduced complex $\Omega(V)$
under this identification.
For $n=0$, we have (recalling \eqref{eq:trace})
\begin{equation}\label{20140626:eq4a}
d(\tr(f))
=
\Big(\mult \Big(\frac{\partial f}{\partial x_i}\Big)^\sigma \Big)_{i=1}^\ell
\,.
\end{equation}
For $A=\big(A_{i_1\dots i_n}\big)_{i_1,\dots,i_n=1}^\ell\in\Sigma^n(V)$,
where $n\geq1$, we have, using notation \eqref{20140606:eq1a}
and recalling \eqref{20140606:eq3b} and \eqref{20140606:eq3c},
\begin{equation}\label{20140626:eq4b}
(dA)_{i_1\dots i_{n+1}}
=
\frac1{n+1}\sum_{s=1}^{n+1}
\sum_{t=1}^{n} (-1)^{sn+t-1}
\bigg(
\Big(
\frac{\partial}{\partial x_{i_{\sigma^s(t)}}}
\Big)_{(t)}
A_{
i_{\sigma^s(1)}
\stackrel{t}{\check{\dots}}
i_{\sigma^s(n+1)}
}
\bigg)^{\sigma^s}
\,,
\end{equation}
where $\stackrel{t}{\check{\dots}}$ means that we skip the index $i_{\sigma^s(t)}$.
Thanks to Lemma \ref{20140606:lem},
we can rewrite formula \eqref{20140626:eq4b} as follows
\begin{equation}\label{20140626:eq4b2}
\begin{split}
& (dA)_{i_1\dots i_{n+1}} \\
& =
\frac{n}{n+1}
\bigg(
\sum_{s=1}^{n}
(-1)^{s-1}
\Big(
\frac{\partial}{\partial x_{i_s}}
\Big)_{(s)}
A_{i_1
\stackrel{s}{\check{\dots}}
i_{n+1}}
+(-1)^k
\Big(\Big(
\frac{\partial}{\partial x_{i_{n+1}}}
\Big)_{(1)}
A_{i_1\dots i_n}
\Big)^{\sigma^n}
\bigg)
\,.
\end{split}
\end{equation}
In particular, for $n=1$, we have, for $F=\big(F_j\big)_{j=1}^\ell\in V^{\oplus\ell}=\Sigma^1(V)$,
\begin{equation}\label{20140626:eq4c}
(d F)_{ij}
=
\frac12\Big(
\frac{\partial F_j}{\partial x_i}
-
\Big(\frac{\partial F_i}{\partial x_j}\Big)^\sigma
\Big)
\,.
\end{equation}
For $n=2$, let
$A=\big(A_{ij}\big)_{i,j=1}^\ell\in\Sigma^2(V)$,
namely $A_{ij}\in V\otimes V$
and $(A_{ji})^\sigma=-A_{ij}$.
We have
\begin{equation}\label{20140626:eq4d}
\begin{array}{l}
\displaystyle{
(dA)_{ijk}
=
\frac23
\bigg(
\Big(
\frac{\partial}{\partial x_{i}}
\Big)_{L}
A_{jk}
-
\Big(
\frac{\partial}{\partial x_{j}}
\Big)_{R}
A_{ik}
+
\Big(\Big(
\frac{\partial}{\partial x_{k}}
\Big)_{L}
A_{ij}
\Big)^{\sigma^2}
\bigg)
\,.}
\end{array}
\end{equation}

As an application of \eqref{20140623:eq6b}, we get the following
\begin{corollary}\label{victor:cor1}
\begin{enumerate}[(a)]
\item
A $0$-form $\tr(f)\in\Omega^0(R_\ell)$ is closed if and only if $f\in\mb F+[R_\ell,R_\ell]$.
\item
A $1$-form $F=\big(F_i\big)_{i=1}^\ell\in R_\ell^{\oplus\ell}=\Sigma^1(R_\ell)$ is closed
if and only if
there exists $\tr(f)\in \quot{R_\ell}{[R_\ell,R_\ell]}$
such that $F_i=\mult\Big(\frac{\partial f}{\partial x_i}\Big)^\sigma$
for every $i=1,\dots,\ell$.
\item
A $2$-form 
$\alpha=\big(A_{ij}\big)_{i,j=1}^\ell\in\Sigma^2(R_\ell)$
is closed if and only if
there exists $F=\big(F_i\big)_{i=1}^\ell\in R_\ell^{\oplus\ell}$
such that 
$$
A_{ij}
=
\frac12\Big(
\frac{\partial F_j}{\partial x_i}
-
\Big(\frac{\partial F_i}{\partial x_j}\Big)^\sigma
\Big)
\,,
$$
for every $i,j=1,\dots,\ell$.
\end{enumerate}
\end{corollary}
\begin{remark}\label{20140711:rem2}
For $F\in V^{\oplus\ell}=\Sigma^1(V)$, define the \emph{Jacobian}
$$
J_F
=\left(\frac{\partial F_i}{\partial x_j}\right)_{i,j=1}^{\ell}
\in\Mat_{\ell\times\ell}(V\otimes V)
\,.
$$
Then equation \eqref{20140626:eq4c} becomes:
$$
dF=\frac{1}{2}\left(
J_F^t-J_F^\sigma
\right)
\,.
$$
(Here, $t$ stands for transpose.)
Therefore, recalling \eqref{20140710:eq6}, we see that $dF=0$ if and only if
$J_F=J_F^\dagger$.
\end{remark}
%

\section{Double Poisson vertex algebras and non-commutative Hamiltonian PDEs}\label{sec:3}

\subsection{Definition of double Poisson vertex algebra}

By a \emph{differential algebra} we mean a unital associative
(not necessarily commutative) algebra
$\mc V$ over the field $\mb F$,
endowed with a derivation $\partial\in\Der\mc V$.
\begin{definition}
An $n$-\emph{fold} $\lambda$-\emph{bracket} on $\mc V$ is a linear map
$$
\ldb-{}_{\lambda_1}-\dots-{}_{\lambda_{n-1}}-\rdb:
\mc V^{\otimes n}\to\mc V^{\otimes n}[\lambda_1,\ldots,\lambda_{n-1}]
$$
which satisfies the following sesquilinearity conditions
\begin{equation}\label{20140702:eq4}
\ldb a_1{}_{\lambda_1}\cdots{}_{\lambda_{i-1}}(\partial a_i)_{\lambda_i}\cdots a_{n-1}{}_{\lambda_{n-1}}a_n\rdb
=-\lambda_i\ldb a_1{}_{\lambda_1}\cdots a_{n-1}{}_{\lambda_{n-1}}a_n\rdb\,,
\end{equation}
for all $i=1,\ldots,n-1$, and
\begin{equation}\label{20140702:eq5}
\ldb a_1{}_{\lambda_1}\cdots a_{n-1}{}_{\lambda_{n-1}}(\partial a_n)\rdb
=(\lambda_1\cdots+\lambda_{n-1}+\partial)
\ldb a_1{}_{\lambda_1}\cdots a_{n-1}{}_{\lambda_{n-1}}a_n\rdb\,,
\end{equation}
and the following Leibniz rules:
\begin{equation}\label{20140702:eq6}
\begin{array}{l}
\displaystyle{
\vphantom{\Big(}
\ldb a_1{}_{\lambda_1}\cdots bc_{\lambda_i}\dots a_{n-1}{}_{\lambda_{n-1}}a_n\rdb
=\big(e^{\partial\partial_{\lambda_i}}b\big)\star_{i}
\ldb a_1{}_{\lambda_1}\cdots c_{\lambda_i}\dots a_{n-1}{}_{\lambda_{n-1}}a_n\rdb
} \\
\displaystyle{
\vphantom{\Big(}
+\ldb a_1{}_{\lambda_1}\cdots b_{\lambda_i+\partial}\dots a_{n-1}{}_{\lambda_{n-1}}a_n\rdb_\to 
\star_{n-i} c
\,,}
\end{array}
\end{equation}
for all $i=1,\ldots,n$, where $\star_n=\star_0$.
The notation on the RHS of \eqref{20140702:eq6} needs some explanations.
The $\star_i$-products were introduced in \eqref{20140606:eq2}.
Given a polynomial 
$P(\lambda_1,\dots,\lambda_{n-1})
=\sum a_1\otimes\dots\otimes a_n\lambda_1^{k_1}\dots\lambda_{n-1}^{k_{n-1}}
\in \mc V^{\otimes n}[\lambda_1,\dots,\lambda_{n-1}]$,
we denote
$$
\begin{array}{l}
\displaystyle{
\vphantom{\Big)}
\big(e^{\partial\partial_{\lambda_i}}f\big)\star_{i}
P(\lambda_1,\dots,\lambda_i,\dots,\lambda_{n-1})
} \\
\displaystyle{
\vphantom{\Big)}
=\sum \sum_{j=0}^{k_i}\binom{k_i}{j}
a_1\otimes\dots a_i\otimes(\partial^j f)a_{i+1}\otimes\dots\otimes a_n\lambda_1^{k_1}\dots
\lambda_i^{k_i-j}\dots\lambda_{n-1}^{k_{n-1}}
\,,}
\end{array}
$$
in other words, applying $e^{\partial\partial_{\lambda_i}}$ amounts to replacing $\lambda_i$
by $\lambda_i+\partial$, and the parentheses mean that $\partial$ should be applied to $f$.
Moreover, we denote
\begin{equation}\label{20140707:eq6}
\begin{array}{l}
\displaystyle{
\vphantom{\Big)}
P(\lambda_1,\dots,\lambda_i+\partial,\dots,\lambda_{n-1})_\to\star_{n-i} f
} \\
\displaystyle{
\vphantom{\Big)}
=\sum \sum_{j=0}^{k_i}\binom{k_i}{j}
a_1\otimes\dots \otimes a_i(\partial^j f)\otimes\dots\otimes a_n\lambda_1^{k_1}\dots
\lambda_i^{k_i-j}\dots\lambda_{n-1}^{k_{n-1}}
\,,}
\end{array}
\end{equation}
in other words, the arrow means that $\partial$ is applied to $f$.
\end{definition}
In the special case of $2$-fold $\lambda$-brackets,
the sesquilinearity conditions \eqref{20140702:eq4} and \eqref{20140702:eq5} are
\begin{equation}\label{20140702:eq4b}
\ldb \partial a_\lambda b\rdb=-\lambda \ldb a_\lambda b\rdb
\,\,,\,\,\,\,
\ldb a_\lambda \partial b\rdb=(\lambda+\partial) \ldb a_\lambda b\rdb
\,,
\end{equation}
and the Leibniz rules \eqref{20140702:eq6} are
\begin{equation}\label{20140702:eq6b}
\begin{array}{l}
\displaystyle{
\vphantom{\Big(}
\ldb a_\lambda bc\rdb
=
\ldb a_\lambda b\rdb c+b\ldb a_\lambda c\rdb
\,} \\
\displaystyle{
\vphantom{\Big(}
\ldb ab_\lambda c\rdb
=
\ldb a_{\lambda+\partial} c\rdb_\to \star_1 b
+(e^{\partial\partial_\lambda}a)\star_1 \ldb b_\lambda c\rdb
\,.}
\end{array}
\end{equation}

In analogy with \eqref{20140606:eq2} and \eqref{20140703:eq1}, 
we also let, for $a,b,c\in\mc V$,
\begin{equation}\label{notation}
\begin{array}{l}
\displaystyle{
\vphantom{\Big(}
\ldb a_\lambda (b\otimes c)\rdb_L:=\ldb a_{\lambda} b\rdb\otimes c
\,\,,\,\,\,\,
\ldb a_\lambda (b\otimes c)\rdb_R:=b\otimes \ldb a_{\lambda} c\rdb
\,,} \\
\displaystyle{
\vphantom{\Big(}
\ldb (a\otimes b)_\lambda c\rdb_L
:=
\ldb a_{\lambda+\partial} c\rdb_\to \otimes_1 b
\,,} \\
\displaystyle{
\vphantom{\Big(}
\ldb (a\otimes b)_\lambda c\rdb_R
:=
(e^{\partial\partial_\lambda}a)\otimes_1\ldb b_\lambda c\rdb
\,\Big(=
\ldb b_{\lambda+\partial} c\rdb_\to \otimes_1 a
\Big)
\,,}
\end{array}
\end{equation}

\begin{definition}\label{20140606:def-2}
A \emph{double Poisson vertex algebra} is a differential algebra $\mc V$,
with derivation $\partial:\,\mc V\to\mc V$, 
endowed with a $2$-fold $\lambda$-bracket $\ldb-_\lambda-\rdb:\,\mc V\times \mc V\to \mc V\otimes \mc V$
satisfying the following axioms:
\begin{enumerate}[(i)]
\item
skewsymmetry: 
\begin{equation}\label{eq:skew2}
\ldb a_{\lambda}b\rdb=-\ldb b_{-\lambda-\partial}a\rdb^\sigma\,,
\end{equation}
where $-\lambda-\partial$ in the RHS is moved to the left, acting on the coefficients.
\item
Jacobi identity (cf. Remark \ref{20140703:rem})
\begin{equation}\label{eq:jacobi2}
\ldb a_{\lambda}\ldb b{}_\mu c\rdb\rdb_L
-\ldb b_{\mu}\ldb a{}_\lambda c\rdb\rdb_R
=\ldb\ldb a_{\lambda} b\rdb_{\lambda+\mu} c\rdb_L
\,.
\end{equation}
\end{enumerate}
\end{definition}
\begin{lemma}\label{20130916:lem1}
\begin{enumerate}[(a)]
\item
The sesquilinearity relations
$$
\begin{array}{l}
\displaystyle{
\vphantom{\Big(}
\ldb \partial A_{\lambda} B\rdb_{L(\text{resp.}R)}
=-\lambda\ldb A_\lambda B\rdb_{L(\text{resp.}R)}
\,,} \\
\displaystyle{
\vphantom{\Big(}
\ldb A_{\lambda}\partial B\rdb_{L(\text{resp.}R)}
=(\lambda+\partial)\ldb A_\lambda B\rdb_{L(\text{resp.}R)}
\,,}
\end{array}
$$
hold if either $A$ or $B$ lies in $\mc V\otimes\mc V$, and the other one lies in $\mc V$.
\item
For $a,b\in\mc V$ and $A,B\in\mc V\otimes\mc V$,
the following skewsymmetry relations hold
$$
\begin{array}{l}
\displaystyle{
\vphantom{\Big(}
\ldb a_{\lambda} B\rdb_{L(\text{resp.}R)}
=-\ldb B^\sigma_{-\lambda-\partial}a\rdb_{R(\text{resp.}L)}^{\sigma(\text{resp.}\sigma^2)}
\,,} \\
\displaystyle{
\vphantom{\Big(}
\ldb A_{\lambda} b\rdb_{L(\text{resp.}R)}
=-\ldb b_{-\lambda-\partial}A^\sigma\rdb_{R(\text{resp.}L)}^{\sigma(\text{resp.}\sigma^2)}
\,.}
\end{array}
$$
\item
For $a,b\in\mc V$ and $A,B\in\mc V\otimes\mc V$,
we have
$$
\ldb a_{\lambda} B\rdb_{L}^\sigma
=\ldb a_\lambda B^\sigma\rdb_{R}
\,\,,\,\,\,\,
\ldb A_{\lambda} b\rdb_{L}
=\ldb {A^\sigma}_\lambda b\rdb_{R}
\,.
$$
\end{enumerate}
\end{lemma}
\begin{proof}
Straightforward.
\end{proof}
\begin{lemma}\label{lem:triple_bracket}
Let $\ldb-_{\lambda}-\rdb$ be a $2$-fold $\lambda$-bracket on $\mc V$ satisfying the skewsymmetry
axiom \eqref{eq:skew2}.
Then the linear map
$\ldb-\,_\lambda-_\mu-\rdb:
\mc V^{\otimes 3}\to\mc V^{\otimes 3}[\lambda,\mu]$
defined by ($a,b,c\in\mc V$):
\begin{equation}
\label{eq:triple1}
\ldb a_\lambda b_\mu c\rdb
:=\ldb a_\lambda \ldb b_\mu c\rdb\rdb_L
-\ldb b_\mu\ldb a_\lambda c\rdb\rdb_R
-\ldb\ldb a_\lambda b\rdb_{\lambda+\mu}c\rdb_L\,,
\end{equation}
is a $3$-fold $\lambda$-bracket on $\mc V$.
\end{lemma}
\begin{proof}
Sesquilinearity properties \eqref{20140702:eq4} and \eqref{20140702:eq5}
follow easily from Lemma \ref{20130916:lem1}(a) and \eqref{20140702:eq4b}.
Let us prove that $\ldb-_{\lambda}-_{\mu}-\rdb$ satisfies the Leibniz rule
\eqref{20140702:eq6} with $i=3$. 
By \eqref{20140702:eq6b} and \eqref{notation}, we have
\begin{equation}\label{20140707:eq1}
\begin{array}{l}
\displaystyle{
\vphantom{\Big(}
\ldb a_{\lambda}\ldb b_\mu cd\rdb\rdb_L
=
\ldb a_{\lambda}\ldb b_\mu c\rdb\rdb_Ld
+c\ldb a_{\lambda}\ldb b_\mu d\rdb\rdb_L
+\ldb a_\lambda c\rdb\ldb b_\mu d\rdb
\,,} \\
\displaystyle{
\vphantom{\Big(}
\ldb b_{\mu}\ldb a_\lambda cd\rdb\rdb_R
=
\ldb b_{\mu}\ldb a_\lambda c\rdb\rdb_Rd
+c\ldb b_{\mu}\ldb a_\lambda d\rdb\rdb_R
+\ldb a_\lambda c\rdb\ldb b_\mu d\rdb
\,,} \\
\displaystyle{
\vphantom{\Big(}
\ldb\ldb a_{\lambda} b\rdb_{\lambda+\mu} cd\rdb_L
=
\ldb\ldb a_{\lambda} b\rdb_{\lambda+\mu} c\rdb_L d
+c\ldb\ldb a_{\lambda} b\rdb_{\lambda+\mu} d\rdb_L
\,.}
\end{array}
\end{equation}
Hence, equation \eqref{20140702:eq6} with $i=3$ holds.

Next we prove that
\begin{equation}\label{eq:skew3}
\ldb a_\lambda b_\mu c\rdb=\ldb b_\mu c_{-\lambda-\mu-\partial}a\rdb^\sigma\,.
\end{equation}
Equation \eqref{20140702:eq6} with $i=1$ and $2$ then follows from
\eqref{eq:skew3} and \eqref{20140702:eq6} with $i=3$.
By Lemma \ref{20130916:lem1}(b) and (c), we have
\begin{equation}\label{20140707:eq2}
\begin{array}{l}
\displaystyle{
\vphantom{\Big(}
\ldb a_\lambda \ldb b_\mu c\rdb\rdb_L
=
\ldb a_\lambda \ldb b_\mu c\rdb^\sigma\rdb_R^{\sigma^2}
=
-\ldb \ldb b_\mu c\rdb _{-\lambda-\partial} a \rdb_L^{\sigma}
\,,} \\
\displaystyle{
\vphantom{\Big(}
\ldb b_\mu \ldb a_\lambda c\rdb\rdb_R
=
\ldb b_\mu \ldb a_\lambda c\rdb^\sigma\rdb_L^\sigma
=
-\ldb b_\mu \ldb c _{-\lambda-\partial} a\rdb\rdb_L^\sigma
\,,} \\
\displaystyle{
\vphantom{\Big(}
\ldb\ldb a_\lambda b\rdb_{\lambda+\mu} c\rdb_L
=
-\ldb{\ldb b_\mu a\rdb^\sigma}_{\lambda+\mu} c\rdb_L
=
\ldb c_{-\lambda-\mu-\partial}{\ldb b_\mu a\rdb}\rdb_R^\sigma
\,.}
\end{array}
\end{equation}
Combining the three equations \eqref{20140707:eq2} we get \eqref{eq:skew3}.
\end{proof}

\subsection{The trace map and connection to Lie conformal algebras}

Given a (not necessarily commutative) differential algebra $\mc V$,
we denote by $[\mc V,\mc V]\subset\mc V$ the commutator subspace.
For $f\in\mc V$, we let $\tr(f)\in\quot{\mc V}{[\mc V,\mc V]}$ be the corresponding coset.
Furthermore, for $f\in\mc V$, we also let
$\tint f$ be the coset of $f\in\mc V$ in the quotient space $\quot{\mc V}{([\mc V,\mc V]+\partial\mc V)}$.

Given a $2$-fold $\lambda$-bracket $\ldb-_\lambda-\rdb$ on $\mc V$ 
we define the following map
$\{-_{\lambda}-\}:\mc V\otimes \mc V\to \mc V[\lambda]$
by
\begin{equation}\label{20140707:eq4}
\{a_\lambda b\}=\mult\ldb a_{\lambda}b\rdb
\end{equation}
(we extend the multiplication map
on $\mc V\otimes\mc V$ to a multiplication map 
$\mult:\,(\mc V\otimes\mc V)[\lambda]\to\mc V[\lambda]$ 
in the obvious way),
and we also define the map
$\{-,-\}:\mc V\otimes \mc V\to \mc V$
by
\begin{equation}\label{20140707:eq5}
\{a,b\}=\mult\ldb a_{\lambda}b\rdb|_{\lambda=0}
\,.
\end{equation}

\begin{lemma}\label{prop:1bis}
If $\ldb-_\lambda-\rdb$ is a skewsymmetric $2$-fold $\lambda$-bracket on $\mc V$,
then the following identity
holds in $\mc V^{\otimes2}[\lambda,\mu]$ ($a,b,c\in\mc V$):
\begin{align}
\begin{split}\label{eq:quasijacobi}
&\{a_\lambda\ldb b_\mu c\rdb\}
-\ldb b_\mu\{a_\lambda c\}\rdb
-\ldb\{a_\lambda b\}_{\lambda+\mu} c\rdb
=(\mult\otimes1)\ldb a_\lambda b_\mu c\rdb
-(1\otimes \mult)\ldb b_\mu a_\lambda c\rdb\,.
\end{split}
\end{align}
where we set $\{a_\lambda b\otimes c\}=\{a_\lambda b\}\otimes c+b\otimes\{a_\lambda c\}$.
\end{lemma}
\begin{proof}
Let us compute the three terms of the LHSof \eqref{eq:quasijacobi}.
Using \eqref{notation}, we get, after a straightforward computation,
$$
\begin{array}{l}
\displaystyle{
\vphantom{\Big(}
\{a_\lambda\ldb b_\mu c\rdb\}
=(\mult\otimes1)\ldb a_{\lambda}\ldb b_\mu c\rdb\rdb_L
+(1\otimes \mult)\ldb a_{\lambda}\ldb b_{\mu}c\rdb\rdb_R
\,,} \\
\displaystyle{
\vphantom{\Big(}
\ldb b_\mu\{a_\lambda c\}\rdb
=(1\otimes \mult)\ldb b_{\mu}\ldb a_\lambda c\rdb\rdb_L
+(\mult\otimes1)\ldb b_{\mu}\ldb a_{\lambda}c\rdb\rdb_R
\,,} \\
\displaystyle{
\vphantom{\Big(}
\ldb \{a_{\lambda}b\}_{\lambda+\mu}c\rdb
=(\mult\otimes1)\ldb \ldb a_\lambda b\rdb_{\lambda+\mu}c\rdb_L
+(1\otimes \mult)\ldb\ldb a_\lambda b\rdb_{\lambda+\mu}c\rdb_R
\,.}
\end{array}
$$
By skewsymmetry and Lemma \ref{20130916:lem1}(c), 
we can replace the last term in the RHS of the third equation
by $-(1\otimes \mult)\ldb\ldb b_\mu a\rdb_{\lambda+\mu}c\rdb_L$.
Hence, combining the three equations above, we get \eqref{eq:quasijacobi}.
\end{proof}

\begin{theorem}\label{20140707:thm}
Let $\mc V$ be a differential algebra, with derivation $\partial$, endowed with a
$2$-fold $\lambda$-bracket $\ldb-_{\lambda}-\rdb$, and let
$\{-_{\lambda}-\}$ and $\{-,-\}$ be defined as in \eqref{20140707:eq4} and \eqref{20140707:eq5}.
\begin{enumerate}[(a)]
\item
$\partial[\mc V,\mc V]\subset[\mc V,\mc V]$.
Hence, we have a well defined induced map 
(denoted, by abuse of notation, by the same symbol)
$\partial:\,\quot{\mc V}{[\mc V,\mc V]}\to\quot{\mc V}{[\mc V,\mc V]}$, given by
$\partial(\tr f)=\tr(\partial f)$.
\item
$\{[\mc V,\mc V]_\lambda\mc V\}=0$,
and $\{\mc V_\lambda{[\mc V,\mc V]}\}\subset[\mc V,\mc V]\otimes\mb F[\lambda]$.
Hence, we have well defined induced maps
(denoted, by abuse of notation, by the same symbol)
$$
\{-_\lambda-\}:\,\quot{\mc V}{[\mc V,\mc V]}\times\mc V\to\mc V[\lambda]
\,,
$$
and
$$
\{-_\lambda-\}:\,\quot{\mc V}{[\mc V,\mc V]}\times\quot{\mc V}{[\mc V,\mc V]}
\to\quot{\mc V}{[\mc V,\mc V]}[\lambda]
\,,
$$
given, respectively, by
\begin{equation}\label{20140707:eq3b}
\{\tr(f)_\lambda g\}=\mult\ldb f_\lambda g\rdb
\,,
\end{equation}
and 
\begin{equation}\label{20140707:eq3c}
\{\tr(f)_\lambda \tr(g)\}=\tr(\mult\ldb f_\lambda g\rdb)
\,.
\end{equation}
\item
If the $2$-fold $\lambda$-bracket is skewsymmetric,
then so is the $\lambda$-bracket \eqref{20140707:eq3c}:
$$
\{\tr(f)_\lambda\tr(g)\}=-\{\tr(g)_{-\lambda-\partial}\tr(f)\}\,.
$$
\item
If the $2$-fold $\lambda$-bracket defines a structure of a double Poisson vertex algebra on $\mc V$,
then the $\lambda$-bracket \eqref{20140707:eq3c} satisfies the Jacobi identity 
($f,g,h\in\mc V$)
$$
\{\tr(f)_\lambda\{\tr(g)_\mu \tr(h)\}\}-\{\tr(g)_\mu\{\tr(f)_\lambda \tr(h)\}\}
=
\{\{\tr(f)_\lambda \tr(g)\}_{\lambda+\mu} \tr(h)\}\,,
$$
thus defining a structure of a Lie conformal algebra on $\quot{\mc V}{[\mc V,\mc V]}$.
Furthermore, the $\lambda$-action of $\quot{\mc V}{[\mc V,\mc V]}$ on $\mc V$,
given by \eqref{20140707:eq3b}, defines a representation of the Lie conformal algebra
$\quot{\mc V}{[\mc V,\mc V]}$ given by conformal derivations of $\mc V$.
\item
$\{[\mc V,\mc V]+\partial V,\mc V\}=0$,
and $\{\mc V,{[\mc V,\mc V]+\partial V}\}\subset([\mc V,\mc V]+\partial\mc V)\otimes\mb F[\lambda]$.
Hence, we have well defined induced brackets
(denoted, by abuse of notation, by the same symbol)
$$
\{-,-\}:\,\quot{\mc V}{([\mc V,\mc V]+\partial\mc V)}\times\mc V\to\mc V
\,,
$$
and
$$
\{-,-\}:\,\quot{\mc V}{([\mc V,\mc V]+\partial\mc V)}\times\quot{\mc V}{([\mc V,\mc V]+\partial\mc V)}
\to\quot{\mc V}{([\mc V,\mc V]+\partial\mc V)}
\,,
$$
given, respectively, by
\begin{equation}\label{20140707:eq3b-lie}
\{\tint f, g\}:=\mult\ldb f_\lambda g\rdb|_{\lambda=0}
\,,
\end{equation}
and 
\begin{equation}\label{20140707:eq3c-lie}
\{\tint f,\tint g\}:=\tint \mult\ldb f_\lambda g\rdb|_{\lambda=0}
\,.
\end{equation}
\item
If the $2$-fold $\lambda$-bracket is skewsymmetric,
then so is the bracket \eqref{20140707:eq3c-lie}.
\item
If the $2$-fold $\lambda$-bracket defines a structure of a double Poisson vertex algebra on $\mc V$,
then the bracket \eqref{20140707:eq3c-lie} defines a
structure of a Lie algebra on $\quot{\mc V}{([\mc V,\mc V]+\partial\mc V)}$.
Furthermore, the action of $\quot{\mc V}{([\mc V,\mc V]+\partial)}$ on $\mc V$,
given by \eqref{20140707:eq3b-lie}, defines a representation of the Lie  algebra
$\quot{\mc V}{([\mc V,\mc V]+\partial\mc V)}$ by derivations of $\mc V$
commuting with $\partial$.
\end{enumerate}
\end{theorem}
\begin{proof}
Part (a) is clear, since $\partial$ is a derivation of the associative product of $\mc V$.
By the second Leibniz rule \eqref{20140702:eq6b}
we have, for $a,b,c\in\mc V$,
$$
\{ab_\lambda c\}
=
\mult\ldb ab_\lambda c\rdb
=
\mult\Big(
(e^{\partial\partial_\lambda}a)\otimes_1\ldb b_\lambda c\rdb
+\ldb a_{\lambda+\partial}c\rdb_\to\otimes_1 b
\Big)
\,.
$$
Note that the expression in parenthesis in the RHS above is
unchanged if we switch $a$ and $b$
(cf. the last identity in \eqref{notation}).
Hence, $\{ab_\lambda c\}=\{ba_\lambda c\}$.
Furthermore, we have
$$
\{a_\lambda bc\}
=
\mult\ldb a_\lambda bc\rdb
=
\mult\big(\ldb a_\lambda b\rdb c+ b\ldb a\lambda c\rdb\big)
=
\{a_\lambda b\}c+b\{a_\lambda c\}\,.
$$
Namely, the $\lambda$-action $\{a_\lambda\,-\}$ 
is by derivations of the associative product of $\mc V$.
But then
$$
\{a_\lambda bc-cb\}
=
[b,\{a_\lambda c\}]+[\{a_\lambda b\},c]\,\in[\mc V,\mc V]\otimes\mb F[\lambda]
\,.
$$
This proves part (b).
Part (c) is immediate.
Part (d) is an immediate consequence of Lemma \ref{prop:1bis}.
Finally, parts (e), (f) and (g)
can be proved in the same way, 
or by using the usual construction that associates
to a Lie conformal algebra $R$
the corresponding Lie algebra $\quot{R}{\partial R}$
and its representation on $R$.
\end{proof}

\subsection{Double Poisson vertex algebra structure on an algebra of (non-commutative)
differential functions}

Consider the algebra of non-commutative differential polynomials 
$\mc R_\ell$ in $\ell$ variables $u_i$, $i\in I=\{1,\dots,\ell\}$.
It is the algebra of non-commutative polynomials in the indeterminates $u_i^{(n)}$,
$$
\mc R_\ell=\mb F\langle u_i^{(n)}\mid i\in I,n\in\mb Z_+\rangle\,,
$$
endowed with a derivation $\partial$, defined on generators by $\partial u_i^{(n)}=u_i^{(n+1)}$,
and partial derivatives $\frac{\partial}{\partial u_i^{(n)}}$, for every $i\in I$ and
$n\in\mb Z_+$, defined on monomials by
\begin{equation}\label{20140627:eq1}
\frac{\partial}{\partial u_i^{(n)}} (u_{i_1}^{(n_1)}\dots u_{i_s}^{(n_s)})
=
\sum_{k=1}^s
\delta_{i_k,i}\delta_{n_k,n}
\,
u_{i_1}^{(n_1)}\dots x_{i_{k-1}}^{(n_{k-1})}
\otimes
u_{i_{k+1}}^{(n_{k+1})}\dots u_{i_s}^{(n_s)}
\,,
\end{equation}
which are commuting $2$-fold derivations of $\mc R_\ell$ such that
\begin{equation}\label{eq:comm}
\left[\frac{\partial}{\partial u_i^{(n)}},\partial\right]=\frac{\partial}{\partial u_i^{(n-1)}}\,,
\end{equation}
where the RHS is zero for $n=0$.
%
%
\begin{lemma}\label{20140626:lem2}
For any non-commutative differential polynomial $f\in \mc R_\ell$, and $i,j\in I$ and
$n,m\in\mb Z_+$, 
the partial derivatives strongly commute, i.e. we have
$$
\left(\frac{\partial}{\partial u_{i}^{(m)}}\right)_L\frac{\partial f}{\partial u_j^{(n)}}
=\left(\frac{\partial}{\partial u_j^{(n)}}\right)_R\frac{\partial f}{\partial u_{i}^{(m)}}
\,.
$$
\end{lemma}
\begin{proof}
Same as the proof of Lemma \ref{20140609:lem1}.
\end{proof}
\begin{definition}
An \emph{algebra of differential functions} in $\ell$ variables
is a unital associative differential algebra $\mc V$, with derivation $\partial$,
endowed with strongly commuting $2$-fold derivations
$\frac{\partial}{\partial u_i^{(n)}}:\,\mc V\to\mc V\otimes\mc V$, $i\in I=\{1,\dots,\ell\},\,n\in\mb Z_+$,
such that \eqref{eq:comm} holds
and, for every $f\in\mc V$, we have $\frac{\partial f}{\partial u_i^{(n)}}=0$
for all but finitely many choices of indices $(i,n)\in I\times\mb Z_+$.
\end{definition}
An example of such an algebra is the algebra $\mc R_\ell$,
endowed with the $2$-fold derivations defined in \eqref{20140627:eq1},
or its localization by non-zero elements.
\begin{lemma}\label{20140307:lem1}
Let $\ldb-_\lambda-\rdb$ be a $2$-fold $\lambda$-bracket
on a differential algebra $\mc V$.
For every $a\in \mc V$ and $B,C\in\mc V^{\otimes2}$, we have
\begin{equation}\label{20140305:eq1}
\begin{array}{l}
\displaystyle{
\ldb a_\lambda B\bullet C\rdb_L
=B\bullet_2\ldb a_\lambda C\rdb_L
+\ldb a_\lambda B\rdb_L\bullet_1 C\,,
}\\
\displaystyle{
\ldb a_\lambda B\bullet C\rdb_R
=B\bullet_2\ldb a_\lambda C\rdb_R
+\ldb a_\lambda B\rdb_R\bullet_3C\,.
}\\
\displaystyle{
\ldb B\bullet C_\lambda a\rdb_L
={{\ldb B_{\lambda+\partial} a\rdb}_L}_\to\bullet_3C
+{{\ldb C_{\lambda+\partial} a\rdb}_L}_\to\bullet_1B^\sigma
\,.
}
\end{array}
\end{equation}
In the last equation 
the arrow has a meaning similar to \eqref{20140707:eq6}
(with $\bullet$ product in place of $\star$ product).
\end{lemma}
\begin{proof}
Straightforward.
\end{proof}
\begin{theorem}\label{20130921:prop1}
\begin{enumerate}[(a)]
\item
Any $2$-fold $\lambda$-bracket on $\mc R_\ell$ has the form ($f,g\in \mc R_\ell$):
\begin{equation}\label{master-infinite}
\ldb f_{\lambda}g\rdb
=\sum_{\substack{i,j\in I\\m,n\in\mb Z_+}}
\frac{\partial g}{\partial u_j^{(n)}}
\bullet
(\lambda+\partial)^n
\ldb u_{i}{}_{\lambda+\partial}u_{j}\rdb_{\rightarrow}
(-\lambda-\partial)^m
\bullet
\left(\frac{\partial f}{\partial u_i^{(m)}}\right)^\sigma\,.
\end{equation}
where $\bullet$ is as in \eqref{20140609:eqc1}.
\item
Let $\mc V$ be an algebra of differential functions in $\ell$ variables.
Let $H(\lambda)$ be an $\ell\times\ell$ matrix with entries in $(\mc V\otimes\mc V)[\lambda]$.
We denote its entries by $H_{ij}(\lambda)=\ldb{u_j}_\lambda{u_i}\rdb$,
$i,j\in I$.
Then formula \eqref{master-infinite} defines a $2$-fold $\lambda$-bracket on $\mc V$.
\item
Equation \eqref{master-infinite} defines a structure of double Poisson vertex algebra on $\mc V$
if and only if the skewsymmetry axiom (i) and the Jacobi identity (ii) hold on the $u_i$'s,
namely, for all $i,j,k\in I$, we have
\begin{equation}\label{skew-gen-b}
\ldb {u_i}_\lambda{u_j}\rdb=-\ldb {u_j}_{-\lambda-\partial}{u_i}\rdb^\sigma
\,,
\end{equation}
and
\begin{equation}\label{jacobi-gen-b}
\ldb {u_i}_\lambda \ldb {u_j}_\mu {u_k}\rdb\rdb_L
-
\ldb {u_j}_\mu \ldb {u_i}_\lambda {u_k}\rdb\rdb_R
=
\ldb{\ldb {u_i}_\lambda {u_j}\rdb}_{\lambda+\mu} {u_k}\rdb_L
\,,
\end{equation}
where each term is understood, in terms of the matrix $H$, via \eqref{master-finite},
for example
\begin{equation}\label{20140709:eq3}
\ldb {u_i}_\lambda \ldb {u_j}_\mu {u_k}\rdb\rdb_L
=
\sum_{i\in I,n\in\mb Z_+}
\Big(\Big(\frac{\partial}{\partial u_h^{(n)}}\Big)_LH_{kj}(\mu)\Big)
\bullet_3 (\lambda+\partial)^n H_{hi}(\lambda)
\,.
\end{equation}
\end{enumerate}
\end{theorem}
\begin{proof}
The proof is similar to the proof of Theorem \ref{thm:master-finite} using Lemma \ref{20140307:lem1}.
\end{proof}
\begin{remark}\label{rem:dagger}
For a matrix differential operator 
$H(\partial)=\big(H_{ij}(\partial)\big)_{i,j=1}^\ell$,
with entries $H_{ij}(\partial)\in(\mc V\otimes\mc V)[\partial]$,
we define its adjoint operator $H^\dagger(\partial)$, with entries
$(H^\dagger)_{ij}(\partial)=((H_{ji})^*(\partial))^\sigma$,
where $A^*(\partial)$ denotes the formal adjoint of the differential operator $A(\partial)$.
Then property \eqref{skew-gen-b}
means that the matrix differential operator $H(\partial)$ is skewadjoint.
\end{remark}
\begin{definition}\label{poisson-structure-b}
A matrix differential operator 
$H(\partial)=\big( H_{ij}(\partial)\big)_{i,j=1}^\ell$,
with entries $H_{ij}(\partial)\in(\mc V\otimes\mc V)[\partial]$,
satisfying \eqref{skew-gen-b} and \eqref{jacobi-gen-b}
is called a \emph{Poisson structure} on the algebra of differential functions $\mc V$.
\end{definition}

\subsection{Evolution PDE and Hamiltonian PDE over an algebra of differential functions}

An \emph{evolution PDE} over the algebra of differential functions $\mc V$ has the form
\begin{equation}\label{evol-eq-var}
\frac{du_i}{dt}
= P_i\,\in\mc V
\,\,,\,\,\,\,
i\in I=\{1,\dots,\ell\}
\,.
\end{equation}
Assuming that time derivative commutes with space derivative, we have
$\frac{du_i^{(n)}}{dt}=\partial^n P_i$,
and, by the chain rule,
a function $f\in\mc V$ evolves according to
\begin{equation}\label{evol-eq2-var}
\frac{df}{dt}
= \sum_{(i,n)\in I\times\mb Z_+}\mult\Big((\partial^nP_i)\otimes_1\frac{\partial f}{\partial u_i^{(n)}}\Big)
= X_P(f)
\,.
\end{equation}
An \emph{integral of motion} is a local functional $\tint f\in\quot{\mc V}{\partial\mc V+[\mc V,\mc V]}$
constant in time:
\begin{equation}\label{integral-of-motion-var}
\frac{d\tint f}{dt}
= \int \sum_{(i,n)\in I\times\mb Z_+}\mult\Big((\partial^nP_i)\otimes_1\frac{\partial f}{\partial u_i^{(n)}}\Big)
= 0
\,.
\end{equation}

A \emph{vector field} on $\mc V$ is a derivation $X:\,\mc V\to\mc V$
of the form (cf. \eqref{evol-vect-field})
\begin{equation}\label{20140704:eq1}
X(f)=\sum_{(i,n)\in I\times\mb Z_+} \mult \Big(P_{i,n}\otimes_1\frac{\partial f}{\partial u_i^{(n)}}\Big)
\,,
\end{equation}
where $P_{i,n}\in\mc V$ for all $i,n$.
Note that the RHS of \eqref{20140704:eq1} is a finite sum since,
by assumptions on $\mc V$, $\frac{\partial f}{\partial u_i^{(n)}}=0$
for all but finitely many choices of indices $(i,n)$.

An \emph{evolutionary vector field}
is a vector field commuting with $\partial$.
By \eqref{eq:comm}, it has the form (cf. \eqref{evol-eq2-var})
\begin{equation}\label{20140704:eq2}
X_P(f)=\sum_{(i,n)\in I\times\mb Z_+} \mult \Big((\partial^n P_i)\otimes_1\frac{\partial f}{\partial u_i^{(n)}}\Big)
\,,
\end{equation}
for $P=(P_i)_{i=1}^\ell\in\mc V^\ell$, called the \emph{characteristics} of the evolutionary vector field $X_P$.

Vector fields form a Lie algebra, denoted by $\Vect(\mc V)$,
and evolutionary vector fields form a Lie subalgebra, denoted $\Vect^\partial(\mc V)$.
The Lie bracket is given by the same formulas as in the finite case, equation \eqref{commut-evf}.

Equation \eqref{evol-eq-var} is called \emph{compatible} with another evolution PDE
$\frac{du_i}{d\tau}=Q_i$, $i\in I$,
if the corresponding evolutionary vector fields commute.

The \emph{Hamiltonian equation} on a double Poisson vertex algebra $\mc V$,
associated to the \emph{Hamiltonian functional} $\tint h\in \quot{\mc V}{([\mc V,\mc V]+\partial\mc V)}$ 
is
\begin{equation}\label{ham-eq}
\frac{du}{dt}
=
\{\tint h,u\}
\,,
\end{equation}
for any $u\in\mc V$.
An \emph{integral of motion} for the Hamiltonian equation \eqref{ord-ham-eq}
is an element $\tint f\in\quot{\mc V}{([\mc V,\mc V]+\partial\mc V)}$ such that 
$\{\tint h,\tint f\}=0$.
In this case, by Theorem \ref{20140707:thm}(g), 
the derivations $X^h=\{\tint h,-\}$ and $X^f=\{\tint f,-\}$ 
(called Hamiltonian vector fields) commute,
and equations $\frac{du}{dt}=\{\tr f,u\}$ and \eqref{ham-eq} are \emph{compatible}.

In the special case when $\mc V$ is an algebra of differential functions
and the $2$-fold $\lambda$-bracket $\ldb-_\lambda-\rdb$ is given by 
a Poisson structure $H(\partial)$ on $\mc V$ via \eqref{master-infinite},
where $\ldb {u_i}_\lambda{u_j}\rdb=H_{ji}(\lambda)$,
the Hamiltonian equation \eqref{ham-eq} becomes the following evolution equation
\begin{equation}\label{20140702:eq1-var}
\frac{du_i}{dt}
=
\mult\sum_{j=1}^\ell
H_{ij}(\partial)\bullet \Big(\frac{\delta h}{\delta u_j}\Big)^\sigma
\,,
\end{equation}
where $\frac{\delta h}{\delta u_j}\in\mc V\otimes\mc V$ denotes the \emph{variational derivative} of $h$:
\begin{equation}\label{var-der}
\frac{\delta h}{\delta u_j}
=
\sum_{n\in\mb Z_+}(-\partial)^n\frac{\partial h}{\partial u_j^{(n)}}
\,.
\end{equation}
In equation \eqref{20140702:eq1-var}
$\partial$ is moved to the right of the $\bullet$ product, acting on $\frac{\delta h}{\delta u_j}$.
Furthermore the Lie algebra bracket $\{-,-\}$ on $\quot{\mc V}{([\mc V,\mc V]+\partial\mc V)}$ defined by 
\eqref{20140707:eq3c-lie},
becomes ($f,g\in \mc V$):
\begin{equation}\label{20140709:eq1b}
\{\tint f,\tint g\}
=\int\sum_{i,j\in I}
\mult\left(\frac{\delta g}{\delta u_j}\right)^\sigma
\mult\left(H_{ji}(\partial)\star_1\mult\left(\frac{\delta f}{\delta u_i}\right)^\sigma
\right)
\,.
\end{equation}
In this case the notions of compatibility and of integrals of motion, are consistent with those
for general evolution equations, due to Theorem \ref{20130921:prop1}.
\begin{remark}\label{20140709:rem1b}
For $H(\partial)\in\Mat_{\ell\times\ell}(\mc V\otimes \mc V)[\partial]$
and $F\in \mc V^{\oplus\ell}$ let $H(\partial)F\in \mc V^\ell$
be defined by (where we write
$H_{ij}(\partial)=\sum_{n\in\mb Z_+}(H_{ij,n}^\prime\otimes H_{ij,n}^{\prime\prime})\partial^n$)
\begin{equation}\label{20140709:eq2b}
(H(\partial)F)_i=\sum_{j\in I} \mult(H_{ij}(\partial)\star_1 F_j)
=\sum_{jin I,n\in\mb Z_+} H_{ij,n}^\prime (\partial^nF_j) H_{ij,n}^{\prime\prime}
\,. 
\end{equation}
Then, formula \eqref{20140709:eq1b} can be written in the more traditional form
\begin{equation}
\{\tint f,\tint g\}
=\int
\delta g\cdot \left(H(\partial)\delta f\right)
\,,
\end{equation}
where $\cdot$ denotes the usual dot product of vectors and
$(\delta f)_i=\mult\left(\frac{\delta f}{\delta u_i}\right)^\sigma$.
The latter notation is compatible with the theory of the variational complex
(see equation \eqref{20140626:eq4a-aff}).
\end{remark}
\begin{remark}\label{rem:20140723}
The theory of the Lenard-Magri scheme discussed in Section \ref{sec:lenard-magri}
for double Poisson algebras $V$ 
holds verbatim for double Poisson vertex algebras $\mc V$
if we replace the matrix $H$ by the matrix differential operator $H(\partial)$
and the map $\tr$ by the map $\tint$.
\end{remark}

\subsection{de Rham complex over an algebra of differential functions}\label{sec:3.5}

Let $\mc V$ be an algebra differential functions.
We define the de Rham complex $\widetilde{\Omega}(\mc V)$
as the free product of the algebra $\mc V$ and 
the algebra $\mb F\langle\delta u_i^{(n)}\,|\,i\in I=\{1,\dots,\ell\},\,n\in\mb Z_+\rangle$
of non-commutative polynomials in the variables $\delta u_i^{(n)}$.
The action of $\partial$ is extended from $\mc V$ to
a derivation of $\widetilde{\Omega}(\mc V)$
by letting $\partial(\delta u_i^{(n)})=\delta u_i^{(n+1)}$ for all $(i,n)\in I\times\mb Z_+$.
This makes $\widetilde{\Omega}(\mc V)$ a differential algebra.
It has a $\mb Z_+$-grading,
such that $f\in\mc V$ has degree $0$ and the $\delta u_i^{(n)}$'s have degree $1$.
We consider it as a superalgebra,
with superstructure compatible with the $\mb Z_+$-grading.
The subspace $\widetilde{\Omega}^k(\mc V)$ of degree $n$ consists of linear combinations
of elements of the form
\begin{equation}\label{20140623:eq1-var}
\widetilde\omega
=
f_1\delta u_{i_1}^{(m_1)}f_2\delta u_{i_2}^{(m_2)}\dots f_k\delta u_{i_k}^{(m_k)}f_{k+1}
\,\,,\,\text{ where }\,
f_1,\dots,f_{k+1}\in\mc V
\,.
\end{equation}
In particular, $\widetilde{\Omega}^0(\mc V)=\mc V$
and $\widetilde{\Omega}^1(\mc V)\simeq\oplus_{(i,n)\in I\times\mb Z_+}\mc V \delta u_i^{(n)}\mc V$.

Define the de Rham differential $\delta$ on $\widetilde{\Omega}(\mc V)$ 
as the odd derivation of degree $1$
on the superalgebra $\widetilde{\Omega}(\mc V)$
by letting
\begin{equation}\label{20140623:eq2-var}
\delta f
=
\sum_{(i,n)\in I\times\mb Z_+}
\Big(\frac{\partial f}{\partial u_i^{(n)}}\Big)^\prime
\delta u_i^{(n)}
\Big(\frac{\partial f}{\partial u_i^{(n)}}\Big)^{\prime\prime}
\in\widetilde{\Omega}^1(\mc V)
\,\text{ for }\,
f\in\mc V
\,,\,\text{ and }
\delta (\delta u_i^{(n)})=0
\,.
\end{equation}
The proof that $\delta^2=0$ is the same as for the algebra of ordinary differential functions
(cf. Section \ref{sec:derham}),
so we can consider the corresponding cohomology complex 
$(\widetilde{\Omega}(\mc V),\delta)$.

Given a vector field 
$X_P=\sum_{(i,n)\in I\times\mb Z_+}
\mult\circ\Big(P_{i,n}\otimes_1\frac{\partial}{\partial u_i^{(n)}}\Big)\in\Vect(\mc V)$ 
(cf. \eqref{20140704:eq1}),
we define the corresponding \emph{Lie derivative} 
$L_P:\,\widetilde{\Omega}(\mc V)\to\widetilde{\Omega}(\mc V)$
as the even derivation of degree $0$ extending $X_P$ from $\mc V$,
and such that $L_P(\delta u_i^{(n)})=\delta P_{i,n}$, $i\in I,\,n\in\mb Z_+$.
Next, we define the corresponding \emph{contraction operator} 
$\iota_P:\,\widetilde{\Omega}(\mc V)\to\widetilde{\Omega}(\mc V)$
as the odd derivation of degree $-1$ defined on generators
by $\iota_P(f)=0$, for $f\in \mc V$, and $\iota_P(\delta u_i^{(n)})=P_{i,n}$.
It is easy to show, in analogy to Proposition \ref{20140623:prop}
and Theorem \ref{20140623:thm},
that $\widetilde{\Omega}(\mc V)$ is a $\Vect(\mc V)$-complex
and the complex $(\widetilde{\Omega}(\mc R_\ell),\delta)$ is acyclic:
$H^n(\widetilde{\Omega}(\mc R_\ell),\delta)=\delta_{n,0}\mb F$.

\subsection{
Variational complex over an algebra of differential functions
}\label{sec:3.6}

\begin{proposition}\label{20140704:prop}
In the de Rham complex $(\widetilde{\Omega}(\mc V),\delta)$
we have:
\begin{enumerate}[(a)]
\item
The commutator subspace 
$[\widetilde{\Omega}(\mc V),\widetilde{\Omega}(\mc V)]$
is compatible with the $\mb Z_+$-grading and is preserved by $\delta$.
\item
$\delta$ and $\partial$ commute: $\partial\circ\delta=\delta\circ\partial$;
hence, $\partial\widetilde{\Omega}(\mc V)$
is compatible with the $\mb Z_+$-grading and is preserved by $\delta$.
\item
The Lie derivative $L_P$ and the contraction operator $\iota_P$,
associated to an evolutionary vector field $X_P$
of characteristics $P=(P_i)_{i=1}^\ell$ (cf. \eqref{20140704:eq2}),
commute with the action of $\partial$ on $\widetilde{\Omega}(\mc V)$.
\end{enumerate}
\end{proposition}
\begin{proof}
Part (a) is immediate, since $\delta$ is a derivation of the associative product 
in $\widetilde{\Omega}(\mc V)$.
For part (b) we just need to prove that 
$\delta(\partial\widetilde{\omega})=\partial(\delta\widetilde{\omega})$
for every $\widetilde{\omega}\in\widetilde{\Omega}(\mc V)$.
Since $[\delta,\partial]$ is a derivation of $\widetilde{\Omega}(\mc V)$,
and it is obviously zero on $\delta u_i^{(n)}$,
it suffices to check that it is zero on $f\in\mc V$.
We have, by \eqref{20140623:eq2-var},
$$
\begin{array}{l}
\displaystyle{
\vphantom{\Big(}
\partial(\delta f)
=
\sum_{(i,n)\in I\times\mb Z_+}
\Big(\partial\frac{\partial f}{\partial u_i^{(n)}}\Big)^\prime
\delta u_i^{(n)}
\Big(\partial\frac{\partial f}{\partial u_i^{(n)}}\Big)^{\prime\prime}
} \\
\displaystyle{
\vphantom{\Big(}
+\sum_{(i,n)\in I\times\mb Z_+}
\Big(\frac{\partial f}{\partial u_i^{(n)}}\Big)^\prime
\delta u_i^{(n+1)}
\Big(\frac{\partial f}{\partial u_i^{(n)}}\Big)^{\prime\prime}
\,,}
\end{array}
$$
and
$$
\delta(\partial f)
=
\sum_{(i,n)\in I\times\mb Z_+}
\Big(\frac{\partial (\partial f)}{\partial u_i^{(n)}}\Big)^\prime
\delta u_i^{(n)}
\Big(\frac{\partial(\partial f)}{\partial u_i^{(n)}}\Big)^{\prime\prime}
\,.
$$
Hence, $\partial(\delta f)=\delta(\partial f)$ by \eqref{eq:comm}.

As for part (c),
since $L_P$, $\iota_P$ and $\partial$ are derivations of the associative product 
of $\widetilde{\Omega}(\mc V)$,
it suffices to check that $[L_P,\partial]$ and $[\iota_P,\partial]$
act as zero on generators, i.e. on $f\in\mc V$ and on $\delta u_i^{(n)}$.
This is straightforward.
\end{proof}

Thanks to Proposition \ref{20140704:prop}(a-b),
we can consider the $\mb Z_+$-graded \emph{variational complex}
\begin{equation}\label{20140623:eq6a-b}
\Omega(\mc V)
=\quot{\widetilde{\Omega}(\mc V)}
{\big(\partial\widetilde{\Omega}(\mc V)+[\widetilde{\Omega}(\mc V),\widetilde{\Omega}(\mc V)]\big)}
=\oplus_{n\in\mb Z_+}\Omega^n(\mc V)
\,,
\end{equation}
with the induced action of $\delta$.
Furthermore, by Proposition \ref{20140704:prop}(c)
the Lie derivatives $L_P$ and contraction operators $\iota_P$, 
associated to the evolutionary vector field $X_P$ of characteristics $P\in\mc V^\ell$, 
induce well defined maps on the variational complex $\Omega(\mc V)$,
and they define a structure of $\Vect^\partial(\mc V)$-complex on $\widetilde{\Omega}(\mc V)$.

Since the total degree vector field $X_\Delta$, with characteristics $\Delta=(u_i)_{i=1}^\ell$,
is an evolutionary vector field on $\mc R_\ell$,
the same argument in the proof of Theorem \ref{20140623:thm}
shows that the complex $(\Omega(\mc R_\ell),\delta)$ is acyclic,
i.e.
\begin{equation}\label{20140623:eq6b-aff}
H^k(\Omega(\mc R_\ell),\delta)=\delta_{k,0}\mb F
\,.
\end{equation}

As we did in the finite case, Section \ref{sec:red-com},
we give an explicit description of the complex $(\Omega(\mc V),\delta)$.
We obviously have $\Omega^0(\mc V)=\quot{\mc V}{(\partial\mc V+[\mc V,\mc V])}$.
For $k\geq1$,
we identify $\Omega^k(\mc V)$ with the space $\Sigma^k(\mc V)$
of arrays
$\big(A_{i_1\dots i_k}(\lambda_1,\dots,\lambda_{k-1}\big)_{i_1,\dots,i_k=1}^\ell$
with entries 
$A_{i_1\dots i_k}(\lambda_1,\dots,\lambda_{k-1})\in \mc V^{\otimes k}[\lambda_1,\dots,\lambda_{k-1}]$,
satisfying the following skewadjointness condition ($i_1,\dots,i_k\in I$):
\begin{equation}\label{20140626:eq1-b}
A_{i_1\dots i_k}(\lambda_1,\dots,\lambda_{k-1})
=
-(-1)^{k}\big(
A_{i_2\dots i_k i_1}(\lambda_2,\dots,\lambda_{k-1},
-\lambda_1-\dots-\lambda_{k-1}-\stackrel{\leftarrow}{\partial})
\big)^\sigma
\,,
\end{equation}
where $\sigma$ denotes the action of the cyclic permutation on $\mc V^{\otimes k}$ 
as in \eqref{20140606:eq3},
and the arrow means that $\partial$ is acting on the coefficients of the polynomial.
To prove the isomorphism $\Omega^k(\mc V)\simeq\Sigma^k(\mc V)$
we write explicitly the maps in both directions.
Given the coset $\omega=[\widetilde{\omega}]\in\Omega^k(\mc V)$,
where $\widetilde{\omega}$ is as in \eqref{20140623:eq1-var},
we map it to the array
$A=\big(A_{j_1\dots j_k}(\lambda_1,\dots,\lambda_{k-1})\big)_{i_1,\dots,i_k=1}^\ell\in\Sigma^k(\mc V)$,
with entries 
$A_{j_1\dots j_k}(\lambda_1,\dots,\lambda_{k-1})=0$
unless $(j_1,\dots,j_k)$ is a cyclic permutation of $(i_1,\dots,i_k)$,
and
\begin{equation}\label{20140626:eq2-aff}
\begin{array}{l}
\displaystyle{
\vphantom{\Big)}
A_{j_1\dots j_k}(\lambda_1,\dots,\lambda_{k-1})
=
\frac1k (-1)^{s(k-s)}
\lambda_1^{n_{s+1}}\dots\lambda_{k-s}^{n_k}
\lambda_{k-s+1}^{n_1}\dots\lambda_{k-1}^{n_{s-1}}
} \\
\displaystyle{
\vphantom{\Big)}
(-\lambda_1-\dots-\lambda_{k-1}-\partial)^{n_s}
\big(
f_{s+1}\otimes\dots\otimes f_k\otimes f_{k+1}f_1\otimes f_2\otimes\dots\otimes f_s
\big)
\,,}
\end{array}
\end{equation}
for $(j_1,\dots,j_k)=(i_{\sigma^s(1)},\dots,i_{\sigma^s(k)})$.
The inverse map $\Sigma^k(V)\to\Omega^k(V)$
is given by (in Sweedler's notation):
\begin{equation}\label{20140626:eq3-aff}
\begin{array}{l}
\displaystyle{
\Big(
\sum_{n_1,\dots,n_{k-1}\in\mb Z_+}A_{i_1\dots i_k}^{n_1\dots n_{k-1}}
\lambda_1^{n_1}\dots\lambda_{k-1}^{n_{k-1}}
\Big)_{i_1,\dots,i_k=1}^\ell
\mapsto
} \\
\displaystyle{
\sum_{\substack{i_1,\dots,i_k\in I \\ n_1,\dots,n_{k-1}\in\mb Z_+}}
\!\!\!\!\!
\big[
(A_{i_1\dots i_k}^{n_1\dots n_{k-1}})^\prime \delta u_{i_1}^{(n_1)}
\dots
(A_{i_1\dots i_k}^{n_1 \dots n_{k-1}})^{\overbrace{\prime\dots\prime}^{k-1}} \delta u_{i_{k-1}}^{(n_{k-1})}
(A_{i_1\dots i_k}^{n_1 \dots n_{k-1}})^{\overbrace{\prime\dots\prime}^{k}} \delta u_{i_k}
\big]
\,.}
\end{array}
\end{equation}
It is easy to check that the maps \eqref{20140626:eq2-aff} and \eqref{20140626:eq3-aff}
are well defined,
and they are inverse to each other,
thus proving that the space $\Omega^k(\mc V)$ and the space of arrays $\Sigma^k(\mc V)$
can be identified using these maps.

It is also not hard to find the formula for the
differential $\delta:\,\Sigma^k(\mc V)\to\Sigma^{k+1}(\mc V)$
corresponding to the differential $\delta$ of the reduced complex $\Omega(\mc V)$
under this identification.
For $k=0$, we have
\begin{equation}\label{20140626:eq4a-aff}
\delta(\tint f)
=
\Big(\sum_{n\in\mb Z_+}
(-\partial)^n\mult \Big(\frac{\partial f}{\partial u_i^{(n)}}\Big)^\sigma \Big)_{i=1}^\ell
\,\Big(
=
\Big(
\mult \Big(\frac{\delta f}{\delta u_i}\Big)^\sigma \Big)_{i=1}^\ell
\Big)
\,.
\end{equation}
For $A=
\big(A_{i_1\dots i_k}(\lambda_1,\dots,\lambda_{k-1})\big)_{i_1,\dots,i_k=1}^\ell\in\Sigma^k(\mc V)$,
where $k\geq1$, we have, using notation \eqref{20140606:eq1a}
and recalling \eqref{20140606:eq3b} and \eqref{20140606:eq3c},
\begin{equation}\label{20140626:eq4b-aff}
\begin{array}{l}
\displaystyle{
(\delta A)_{i_1\dots i_{k+1}}(\lambda_1,\dots,\lambda_k)
=
\frac1{k+1}\sum_{s=1}^{k+1}
\sum_{t=1}^{k} 
(-1)^{sk+t-1}
\times
} \\
\displaystyle{
\times
\sum_{n\in\mb Z_+}
\bigg(
\Big(
\frac{\partial}{\partial u_{i_{\sigma^s(t)}}^{(n)}}
\Big)_{(t)}
A_{
i_{\sigma^s(1)}
\stackrel{t}{\check{\dots}}
i_{\sigma^s(k+1)}
}
(\lambda_{\sigma^s(1)}
\stackrel{t}{\check{\dots}}
\lambda_{\sigma^s(k)})
\lambda_{\sigma^s(t)}^{n}
\bigg)^{\sigma^s}
\,,}
\end{array}
\end{equation}
where $\stackrel{t}{\check{\dots}}$ means that we skip the index $i_{\sigma^s(t)}$
and, as before, 
we need to replace $\lambda_{k+1}$,
when it occurs, by $-\lambda_1-\dots-\lambda_{k}-\partial$,
with $\partial$ acting on the coefficients
(here $\sigma$ denotes the action of the cyclic element $(1,2,\dots,k+1)$).
We can use Lemma \ref{20140606:lem} to rewrite formula \eqref{20140626:eq4b-aff}
as follows
\begin{equation}\label{20140626:eq4b2-aff}
\begin{array}{l}
\displaystyle{
(\delta A)_{i_1\dots i_{k+1}}(\lambda_1,\dots,\lambda_k)
} \\
\displaystyle{
=
\frac{k}{k+1}\sum_{n\in\mb Z_+}
\bigg(
\sum_{s=1}^{k} (-1)^{s+1}
\Big(
\frac{\partial}{\partial u_{i_{s}}^{(n)}}
\Big)_{(s)}
A_{i_{1}
\stackrel{s}{\check{\dots}}
i_{k+1}}
(\lambda_{1},
\stackrel{s}{\check{\dots}},
\lambda_{k})
\lambda_{s}^{n}
} \\
\displaystyle{
+
(-1)^{k}
(-\lambda_1-\dots-\lambda_k-\partial)^{n}
\Big(
\Big(
\frac{\partial}{\partial u_{i_{k+1}}^{(n)}}
\Big)_{(1)}
A_{i_1,\dots,i_k}
(\lambda_{1},\dots,\lambda_{k-1})
\Big)^{\sigma^k}
\bigg)
\,.}
\end{array}
\end{equation}
In particular, for $k=1$, we have, for $F=\big(F_j\big)_{j=1}^\ell\in\mc V^\ell=\Sigma^1(\mc V)$,
\begin{equation}\label{20140626:eq4c-aff}
(\delta F)_{ij}(\lambda)
=
\frac12
\sum_{n\in\mb Z_+}
\Big(
\frac{\partial F_j}{\partial u_i^{(n)}} \lambda^n
- (-\lambda-\partial)^n \Big(\frac{\partial F_i}{\partial u_j^{(n)}}\Big)^\sigma
\Big)
\,.
\end{equation}
For $k=2$, let
$A=\big(A_{ij}(\lambda)\big)_{i,j=1}^\ell\in\Sigma^2(\mc V)$,
namely $A_{ij}(\lambda)\in\mc V^{\otimes2}[\lambda]$
and $(A_{ji}(-\lambda-\partial))^\sigma=-A_{ij}(\lambda)$.
We have
\begin{equation}\label{20140626:eq4d-aff}
\begin{array}{l}
\displaystyle{
(\delta A)_{ijk}(\lambda,\mu)
=
\frac{2}{3}\sum_{n\in\mb Z_+}
\bigg(
\Big(
\frac{\partial}{\partial u_{i}^{(n)}}
\Big)_L
A_{jk}
(\mu)
\lambda^{n}
} \\
\displaystyle{
-
\Big(
\frac{\partial}{\partial u_{j}^{(n)}}
\Big)_R
A_{ik}
(\lambda)
\mu^{n}
+
(-\lambda-\mu-\partial)^{n}
\Big(
\Big(
\frac{\partial}{\partial u_{k}^{(n)}}
\Big)_L
A_{ij}
(\lambda)
\Big)^{\sigma^2}
\bigg)
\,.}
\end{array}
\end{equation}
\begin{remark}\label{rem:beltrami}
In analogy with \cite{BDSK09}, we introduce the \emph{Beltrami} $2$\emph{-fold} $\lambda$\emph{-bracket}
$\ldb-_{\lambda}-\rdb^B$ using equation \eqref{master-infinite} where we let
$\ldb u_i{}_\lambda u_j\rdb^B=\delta_{ij}(1\otimes1)$, for $i,j\in I$ (note that
this $2$-fold $\lambda$-bracket is commutative and does not define a double
Poisson vertex algebra structure on $\mc V$).
Then equation \eqref{20140626:eq4d-aff} can be rewritten as follows
$$
(\delta A)_{ijk}(\lambda,\mu)
=\frac23\left(
\ldb u_i{}_\lambda A_{jk}(\mu)\rdb^B_L
-\ldb u_j{}_\mu A_{ik}(\lambda)\rdb^B_R
+\ldb A_{ij}(\lambda)_{\lambda+\mu}u_k\rdb^B_L
\right)\,.
$$
More generally, generalizing notation \eqref{notation}, we set ($s=1,\ldots,k$ and $a,b_i\in\mc V$)
$$
\ldb a_\lambda b_1\otimes \dots\otimes b_k\rdb^B_{(s)}
=b_1\otimes \dots\otimes b_{s-1}\otimes
\ldb a_\lambda b_s\rdb^B\otimes b_{s+1}\otimes\dots\otimes b_k
\,,
$$
and
$$
\ldb b_1\otimes \dots\otimes b_k{}_\lambda a\rdb^B_{(1)}
=\ldb b_1{}_{\lambda+\partial} a\rdb^B_\to\otimes_1 \left(b_{2}\otimes \dots\otimes b_k\right)
\,.
$$
Then we can rewrite equation \eqref{20140626:eq4b2-aff} as follows:
$$
\begin{array}{l}
\displaystyle{
(\delta A)_{i_1\dots i_{k+1}}(\lambda_1,\dots,\lambda_k)
=
\frac{k}{k+1}
\Big(
\sum_{s=1}^{k} (-1)^{s+1}
\ldb u_{i_{s}}{}_{\lambda_s}
A_{i_{1}
\stackrel{s}{\check{\dots}}
i_{k+1}}
(\lambda_{1},
\stackrel{s}{\check{\dots}},
\lambda_{k})
\rdb_{(s)}^B
}\\
\displaystyle{
+
(-1)^{k}
\ldb A_{i_1,\dots,i_k}
(\lambda_{1},\dots,\lambda_{k-1})
_{\lambda_1+\dots+\lambda_k}
u_{i_{k+1}}
\rdb_{(1)}^B
\Big)
\,.
}
\end{array}
$$

\end{remark}

As an application of \eqref{20140623:eq6b-aff}, we get the following
\begin{corollary}\label{victor:cor2}
\begin{enumerate}[(a)]
\item
A $0$-form $\tint f\in\Omega^0(\mc R_\ell)$ is closed if and only if 
$f\in\mb F+[\mc R_\ell,\mc R_\ell]+\partial\mc R_\ell$.
\item
A $1$-form $F=\big(F_i\big)_{i=1}^\ell\in\mc R_\ell^{\oplus\ell}=\Sigma^1(\mc R_\ell)$ is closed
if and only if
there exists $\tint f\in \quot{\mc R_\ell}{([\mc R_\ell,\mc R_\ell]+\partial\mc R_\ell)}$
such that $F_i=\mult\Big(\frac{\delta f}{\delta u_i}\Big)^\sigma$
for every $i=1,\dots,\ell$.
\item
A $2$-form 
$\alpha=\big(A_{ij}(\lambda)\big)_{i,j=1}^\ell\in\Sigma^2(\mc R_\ell)$
is closed if and only if
there exists $F=\big(F_i\big)_{i=1}^\ell\in \mc R_\ell^{\oplus\ell}$
such that 
$$
A_{ij}(\lambda)
=
\frac12
\sum_{n\in\mb Z_+}
\Big(
\frac{\partial F_j}{\partial u_i^{(n)}} \lambda^n
- (-\lambda-\partial)^n \Big(\frac{\partial F_i}{\partial u_j^{(n)}}\Big)^\sigma
\Big)
\,,
$$
for every $i,j=1,\dots,\ell$.
\end{enumerate}
\end{corollary}
\begin{remark}\label{20140711:rem3}
For $F\in \mc V^{\oplus\ell}=\Sigma^1(\mc V)$, define the \emph{Frechet derivative}
$$
D_F(\lambda)
=
\left(
\sum_{n\in\mb Z_+}\frac{\partial F_i}{\partial u_j^{(n)}} \lambda^n
\right)_{i,j=1}^{\ell}
\in\Mat_{\ell\times\ell}(\mc   V \otimes\mc V)[\lambda]
\,.
$$
Therefore, recalling Remark \ref{rem:dagger}, we see that $\delta F=0$ if and only if
$D_F(\partial)$ is selfadjoint.

Note also that, using the Beltrami $2$-fold $\lambda$-bracket introduced in Remark \ref{rem:beltrami}
we have
$$
D_F(\lambda)_{ij}=\ldb u_j{}_\lambda F_i\rdb^B
\,.
$$
Furthermore, for $f\in\mc V$, we have ($i\in I$)
$$
\left(\frac{\delta f}{\delta u_i}\right)^\sigma
=\ldb f_{\lambda}u_i\rdb^B\big|_{\lambda=0}
\,.
$$
\end{remark}

\subsection{Poisson vertex algebra structure on \texorpdfstring{$\mc V_m$}{Vm}}
\label{sec:V_m}
%
%
%
%
%
In this subsection for each double PVA $\mc V$ and a positive integer $m$ we construct a Poisson
vertex algebra $\mc V_m$. This construction is similar to the construction of a Poisson algebra $V_m$,
associated to a double Poisson algebra $V$ in \cite{VdB08}.

The construction of $V_m$ is motivated by the following considerations.
Let $V$ be a unital finitely generated associative algebra, and let $Y_m$ be the affine algebraic
variety of all homomorphisms from $V$ to $\Mat_{m\times m}\mb F$.
Let $V_m:=\mb F[Y_m]$ be the algebra of polynomial functions on $Y_m$.
We have a natural map $V\to\Mat_{m\times m}V_m$, which induces the map
$$
\varphi_m:\quot{V}{[V,V]}\to V_m
\,.
$$
It was shown in \cite{VdB08} that if $V$ is a double Poisson algebra with a $2$-fold bracket
$\ldb-,-\rdb$, then $V_m$ is a Poisson algebra with the bracket
$$
\{a_{ij},b_{hk}\}=\ldb a,b\rdb'_{hj}\ldb a,b\rdb''_{ik}
\,.
$$
It is easy to check that $\varphi_m$ is then a Lie algebra homomorphism.

Note that by the above interpretation of the algebra $V_m$, it follows that for
$V=R_N=\mb F\langle x_1,\dots,x_N\rangle$ the algebra $V_m$ is a polynomial algebra
as well (on generators $(x_s)_{ij}$, $s=1,\dots,N$, $i,j=1,\dots,m$).

Given a differential algebra $\mc V$ we define for each positive integer $m$
a commutative differential algebra $\mc V_m$,
as the commutative algebra generated by the symbols $a_{ij}$,
for $a\in\mc V$ and $i,j=1,\ldots,m$, subject to the
relations ($k\in\mb F$, $a,b\in\mc V$):
\begin{equation}\label{20130917:eq2}
(ka)_{ij}=ka_{ij}\,,
\qquad
(a+b)_{ij}=a_{ij}+b_{ij}\,,
\qquad
(ab)_{ij}=\sum_{k=1}^ma_{ik}b_{kj}\,.
\end{equation}
The derivation (which we still denote by $\partial$) on $\mc V_m$ is defined by
$$
\partial( a_{ij})=(\partial a)_{ij}\,.
$$
\begin{example}\label{20131217:exa1}
Let $\mc V=\mc R_I=\mb F\langle u_i^{(n)}\mid i\in I,n\in\mb Z_+\rangle$
be the algebra of non-commutative differential
polynomials in the variables $u_i$, $i\in I$. Then, by the above discussion,
for every $m\geq1$,
$$
\mc V_m=\mb F[u_{i,ab}^{(n)}\mid i\in I,a,b\in\{1,\dots,m\},n\in\mb Z_+]\,,
$$
where $u_{i,ab}^{(n)}=(u_i^{(n)})_{ab}$,
is the (commutative) algebra of differential polynomials in the variables
$u_{i,ab}$.
\end{example}
\begin{proposition}\label{20131217:prop1}
Let $\mc V$ be a differential algebra endowed with a $2$-fold $\lambda$-bracket
$\ldb-_\lambda-\rdb$, written as
$\ldb a_\lambda b\rdb=\sum_{n\in\mb Z_+}\left((a_nb)'\otimes (a_nb)''\right)\lambda^n$,
for all $a,b\in\mc V$. Then, for every $m\geq1$,
we have a well-defined $\lambda$-bracket on $\mc V_m$, given by
($a_{ij},b_{hk}\in\mc V_m$)
\begin{equation}\label{20130917:eq1}
\{a_{ij}{}_\lambda b_{hk}\}=\sum_{n\in\mb Z_+}(a_nb)'_{hj}(a_nb)''_{ik}\lambda^n\,,
\end{equation}
and extended to a bilinear map
$\{\,\cdot_\lambda\,\cdot\}:\mc V_m\otimes\mc V_m\to\mc V_m[\lambda]$
by sesquilinearity \eqref{eq:0.1} and the left and right Leibniz rules \eqref{eq:0.2} and \eqref{eq:0.3}.
This $\lambda$-bracket is skewsymmetric (i.e., it satisfies \eqref{eq:0.4})
if the $2$-fold $\lambda$-bracket is.
\end{proposition}
\begin{proof}
First we need to verify that the map
$\{\,\cdot_\lambda\,\cdot\}:\mc V_m\otimes\mc V_m\to\mc V_m[\lambda]$ is well
defined, that is, it is compatible with the defining relations \eqref{20130917:eq2}
of $\mc V_m$. Clearly, the RHS of \eqref{20130917:eq1} is linear in $a$ and $b$
since $\ldb-_{\lambda}-\rdb$ is a linear map.
Hence, we are left to show that
\begin{equation}\label{toprove1}
\{a_{ij}{}_{\lambda}(bc)_{hk}\}
=\sum_{l=1}^N\{a_{ij}{}_{\lambda}b_{hl}c_{lk}\}
\end{equation}
and
\begin{equation}\label{toprove2}
\{(bc)_{hk}{}_{\lambda}a_{ij}\}
=\sum_{l=1}^N\{b_{hl}c_{lk}{}_{\lambda}a_{ij}\}\,.
\end{equation}
Using the definition
of the $\lambda$-bracket \eqref{20130917:eq1} and \eqref{eq:0.2},
we have
$$
RHS\eqref{toprove1}
=\sum_{l=1}^N\sum_{n\in\mb Z_+}\left(
(a_nb)'_{hj}(a_nb)''_{il}c_{lk}+(a_nc)'_{lj}(a_nc)''_{ik}b_{hl}
\right)\lambda^n\,.
$$
On the other hand, by the first formula in \eqref{20140702:eq6b},
$$
\ldb a_\lambda bc\rdb
=b\ldb a_\lambda c\rdb+\ldb a_\lambda b\rdb c
=\sum_{n\in\mb Z_+}\left(
b(a_nc)'\otimes(a_nc)''+(a_nb)'\otimes(a_nb)''c
\right)\lambda^n.
$$
Hence, using \eqref{20130917:eq1}, we have
\begin{align*}
LHS\eqref{toprove1}
&=\sum_{n\in\mb Z_+}\left(
(b(a_nc)')_{hj}(a_nc)''_{ik}+(a_nb)'_{hj}((a_nb)''c)_{ik}
\right)\lambda^n\\
&=\sum_{l=1}^N\sum_{n\in\mb Z_+}\left(
b_{hl}(a_nc)'_{lj}(a_nc)''_{ik}+(a_nb)'_{hj}(a_nb)''_{il}c_{lk}
\right)\lambda^n\,,
\end{align*}
where in the second identity we used the third relation in \eqref{20130917:eq2}.
Hence, the identity \eqref{toprove1} follows by the commutativity of $\mc V_m$.
For \eqref{toprove2}, using the right Leibniz rule \eqref{eq:0.3}
and the definition
of the $\lambda$-bracket \eqref{20130917:eq1} we have
$$
RHS\eqref{toprove2}
=\sum_{l=1}^N\sum_{n\in\mb Z_+}\left(
(b_na)'_{il}(b_na)''_{hj}(\lambda+\partial)^nc_{lk}
+(c_na)'_{ik}(c_na)''_{lj}(\lambda+\partial)^nb_{hl}
\right)\,.
$$
On the other hand, by the second formula in \eqref{20140702:eq6b},
$$
\begin{array}{l}
\displaystyle{
\ldb bc_\lambda a\rdb
=\ldb b_{\lambda+\partial}a\rdb_{\rightarrow}\star_1 c
+(e^{\partial\frac{d}{d\lambda}}b)\star_1\ldb c_{\lambda}a\rdb
}
\\
\displaystyle{
=\sum_{n\in\mb Z_+}
\left(
(b_na)'(\lambda+\partial)^nc\otimes(b_na)''
+(c_na)'\otimes\left((\lambda+\partial)^nb\right)(c_na)''
\right)\,.
}
\end{array}
$$
Hence, using \eqref{20130917:eq1}, we have
\begin{align*}
LHS\eqref{toprove2}
&=\sum_{n\in\mb Z_+}
\left(
\left((b_na)'(\lambda+\partial)^nc\right)_{ik}(b_na)''_{hj}
+(c_na)'_{ik}\left(((\lambda+\partial)^nb)(c_na)''\right)_{hj}
\right)\\
&=\sum_{n\in\mb Z_+}
\left(
(b_na)'_{il}(\lambda+\partial)^nc_{lk}(b_na)''_{hj}
+(c_na)'_{ik}((\lambda+\partial)^nb_{hl})(c_na)''_{lj}
\right)\,,
\end{align*}
where in the second identity we used the third relation in \eqref{20130917:eq2}.
Again, the identity \eqref{toprove2} follows by the commutativity of $\mc V_m$.

Let us prove skewsymmetry of the $\lambda$-bracket given by equation \eqref{20130917:eq1}
provided that the skewsymmetry \eqref{eq:skew2} of the $2$-fold $\lambda$-bracket holds.
We have to show
that
\begin{equation}\label{toprove3}
\{a_{ij}{}_{\lambda}b_{hk}\}=-\{b_{hk}{}_{-\lambda-\partial}a_{ij}\}\,,
\end{equation}
for all $a_{ij},b_{hk}\in\mc V_m$.
We can rewrite \eqref{eq:skew2} as the following identity
$$
\sum_{n\in\mb Z_+}\left((a_nb)'\otimes(a_nb)''\right)\lambda^n
=-\sum_{n\in\mb Z_+}(-\lambda-\partial)^n\left((b_na)''\otimes(b_na)'\right)\,.
$$
Hence, we have
\begin{equation}\label{20130921:eq1}
\sum_{n\in\mb Z_+}(a_nb)_{hj}'(a_nb)_{ik}''\lambda^n
=-\sum_{n\in\mb Z_+}(-\lambda-\partial)^n(b_na)_{hj}''(b_na)_{ik}'\,,
\end{equation}
for all $i,j,h,k=1,\ldots,m$. Using \eqref{20130917:eq1} and the commutativity of $\mc V_m$, it follows
that \eqref{20130921:eq1} is equivalent to \eqref{toprove3}, thus concluding the proof.
\end{proof}
\begin{proposition}\label{20131217:prop2}
Let $\mc V$ be a differential algebra with a $2$-fold $\lambda$-bracket $\ldb-_\lambda-\rdb$.
Let
$$
\ldb a_\lambda b_\mu c\rdb
=\sum_{p,q\in\mb Z_+}\left( h_{pq}'\otimes h_{pq}''\otimes h_{pq}'''\right)\lambda^p\mu^q\,,
$$
and
$$
\ldb b_\mu a_\lambda c\rdb
=\sum_{p,q\in\mb Z_+}\left( g_{pq}'\otimes g_{pq}''\otimes g_{pq}'''\right)\lambda^p\mu^q\,,
$$
be the associated, by \eqref{eq:triple1}, $3$-fold $\lambda$-bracket.
Then, for any $m\geq1$, we have ($a_{ij}, b_{hk}, c_{ln}\in\mc V_m$):
\begin{align}
\begin{split}\label{20130921:eq2}
\{ a_{ij}{}_\lambda \{b_{hk}{}_\mu c_{ln}\} \}
-\{ b_{hk}{}_\mu \{a_{ij}{}_\lambda c_{ln}\} \}
-\{ \{a_{ij}{}_\lambda b_{hk}\}_{\lambda+\mu} c_{ln}\}\\
=\sum_{p,q\in\mb Z_+}\left(\left(h_{pq}'\right)_{lj}\left(h_{pq}''\right)_{ik}\left(h_{pq}'''\right)_{hn}
-\left(g_{pq}'\right)_{lk}\left(g_{pq}''\right)_{hj}\left(g_{pq}'''\right)_{in}\right)\lambda^p\mu^q\,.
\end{split}
\end{align}
\end{proposition}
\begin{proof}
Expanding in powers of $\lambda$ and $\mu$ equation \eqref{eq:triple1} using \eqref{notation}, we get
\begin{align*}
\ldb a_\lambda b_\mu c\rdb
&=\sum_{p,q\in\mb Z_+}
\left( \left(a_p(b_qc)'\right)'\otimes\left(a_p(b_qc)'\right)''\otimes\left(b_qc\right)''\right.\\
&\left.-\left(a_pc\right)'\otimes\left(b_q(a_pc)''\right)'\otimes\left(b_q(a_pc)''\right)''
\right)\lambda^p\mu^q\\
&-\sum_{p,q\in\mb Z_+}\left(
\left((a_pb)'_qc\right)'\otimes(\lambda+\mu+\partial)^q\left(a_pb\right)''\otimes\left((a_pb)'_qc\right)''
\right)\lambda^p\,.
\end{align*}
Exchanging $a$ with $b$ and $\lambda$ with $\mu$,
we get
\begin{align*}
\ldb b_\mu a_\lambda c\rdb
&=\sum_{p,q\in\mb Z_+}
\left( \left(b_q(a_pc)'\right)'\otimes\left(b_q(a_pc)'\right)''\otimes\left(a_pc\right)''\right.\\
&\left.-\left(b_qc\right)'\otimes\left(a_p(b_qc)''\right)'\otimes\left(a_p(b_qc)''\right)''
\right)\lambda^p\mu^q\\
&+\sum_{p,q\in\mb Z_+}\left(
\left((a_pb)''_qc\right)'\otimes(\lambda+\mu+\partial)^q\left(a_pb\right)'\otimes\left((a_pb)''_qc\right)''
\right)\lambda^p\,,
\end{align*}
where in the last term we used skewsymmetry \eqref{eq:skew2} and Lemma \ref{20130916:lem1}(a).
From the above equations it follows that
\begin{align}
\begin{split}\label{20130921:eq3}
RHS\eqref{20130921:eq2}
&=\sum_{p,q\in\mb Z_+}
\left( \left(a_p(b_qc)'\right)_{lj}'\left(a_p(b_qc)'\right)_{ik}''\left(b_qc\right)_{hn}''\right.\\
&\left.-\left(a_pc\right)'_{lj}\left(b_q(a_pc)''\right)_{ik}'\left(b_q(a_pc)''\right)_{hn}''
\right)\lambda^p\mu^q\\
&-\sum_{p,q\in\mb Z_+}
\left(
\left((a_pb)'_qc\right)_{lj}'\left((\lambda+\mu+\partial)^q\left(a_pb\right)_{ik}''\right)
\left((a_pb)'_qc\right)''_{hn}\right)\lambda^p\\
&-\sum_{p,q\in\mb Z_+}
\left( \left(b_q(a_pc)'\right)_{lk}'\left(b_q(a_pc)'\right)_{hj}''\left(a_pc\right)_{in}''\right.\\
&\left.-\left(b_qc\right)'_{lk}\left(a_p(b_qc)''\right)_{hj}'\left(a_p(b_qc)''\right)_{in}''\right)\lambda^p\mu^q\\
&-\sum_{p,q\in\mb Z_+}
\left(
\left((a_pb)''_qc\right)'_{lk}\left((\lambda+\mu+\partial)^q\left(a_pb\right)_{hj}'\right)
\left((a_pb)''_qc\right)''_{in}\right)\lambda^p\,.
\end{split}
\end{align}
On the other hand, using \eqref{20130917:eq1} and left and right Leibniz rules \eqref{eq:0.4}
and \eqref{eq:0.5}, we get
\begin{align}
\begin{split}\label{20130921:eq4.1}
&\{ a_{ij}{}_\lambda \{b_{hk}{}_\mu c_{ln}\} \}
=\sum_{p,q\in\mb Z_+}
\{a_{ij}{}_{\lambda}\left(b_qc\right)_{lk}'\left(b_qc\right)''_{hn}\}\mu^q\\
&=\sum_{p,q\in\mb Z_+}
\left( \left(a_p(b_qc)'\right)_{lj}'\left(a_p(b_qc)'\right)_{ik}''\left(b_qc\right)_{hn}''\right.\\
&\left.\qquad\qquad
-\left(b_qc\right)'_{lk}\left(a_p(b_qc)''\right)_{hj}'\left(a_p(b_qc)''\right)_{in}''\right)\lambda^p\mu^q\,;
\end{split}\\
\begin{split}\label{20130921:eq4.2}
&-\{ b_{hk}{}_\mu \{a_{ij}{}_\lambda c_{ln}\} \}
=-\sum_{p,q\in\mb Z_+}
\{b_{hk}{}_{\mu}\left(a_pc\right)_{lj}'\left(a_pc\right)''_{in}\}\lambda^p\\
&=-\sum_{p,q\in\mb Z_+}
\left( \left(b_q(a_pc)'\right)_{lk}'\left(b_q(a_pc)'\right)_{hj}''\left(a_pc\right)_{in}''\right.\\
&\left.\qquad\qquad
+\left(a_pc\right)'_{lj}\left(b_q(a_pc)''\right)_{ik}'\left(b_q(a_pc)''\right)_{hn}''
\right)\lambda^p\mu^q\,;
\end{split}\\
\begin{split}\label{20130921:eq4.3}
&-\{ \{a_{ij}{}_\lambda b_{hk}\}_{\lambda+\mu} c_{ln}\}
=-\sum_{p,q\in\mb Z_+}
\{\left(a_pb\right)_{hj}'\left(a_pb\right)_{ik}''{}_{\lambda+\mu}c_{ln}\}\lambda^p\\
&=-\sum_{p,q\in\mb Z_+}
\left(\left((a_pb)'_qc\right)'_{lj}\left((a_pb)'_qc\right)''_{hn}
(\lambda+\mu+\partial)^q\left(a_pb\right)''_{ik}\right.\\
&\left.\qquad\qquad
+\left((a_pb)''_qc\right)'_{lk}\left((a_pb)''_qc\right)''_{in}
(\lambda+\mu+\partial)^q\left(a_pb\right)'_{hj}
\right)\lambda^p\,.
\end{split}
\end{align}
Adding up \eqref{20130921:eq4.1}, \eqref{20130921:eq4.2} and \eqref{20130921:eq4.3}
and using the commutativity of $\mc V_m$
we see that the RHS of \eqref{20130921:eq1} is equal to \eqref{20130921:eq3},
thus concluding
the proof.
\end{proof}
\begin{theorem}\label{20130921:cor1}
If $\mc V$ is a double Poisson vertex algebra, then, for any positive integer $m$,
$\mc V_m$ is a Poisson vertex
algebra with $\lambda$-bracket defined by \eqref{20130917:eq1}.
\end{theorem}
\begin{proof}
By Proposition \ref{20131217:prop1}, formula \eqref{20130917:eq1}
gives a well-defined skewsymmetric $\lambda$-bracket on $\mc V_m$.
Since $\mc V$ is a double Poisson vertex algebra,
by Definition \ref{20140606:def-2}, $\ldb a_\lambda b_\mu c\rdb=0$, for all
$a,b,c\in\mc V$.
Hence, by Proposition \ref{20131217:prop2},
the RHS of \eqref{20130921:eq2} is equal to zero, 
for all $a,b,c\in\mc V$,
proving that the $\lambda$-bracket
defined by \eqref{20130917:eq1} satisfies the Jacobi identity \eqref{eq:0.5}.
\end{proof}
%

\subsection{Examples: affine and AGD double PVA structures}

\subsubsection{Affine Poisson vertex algebra for
\texorpdfstring{$\mf g=\mf{gl}_m$}{g=glm}}\label{sec:affine}
Let us consider the differential algebra
$\mc V=\mc R_1=\mb F\langle u,u',u'',\ldots\rangle$ and let
\begin{equation}\label{20130923:eq1}
\ldb u_{\lambda}u\rdb=1\otimes u-u\otimes1+c(1\otimes1)\lambda
\in\mc V^{\otimes2}[\lambda]\,,
\end{equation}
where $c\in\mb F$.
It is obvious that the skewsymmetry \eqref{eq:skew2} holds for the pair $u,u$, and
it is easy to check that the Jacobi identity \eqref{eq:jacobi2} holds
for the triple $u,u,u$. Hence, by Proposition \ref{20130921:prop1}, we have
a family of compatible double Poisson vertex algebra structures on $\mc V$,
uniquely extending \eqref{20130923:eq1}.
By Theorem \ref{20130921:cor1}, we get Poisson vertex algebra structures on
the commutative differential algebra
$\mc V_m=\mb F[u_{ij}^{(n)}\mid i,j=1,\ldots,m,n\in\mb Z_+]$
(see Example \ref{20131217:exa1}).
Using \eqref{20130923:eq1} and \eqref{20130917:eq1}, we get the explicit formula
for the $\lambda$-bracket among the generators of $\mc V_m$:
$$
\{u_{ij}{}_{\lambda}u_{hk}\}=\delta_{jh}u_{ik}-\delta_{ki}u_{hj}
+\delta_{jh}\delta_{ik}c\lambda\,,
$$
for all $i,j,h,k=1,\ldots,m$.
This is the $\lambda$-bracket among the generators of the affine Poisson vertex
algebra for $\mf{gl}_m$ (where we fixed the nondegenerate invariant bilinear 
form to be a scalar multiple of the trace form).
\begin{example}
Consider the differential algebra 
$\mc R_2=\mb F\langle u,v,u',v',\dots\rangle$ with the double lambda bracket 
$\ldb u_\lambda u\rdb$, given by \eqref{20130923:eq1}
with $v$ central. 
Then a similar computation as the one in Section \ref{sec:2.6} shows that
$h_0=1,\, h_n=\frac1n (u+v)^n$ for $n>0$, are in involution, 
and the corresponding integrable hierarchy of Hamiltonian equations is:
$$
\frac{du}{dt_n} =v(u+v)^n u -u(u+v)^n v +c\partial (u+v)^{n+1}
\,\,,\,\,\,\,
dv/dt_n=0
\,,\,\, n\in \mb Z_+
\,.
$$
\end{example}

\subsubsection{Adler type non-commutative pseudodifferential operators}
\label{sec:AGD}
Let us recall from \cite{Kac98} that the
\emph{$\delta$-function} is, by definition,
the $\mb F$-valued formal distribution
$$
\delta(z-w)=\sum_{n\in\mb Z_+}z^nw^{-n-1}
\in\mb F[[z,z^{-1},w,w^{-1}]]
\,.
$$
Let $a(z)=\sum_{i\in\mb Z}a_iz^i$ be a formal Laurent series in $z^{-1}$
with values in some vector space $\mc V$. Then, we have
\begin{equation}\label{eq:delta_prop}
a(z)\delta(z-w)=a(w)\delta(z-w)
\,.
\end{equation}
We denote by $i_z$ the power series expansion
for large $|z|$. For example,
$$
i_z(z-w)^{-1}
=\sum_{k\in\mb Z_+}z^{-k-1}w^k
\,,
\quad
i_z(z-w-\lambda-\partial)^{-1}
=\sum_{k\in\mb Z_+}z^{-k-1}(w+\lambda+\partial)^k
\,.
$$
Using this notation, the $\delta$-function can be rewritten as follows:
\begin{equation}\label{eq:delta}
\delta(z-w)=i_z(z-w)^{-1}-i_w(z-w)^{-1}
\,.
\end{equation}
In the sequel, we use the following notation:
if $a(z)=\sum_{n\in\mb Z}a_nz^n$ and $b(z)$ are two formal 
Laurent series in $z^{-1}$, then
$$
a(z+\lambda+\partial)\otimes b(z)
=\sum_{n\in\mb Z} a_n\otimes i_z(z+\lambda+\partial)^nb(z)\,,
$$
namely, for any $n\in\mb Z$, we expand $(z+\lambda+\partial)^n$
in non-negative
powers of $\lambda+\partial$
and we let powers of $\partial$ act to the right, on the
coefficients of $b(z)$.
%
In order to emphasize that $\partial$ does not act on some factors,
we enclose in parenthesis the terms on which $\partial$ does act.

We denote by $\res_z$ the coefficient of $z^{-1}$.
The following identity holds for every formal power series $a(z)=\sum_n\in\mb Z a_nz^n$,
\begin{equation}\label{valentinesday}
\res_z a(z)i_z(z-w)^{-1}=a(w)_+
\,,
\end{equation}
where $a(w)_+=\sum_{n=0}^\infty a_nz^n$ is the positive part of $a(z)$.
Furthermore, 
the following identity is a consequence of integration by parts:
\begin{equation}\label{star}
\res_z(\iota_z a(z+t)\otimes b(z))=\res_z(a(z)\otimes \iota_z  b(z-t))
\,.
\end{equation}
\begin{definition}\label{20131217:def1}
Let $\mc V$ be a differential algebra endowed with a $2$-fold $\lambda$-bracket
$\ldb-_\lambda-\rdb$.
We call a pseudodifferential operator $L(\partial)\in\mc V((\partial^{-1}))$
of \emph{Adler type} for $\ldb-_\lambda-\rdb$ if the following identity holds in
$\mc V^{\otimes2}[\lambda,\mu]((z^{-1},w^{-1}))$
(cf. \cite{DSKV14a}):
\begin{equation}\label{20130923:eq2}
\begin{array}{c}
\displaystyle{
\ldb L(z)_{\lambda}L(w)\rdb
=L(z)\otimes i_z(z-w-\lambda-\partial)^{-1}L(w)
}
\\
\displaystyle{
-L(w+\lambda+\partial)\otimes i_z(z-w-\lambda-\partial)^{-1}L^*(-z+\lambda)
\,,
}
\end{array}
\end{equation}
where $L^*(\partial)$ is the formal adjoint of the pseudodifferential operator $L(\partial)$,
and $L^*(z)$ is its symbol.
\end{definition}
\begin{proposition}\label{20130923:prop1}
Let $\mc V$ be a differential algebra, let $\ldb-_\lambda-\rdb$
be a $2$-fold $\lambda$-bracket on $\mc V$, and let $L(\partial)\in\mc V((\partial^{-1}))$
be an Adler type pseudodifferential operator. Then:
\begin{enumerate}[(a)]
\item The following identity holds in $\mc V^{\otimes2}[\lambda]((z^{-1},w^{-1}))$:
$$
\ldb L(z)_{\lambda}L(w)\rdb
=-\ldb L(w)_{-\lambda-\partial}L(z)\rdb^\sigma\,.
$$
\item
The following identity holds in $\mc V^{\otimes3}[\lambda,\mu]((z_1^{-1},z_2^{-1},z_3^{-1}))$:
$$
\begin{array}{c}
\ldb L(z_1)_{\lambda}\ldb L(z_2)_{\mu}L(z_3)\rdb\rdb_L
-\ldb L(z_2)_{\mu}\ldb L(z_1)_{\lambda}L(z_3)\rdb\rdb_R\\
=\ldb\ldb L(z_1)_{\lambda}L(z_2)\rdb_{\lambda+\mu}L(z_3)\rdb_L
\,.
\end{array}
$$
\end{enumerate}
\end{proposition}
\begin{proof}
By \eqref{20130923:eq2} we get
$$
\begin{array}{c}
\ldb L(w)_{-\lambda-\partial}L(z)\rdb^\sigma
=L(w+\lambda+\partial)\otimes i_w(z-w-\lambda-\partial)^{-1}L^*(-z+\lambda)\\
-L(z)\otimes i_w(z-w-\lambda-\partial)^{-1}L(w)\,.
\end{array}
$$
Hence, in order to prove (a) we are left to show that
$$
\begin{array}{l}
L(z)\otimes i_z(z-w-\lambda-\partial)^{-1}L(w)\\
-L(w+\lambda+\partial)\otimes i_z(z-w-\lambda-\partial)^{-1}L^*(-z+\lambda)\\
=L(z)\otimes i_w(z-w-\lambda-\partial)^{-1}L(w)\\
-L(w+\lambda+\partial)\otimes i_w(z-w-\lambda-\partial)^{-1}L^*(-z+\lambda)\,,
\end{array}
$$
which can be rewritten, using equation \eqref{eq:delta}, as
\begin{equation}\label{20130923:toprove1}
L(z)\otimes\delta(z-w-\lambda-\partial)L(w)
=L(w+\lambda+\partial)\otimes\delta(z-w-\lambda-\partial)L^*(-z+\lambda)
\,.
\end{equation}
Identity \eqref{20130923:toprove1} holds by applying
the property of the
$\delta$-function given by equation \eqref{eq:delta_prop}.

In order to prove part (b) let us compute the three terms separately
using \eqref{notation} and \eqref{20130923:eq2}. We have
%
%
\begin{align}
&\ldb L(z_1)_{\lambda}\ldb L(z_2)_{\mu}L(z_3)\rdb\rdb_L\notag\\
&\label{eq:a}
=L(z_1)
\otimes
\left(i_{z_{1}}(z_1-z_2-\lambda-\partial)^{-1}L(z_2)
\right)
\otimes
i_{z_2}(z_2-z_3-\mu-\partial)^{-1}L(z_3)\\
\begin{split}\label{eq:b}
&-\left(
L(z_2+\lambda+\partial)
\otimes
i_{z_1}(z_1-z_2-\lambda-\partial)^{-1}L^*(-z_1+\lambda)
\right)\\
&\qquad\qquad\otimes
i_{z_2}(z_2-z_3-\mu-\partial)^{-1}L(z_3)
\end{split}\\
\begin{split}\label{eq:c}
&-L(z_1)
\otimes
i_{z_1}(z_1-z_3-\lambda-\mu-\partial)^{-1}L(z_3+\mu+\partial)\\
&\qquad\qquad\otimes
i_{z_2}(z_2-z_3-\mu-\partial)^{-1}L^*(-z_2+\mu)
\end{split}\\
\begin{split}\label{eq:d}
&+L(z_3+\lambda+\mu+\partial)
\otimes
i_{z_1}(z_1-z_3-\lambda-\mu-\partial)^{-1}L^*(-z_1+\lambda)\\
&\qquad\qquad\otimes
i_{z_2}(z_2-z_3-\mu-\partial)^{-1}L^*(-z_2+\mu)\,;
\end{split}
\end{align}
%
%
\begin{align}
&-\ldb L(z_2)_{\mu}\ldb L(z_1)_{\lambda}L(z_3)\rdb\rdb_R\notag\\
&\label{eq:a'}
=-L(z_1)
\otimes
i_{z_{1}}(z_1-z_3-\lambda-\mu-\partial)^{-1}L(z_2)
\otimes
i_{z_2}(z_2-z_3-\mu-\partial)^{-1}L(z_3)\\
\begin{split}\label{eq:b'}
&+L(z_1)
\otimes
i_{z_1}(z_1-z_3-\lambda-\mu-\partial)^{-1}L(z_3+\mu+\partial)\\
&\qquad\qquad\otimes
i_{z_2}(z_2-z_3-\mu-\partial)^{-1}L^*(-z_2+\mu)
\end{split}\\
\begin{split}\label{eq:c'}
&-L(z_3+\lambda+\mu+\partial)
\otimes
i_{z_1}(z_1-z_3-\lambda-\mu-\partial)^{-1}\\
&\qquad\qquad\times\left(
L^*(\lambda+\mu+\partial-z_1)
\otimes
i_{z_2}(z_1-z_2-\lambda-\partial)^{-1}L(z_2)
\right)^\sigma
\end{split}\\
\begin{split}
&\label{eq:d'}
+L(z_3+\lambda+\mu+\partial)
\otimes
i_{z_1}(z_1-z_3-\lambda-\mu-\partial)^{-1}\\
&\qquad\qquad\times\left(
L^*(-z_2+\mu)
\otimes
i_{z_2}(z_1-z_2-\lambda-\partial)^{-1}L^*(-z_1+\lambda)
\right)^\sigma\,;
\end{split}
\end{align}
%
%
\begin{align}
&\ldb\ldb L(z_1)_{\lambda}L(z_2)\rdb_{\lambda+\mu}L(z_3)\rdb_L
\notag\\
&\label{eq:a''}
=L(z_1)
\otimes
i_{z_1}(z_1-z_3-\lambda-\mu-\partial)^{-1}
\left(
L(z_3)
\otimes
i_{z_1}(z_1-z_2-\lambda-\partial)^{-1}L(z_2)
\right)^\sigma\\
\begin{split}\label{eq:b''}
&-L(z_3+\lambda+\mu+\partial)
\otimes
i_{z_1}(z_1-z_3-\lambda-\mu-\partial)^{-1}\\
&\qquad\qquad\times\left(
L^*(\lambda+\mu+\partial-z_1)
\otimes
i_{z_1}(z_1-z_2-\lambda-\partial)^{-1}L(z_2)
\right)^\sigma
\end{split}\\
\begin{split}\label{eq:c''}
&-\left(
L(z_2+\lambda+\partial)
\otimes
i_{z_1}(z_1-z_2-\lambda-\partial)^{-1}L^*(-z_1+\lambda)
\right)\\
&\qquad\qquad\otimes
i_{z_2}(z_2-z_3-\mu-\partial)^{-1}L(z_3)
\end{split}\\
\begin{split}\label{eq:d''}
&+L(z_3+\lambda+\mu+\partial)
\otimes
\left(
i_{z_1}(z_1-z_2-\lambda-\partial)^{-1}L^*(-z_1+\lambda)
\right)\\
&\qquad\qquad\otimes
i_{z_2}(z_2-z_3-\mu-\partial)^{-1}L^*(-z_2+\mu)\,.
\end{split}
\end{align}

Note that
$$
\eqref{eq:b}=\eqref{eq:c''}
\quad\text{and}\quad
\eqref{eq:c}+\eqref{eq:b'}=0\,.
$$
Next, we claim that
$$
\eqref{eq:a}+\eqref{eq:a'}=\eqref{eq:a''}\,.
$$
Indeed, we can rewrite the above identity as
\begin{align}
\begin{split}\label{20130923:eq3}
&i_{z_1}(z_1-z_2-x)^{-1}
i_{z_2}(z_2-z_3-y)^{-1}
L(z_1)\otimes (\big|_{x=\lambda+\partial}L(z_2))\otimes (\big|_{y=\mu+\partial}L(z_3))\\
&=i_{z_1}(z_1-z_3-x-y)^{-1}
\left(i_{z_1}(z_1-z_2-x)^{-1}
\right.\\
&\left.\qquad\qquad +i_{z_2}(z_2-z_3-y)^{-1}
\right)
L(z_1)\otimes(\big|_{x=\lambda+\partial}L(z_2))\otimes (\big|_{y=\mu+\partial}L(z_3))\,.
\end{split}
\end{align}
(Here and further, for a Laurent series $P(z)=\sum_{n=-\infty}^Nc_nz^n\in\mc V((z^{-1}))$
and elements $a,b\in\mc V$,
we let
$$
a\Big(\Big|_{x=\nu+\partial}P(z+x)b\Big)
=
\sum_{n=-\infty}^Nai_z(z+\nu+\partial)^n(c_nb)
\,.)
$$
The identity \eqref{20130923:eq3} follows from the obvious identity
$$
a^{-1}b^{-1}=i_a(a+b)^{-1}b^{-1}+i_a(a+b)^{-1}a^{-1}\,.
$$
In order to prove part b) we are left to show that
\begin{equation}\label{eq:final}
\eqref{eq:d}+\eqref{eq:c'}+\eqref{eq:d'}=\eqref{eq:b''}+\eqref{eq:d''}\,.
\end{equation}

Note that, using the definition and the properties of the $\delta$-function
we have
\begin{align}
\begin{split}\label{20130923:eq4}
&\eqref{eq:c'}-\eqref{eq:b''}
\\
&=L(z_3+\lambda+\mu+\partial)
\otimes
i_{z_1}(z_1-z_3-\lambda-\mu-\partial)^{-1}
\\
&\qquad\times\left(
L^*(\lambda+\mu+\partial-z_1)
\otimes
\delta(z_1-z_2-\lambda-\partial)L(z_2)
\right)^\sigma\\
&=L(z_3+\lambda+\mu+\partial)
\otimes
i_{z_1}(z_1-z_3-\lambda-\mu-\partial)^{-1}
\\
&\qquad\times\left(
L^*(-z_2+\mu)
\otimes
\delta(z_1-z_2-\lambda-\partial)L^*(-z_1+\lambda)
\right)^\sigma\,;
\end{split}
\\
\begin{split}\label{20130923:eq5}
&\eqref{eq:d'}
\\
&=L(z_3+\lambda+\mu+\partial)
\otimes
i_{z_1}(z_1-z_3-\lambda-\mu-\partial)^{-1}
\\
&\qquad\times
\left(
L^*(-z_2+\mu)
\otimes
i_{z_1}(z_1-z_2-\lambda-\partial)^{-1}L^*(-z_1+\lambda)
\right)^\sigma
\\
&-L(z_3+\lambda+\mu+\partial)
\otimes
i_{z_1}(z_1-z_3-\lambda-\mu-\partial)^{-1}
\\
&\qquad\times\left(
L^*(-z_2+\mu)
\otimes
\delta(z_1-z_2-\lambda-\partial)L^*(-z_1+\lambda)
\right)^\sigma\,.
\end{split}
\end{align}
Using \eqref{20130923:eq4} and \eqref{20130923:eq4}, identity
\eqref{eq:final}
is proved once we show that
\begin{align*}
&i_{z_1}(z_1-z_2-x)^{-1}
i_{z_2}(z_2-z_3-y)^{-1}
\\
&\qquad\times
L(z_3+x+y)\otimes (\big|_{x=\lambda+\partial}L^*(-z_1+\lambda))\otimes (\big|_{y=\mu+\partial}L^*(-z_2+\mu))
\\
&=i_{z_1}(z_1-z_3-x-y)^{-1}\left(i_{z_1}(z_1-z_2-x)^{-1}
+i_{z_2}(z_2-z_3-y)^{-1}\right)
\times
\\
&
\times
L(z_3+x+y)\otimes (\big|_{x=\lambda+\partial}L^*(-z_1+\lambda))\otimes (\big|_{y=\mu+\partial}L^*(-z_2+\mu))\,,
\end{align*}
which can be proved as we did in the previous claim.
\end{proof}
\begin{theorem}\label{subalgebra}
Let $\mc V$ be a differential algebra, endowed with a
$2$-fold $\lambda$-bracket $\ldb-_\lambda-\rdb$.
Let $L(\partial)\in\mc V((\partial^{-1}))$ be a pseudodifferential operator,
and let $\mc U\subset\mc V$ be the differential subalgebra 
generated by the coefficients of $L(\partial)$.
If $L(\partial)$ is of Adler type,
then $\ldb-_\lambda-\rdb$ restricts to a double PVA $\lambda$-bracket on $\mc U$.
\end{theorem}
\begin{proof}
By \eqref{20130923:eq2} $\ldb-_\lambda-\rdb$ restricts to a $2$-fold $\lambda$-bracket of $\mc U$,
and by Proposition \ref{20130923:prop1}, together with Proposition
\ref{20130921:prop1}, we have that $\mc U$ is a double PVA.
\end{proof}

\subsubsection{AGD Poisson vertex algebra structures for
generic matrix (pseudo)differential operators
of order \texorpdfstring{$N$}{N}}

Let $\mc V$ be the algebra of non-commutative 
differential polynomials in the variables $u_i$, $i\in I$, namely
$$
\mc V=\mc R_I=\mb F\langle u_i^{(n)}\mid i\in I,n\in\mb Z_+\rangle\,.
$$
The index set $I$ will be either $I=\{i\in\mb Z\mid i\geq-N\}$
or $I_-=\{-N,-N+1,\dots,-1\}$.
Let us collect the differential generators of $\mc V$ into
the generating series
$$
L(z)=z^N+\sum_{i\in I}u_{i}z^{-i-1}\in\mc V((z^{-1}))\,.
$$
By Example \ref{20131217:exa1} we have that
$\mc V_m=\mb F[u_{i,ab}\mid i\in I,a,b\in\{1,\dots,m\},n\in\mb Z_+]$,
for every $m\geq1$.
Let us denote
$$
L_{ab}(z)=z^N+\sum_{i\in I}u_{i,ab}z^{-i-1}\in\mc V_m((z^{-1}))\,,
$$
for every $a,b=1,\dots,m$.
\begin{corollary}\label{cor:bi_dPVA}
$\mc V$ has two compatible double Poisson vertex algebra structures
with $2$-fold
$\lambda$-brackets defined on generators by the following generating series:
\begin{equation}\label{eq:H}
\begin{array}{c}
\displaystyle{
\ldb L(z)_{\lambda}L(w)\rdb_H
=L(z)\otimes i_z(z-w-\lambda-\partial)^{-1}L(w)
}
\\
\displaystyle{
-L(w+\lambda+\partial)\otimes i_z(z-w-\lambda-\partial)^{-1}L^*(-z+\lambda)
}
\end{array}
\end{equation}
and
\begin{equation}\label{eq:K}
\begin{array}{c}
\displaystyle{
\ldb L(z)_{\lambda}L(w)\rdb_K
=i_z(z-w-\lambda)^{-1}(L(z)-L(w+\lambda))\otimes1
}
\\
\displaystyle{
+1\otimes i_z(z-w-\lambda-\partial)^{-1}(L(w)-L^*(-z+\lambda))\,.
}
\end{array}
\end{equation}
Furthermore, for all $m\geq1$,
$\mc V_m$ has two compatible Poisson vertex algebra structures,
with $\lambda$-brackets on generators
given by the generating series
$$
\begin{array}{c}
\{L_{ab}(z)_{\lambda}L_{cd}(w)\}
=L_{cb}(z)i_z(z-w-\lambda-\partial)^{-1}L_{ad}(w)\\
-L_{cb}(w+\lambda+\partial)i_z(z-w-\lambda-\partial)^{-1}L_{ad}^*(\lambda-z)
\end{array}
$$
and
$$
\begin{array}{c}
\{L_{ab}(z)_\lambda L_{cd}(w)\}_K
=\delta_{ad}i_z(z-w-\lambda)^{-1}\left(L_{cb}(z)-L_{cb}(w+\lambda)\right)
\\
+\delta_{cb}i_z(z-w-\lambda-\partial)^{-1}\left(L_{ad}(w)-L_{ad}^*(\lambda-z)\right)\,,
\end{array}
$$
for all $a,b,c,d=1,\ldots,m$.
\end{corollary}
\begin{proof}
The fact that $\mc V$ has two compatible double Poisson vertex algebra 
structures follows by Theorem \ref{subalgebra} applied to the
pseudodifferential operator $L-c$, $c\in\mb F$.
The second claim follows by Theorem \ref{20130921:cor1}.
\end{proof}
\begin{remark}
The compatible PVA structures on $\mc V_m$ are
the same as the one appearing in \cite[Eq. (4.7)]{DSKV14a}.
\end{remark}

The following results will be used in Section \ref{sec:non-local}.
\begin{lemma}\label{20140311:lem1}
In the double Poisson vertex algebra $\mc V$,
with $\lambda$-bracket defined by \eqref{eq:H}, we have:
\begin{enumerate}[(a)]
\item $\ldb u_{-N}{}_\lambda L(w)\rdb_H=1\otimes L(w)-L(w+\lambda)\otimes1$;
\item $\ldb L(z)_\lambda u_{-N}\rdb_H=1\otimes L^*(-z+\lambda)-L(z)\otimes1$;
\item $\ldb u_{-N}{}_\lambda u_{-N}\rdb_H=1\otimes u_{-N}-u_{-N}\otimes1
-(1\otimes1)N\lambda$.
\end{enumerate}
\end{lemma}
\begin{proof}
Same proof as Lemma 2.17 in \cite{DSKV14a}.
\end{proof}
The following results follow immediately from the previous lemma and
equation \eqref{20130917:eq1}.
\begin{corollary}
For $m\geq1$, in the Poisson vertex algebra $\mc V_m$ we have
($a,b,c,d=1,\dots,m$):
\begin{enumerate}[(a)]
\item $\{ u_{-N,ab}{}_\lambda L_{cd}(w)\}
=\delta_{cb}L_{ad}(w)-\delta_{ad}L_{cb}(w+\lambda)$;
\item $\{ L_{ab}(z)_\lambda u_{-N,cd}\}
=\delta_{cb}L_{ad}^*(-z+\lambda)-\delta_{ad}L_{cb}(z)$;
\item $\{ u_{-N,ab}{}_\lambda u_{-N,cd}\}
=\delta_{cb}u_{-N,ad}-\delta_{ad}u_{-N,cb}-\delta_{ad}\delta_{cb}N\lambda$.
\end{enumerate}
\end{corollary}

\section{Non-local double Poisson vertex algebras and their Dirac reduction}

\subsection{Non-local double Poisson vertex algebras}\label{sec:non-local}
The notion of a non-local double PVA is similar to that of a non-local PVA, introduced in
\cite{DSK13}.
As in that paper, given a vector space
$V$, we denote
$$
V_{\lambda,\mu}:=V[[\lambda^{-1},\mu^{-1},(\lambda+\mu)^{-1}]][\lambda,\mu]\,,
$$
namely, the quotient of the $\mb F[\lambda,\mu,\nu]$-module
$V[[\lambda^{-1},\mu^{-1},\nu^{-1}]][\lambda,\mu,\nu]$
by the submodule 
$(\nu-\lambda-\mu)V[[\lambda^{-1},\mu^{-1},\nu^{-1}]][\lambda,\mu,\nu]$.
Recall that we have the natural embedding 
$i_{\mu}:\,V_{\lambda,\mu}\hookrightarrow V((\lambda^{-1}))((\mu^{-1}))$
(see the beginning of Section \ref{sec:AGD}).

Let $\mc V$ be a differential algebra with a derivation $\partial$.
A \emph{(non-local) $2$-fold} $\lambda$-\emph{bracket} on $\mc V$ is a linear map
$\ldb-_\lambda-\rdb:\,\mc V^{\otimes2}\to \mc V^{\otimes2}((\lambda^{-1}))$
satisfying the \emph{sesquilinearity} conditions \eqref{20140702:eq4b}
and the left and right \emph{Leibniz rules} \eqref{20140702:eq6b}.
Here and further an expression $\ldb a_{\lambda+\partial}b\rdb_\to\star_1 c$ is interpreted as follows:
if $\ldb a_{\lambda}b\rdb=\sum_{n=-\infty}^N(h_n'\otimes h_n'')\lambda^n$, 
then $\ldb a_{\lambda+\partial}b\rdb_\to\star_1 c
=\sum_{n=-\infty}^Nh_n'\left((\lambda+\partial)^nc\right)\otimes h_n''$,
where we expand $(\lambda+\partial)^n$ in non-negative powers of $\partial$
acting on $c$.
The (non-local) $2$-fold $\lambda$-bracket $\ldb \cdot\,_\lambda\,\cdot\rdb$ 
is called \emph{skewsymmetric} if equation \eqref{eq:skew2} holds.
The RHS of the skewsymmetry condition should be interpreted as follows:
we move $-\lambda-\partial$ to the left and
we expand its powers in non-negative powers of $\partial$,
acting on the coefficients on the $2$-fold $\lambda$-bracket.
We say that the (non-local) $2$-fold $\lambda$-bracket $\ldb-_\lambda-\rdb$ 
is \emph{admissible} if
\begin{equation}\label{20110921:eq4}
\ldb a_\lambda\ldb b_\mu c\rdb\rdb_L\in\mc V^{\otimes3}_{\lambda,\mu}
\qquad\text{for all }a,b,c\in\mc V\,.
\end{equation}
Here we are identifying the space $\mc V_{\lambda,\mu}^{\otimes3}$
with its image in $\mc V^{\otimes3}((\lambda^{-1}))((\mu^{-1}))$ via the embedding $i_{\mu
}$.
\begin{remark}
If $\ldb-_\lambda-\rdb$ is a skewsymmetric admissible 
(non-local) $2$-fold $\lambda$-bracket on $\mc V$,
by equation \eqref{20140707:eq2} we have also
$\ldb b_\mu\ldb a_\lambda c\rdb\rdb_R\in\mc V_{\lambda,\mu}^{\otimes3}$,
since $\mc V^{\otimes3}_{\lambda,\mu}=\mc V^{\otimes3}_{\mu,\lambda}$, 
and similarly 
$\ldb\ldb a_\lambda b\rdb_{\lambda+\mu} c\rdb_L\in\mc V_{\lambda,\mu}^{\otimes3}$,
for all $a,b,c\in\mc V$.
\end{remark}
\begin{definition}\label{20130513:def}
A \emph{non-local double Poisson vertex algebra} is a differential algebra $\mc V$
endowed with a non-local $2$-fold $\lambda$-bracket,
$\ldb-_\lambda-\rdb:\,\mc V\otimes \mc V\to \mc V((\lambda^{-1}))$
satisfying
skewsymmetry \eqref{eq:skew2},
admissibility \eqref{20110921:eq4},
and 
Jacobi identity \eqref{eq:jacobi2},
where the latter is understood as an identity in the space $\mc V_{\lambda,\mu}^{\otimes3}$.
%
\end{definition}

\subsection{Dirac reduction for non-local double Poisson vertex algebras}\label{sec:dirac}
Dirac reduction for a double PVA is similar to that of a PVA, constructed in \cite{DSKV14a}.
Let $\mc V$ be a non-local double Poisson vertex algebra 
with $2$-fold $\lambda$-bracket $\ldb-_{\lambda}-\rdb$.
Let $\theta_1,\ldots,\theta_m$ be elements of $\mc V$
and let us consider the matrix pseudodifferential operator
$$
C(\partial)=(C_{\alpha\beta}(\partial))_{\alpha,\beta=1}^m
\in\Mat_{m\times m}\mc V^{\otimes2}((\partial^{-1}))
\,,
$$
where the symbol of the pseudodifferential operator $C_{\alpha\beta}(\partial)$
is
\begin{equation}\label{C}
C_{\alpha\beta}(\lambda)=\ldb\theta_{\beta}{}_{\lambda}\theta_{\alpha}\rdb\,.
\end{equation}
Given a pseudodifferential operator $A(\partial)=\sum_{n\leq N}A_n\partial^n
\in\Mat_{m\times m}\mc V^{\otimes 2}((\partial^{-1}))$, recall that we defined its adjoint by
$$
A(\partial)^\dagger=\sum_{n\leq N}(-\partial)^n\circ (A_n^t)^\sigma\,,
$$
where $A_n^t$ denotes the transpose matrix of $A_n$. By the skewsymmetry
condition \eqref{eq:skew2}, the pseudodifferential operator $C(\partial)$
is skewadjoint.
We shall assume that the matrix pseudodifferential operator $C(\partial)$ is invertible
with respect to the $\bullet$-product \eqref{20140609:eqc1}, 
and we denote its inverse by
$C^{-1}(\partial)=\big((C^{-1})_{\alpha\beta}(\partial)\big)_{\alpha,\beta=1}^m
\in\Mat_{m\times m}\mc V^{\otimes2}((\partial^{-1}))$.
\begin{definition}\label{20130514:def}
The \emph{Dirac modification} of the $2$-fold
$\lambda$-bracket $\ldb-_{\lambda}-\rdb$,
associated to the elements $\theta_1,\dots,\theta_m$,
is the map
$\ldb-_{\lambda}-\rdb^D:\,\mc V^{\otimes2}\to\mc V^{\otimes2}((\lambda^{-1}))$
given by ($a,b\in\mc V$):
\begin{equation}\label{dirac}
\ldb a_{\lambda}b\rdb^D
=\ldb a_{\lambda}b\rdb
-\sum_{\alpha,\beta=1}^m
\ldb{\theta_{\beta}}_{\lambda+\partial}b\rdb_{\to}
\bullet
(C^{-1})_{\beta\alpha}(\lambda+\partial)
\bullet
\ldb a_{\lambda}\theta_{\alpha}\rdb\,.
\end{equation}
\end{definition}
\begin{theorem}\label{prop:dirac}
Let $\mc V$ be a non-local double Poisson vertex algebra with $2$-fold $\lambda$-bracket $\ldb-_\lambda-\rdb$.
Let $\theta_1,\dots,\theta_m\in\mc V$ be  elements such that
the corresponding matrix pseudodifferential operator 
$C(\partial)=(C_{\alpha\beta}(\partial))_{\alpha,\beta=1}^m\in\Mat_{m\times m}\mc V^{\otimes2}((\partial^{-1}))$
given by \eqref{C} is invertible.
\begin{enumerate}[(a)]
\item
The Dirac modification $\ldb-_\lambda-\rdb^D$
given by equation \eqref{dirac}
is a $2$-fold $\lambda$-bracket on $\mc V$,
giving $\mc V$ a structure of a non-local double PVA.
\item
All the elements $\theta_i,\,i=1,\dots,m$, are central 
with respect to the Dirac modified $\lambda$-bracket:
$\ldb a_\lambda\theta_i\rdb^D=\ldb{\theta_i}_\lambda a\rdb^D=0$
for all $i=1,\dots,m$ and $a\in\mc V$.
\end{enumerate}
\end{theorem}
Before proving the theorem we state some lemmas that we will use in the
proof.
\begin{lemma}\label{20111006:lem}
Let $A(\lambda,\mu)\in\mc V^{\otimes3}_{\lambda,\mu}$ and
$B(\lambda,\mu)\in\mc V^{\otimes2}_{\lambda,\mu}$,
and let $S,T:\,\mc V^{\otimes3}\to \mc V^{\otimes3}$ be endomorphisms of
$\mc V^{\otimes3}$ (viewed as a vector space). 
Then
$$
\begin{array}{l}
A(\lambda+S,\mu+T)\bullet_1 B(\lambda,\mu)
\in\mc V^{\otimes3}_{\lambda,\mu}\,,\\
A(\lambda+S,\mu+T)\bullet_2 B(\lambda,\mu)
\in\mc V^{\otimes3}_{\lambda,\mu}\,,\\
A(\lambda+S,\mu+T)\bullet_3 B(\lambda,\mu)
\in\mc V^{\otimes3}_{\lambda,\mu}\,,
\end{array}
$$
where we expand the negative powers of $\lambda+S$ and $\mu+T$ 
in non-negative powers of $S$ and $T$, acting on the coefficients of $B$,
and
$$
\begin{array}{l}
B(\lambda+S,\mu+T)\bullet_1A(\lambda,\mu)
\in\mc V^{\otimes3}_{\lambda,\mu}\,,\\
B(\lambda+S,\mu+T)\bullet_2A(\lambda,\mu)
\in\mc V^{\otimes3}_{\lambda,\mu}\,,\\
B(\lambda+S,\mu+T)\bullet_3A(\lambda,\mu)
\in\mc V^{\otimes3}_{\lambda,\mu}\,,
\end{array}
$$
where now $S$ and $T$ act on the coefficients of $A$.
%
\end{lemma}
\begin{proof}
The proof follows the same lines as the proof of Lemma 2.3 in \cite{DSK13}
using the definition of the $\bullet$-products given by equation
\eqref{20140609:eqc2}.
\end{proof}
\begin{lemma}\label{20111012:lem}
Let $\ldb-_\lambda-\rdb:\,\mc V^{\otimes2}\times\mc V^{\otimes2}\to\mc V^{\otimes2}((\lambda^{-1}))$ 
be a $2$-fold $\lambda$-bracket on 
the differential algebra $\mc V$.
Suppose that 
$C(\partial)=\big(C_{ij}(\partial)\big)_{i,j=1}^\ell\in\Mat_{\ell\times\ell}\mc V^{\otimes2}((\partial^{-1}))$
is an invertible $\ell\times\ell$ matrix pseudodifferential operator with coefficients in $\mc V^{\otimes2}$,
and let
$C^{-1}(\partial)=\big((C^{-1})_{ij}(\partial)\big)_{i,j=1}^\ell
\in\Mat_{\ell\times\ell}\mc V^{\otimes2}((\partial^{-1}))$
be its inverse.
The following identities hold for every $a\in\mc V$ and $i,j=1,\dots,\ell$:
\begin{equation}\label{20111012:eq2aL}
\begin{array}{l}
\displaystyle{
\ldb a_\lambda (C^{-1})_{ij}(\mu)\rdb_L
=
-\sum_{r,t=1}^\ell
i_\mu(C^{-1})_{ir}(\lambda+\mu+\partial)
} \\
\displaystyle{
\qquad\qquad
\bullet_2
\ldb a_\lambda C_{rt}(y)\rdb_L\bullet_1\Big(\Big|_{y=\mu+\partial}(C^{-1})_{tj}(\mu)
\Big)
\,\in\mc V^{\otimes3}((\lambda^{-1}))((\mu^{-1}))
\,,
}
\end{array}
\end{equation}
\begin{equation}\label{20111012:eq2aR}
\begin{array}{l}
\displaystyle{
\ldb a_\lambda (C^{-1})_{ij}(\mu)\rdb_R
=
-\sum_{r,t=1}^\ell
i_\mu(C^{-1})_{ir}(\lambda+\mu+\partial)
} \\
\displaystyle{
\qquad\qquad
\bullet_2
\ldb a_\lambda C_{rt}(y)\rdb_R
\bullet_3
\Big(\Big|_{y=\mu+\partial}(C^{-1})_{tj}(\mu)
\Big)
\,\in\mc V^{\otimes3}((\lambda^{-1}))((\mu^{-1}))
\,,
}
\end{array}
\end{equation}
and
\begin{equation}\label{20111012:eq2b}
\begin{array}{l}
\displaystyle{
\ldb(C^{-1})_{ij}(\lambda) _{\lambda+\mu} a\rdb_L^{\sigma^2}
=
-\sum_{r,t=1}^\ell
\ldb C_{rt}(x)_{x+y}a\rdb_L^{\sigma^2}
\bullet_3
\Big(\Big|_{x=\lambda+\partial} (C^{-1})_{tj}(\lambda)\Big)
} \\
\qquad\qquad\qquad
\displaystyle{
\bullet_1
\Big(\Big|_{y=\mu+\partial}
i_\lambda({C^*}^{-1})_{ri}(\mu) \Big)
\,\in\mc V^{\otimes3}(((\lambda+\mu)^{-1}))((\lambda^{-1}))
\,,
}
\end{array}
\end{equation}
where
$i_\mu:\,\mc V^{\otimes3}_{\lambda,\mu}\to\mc V^{\otimes3}((\lambda^{-1}))((\mu^{-1}))$
and $i_\lambda:\,
\mc V^{\otimes3}_{\lambda,\mu}\to\mc V^{\otimes3}(((\lambda+\mu)^{-1}))((\lambda^{-1}))$ 
are the natural embeddings defined above.
In equations \eqref{20111012:eq2aL} and \eqref{20111012:eq2b},
$C(\lambda)\in\Mat_{\ell\times\ell}\mc V((\lambda^{-1}))$ 
denotes the symbol of the matrix pseudodifferential operator $C$
and $C^*$ denotes its adjoint (its inverse being $(C^{-1})^*$).
\end{lemma}
\begin{proof}
The identity $C\circ C^{-1}=1\otimes1$ becomes, in terms of symbols,
$$
\sum_{t=1}^\ell C_{rt}(\mu+\partial)\bullet(C^{-1})_{tj}(\mu)=\delta_{rj}(1\otimes1)\,.
$$
For all $a\in\mc V$, we have, by sesquilinearity and the left Leibniz rule
(the first equation in \eqref{20140305:eq1}):
$$
\begin{array}{l}
\displaystyle{
0=\sum_{t=1}^\ell\ldb a_\lambda C_{rt}(\mu+\partial)\bullet(C^{-1})_{tj}(\mu)\rdb_L
} \\
\displaystyle{
=
\sum_{t=1}^\ell
\ldb a_\lambda C_{rt}(y)\rdb_L
\bullet_1
\Big(\Big|_{y=\mu+\partial}(C^{-1})_{tj}(\mu)\Big)
} \\
\displaystyle{
+\sum_{t=1}^\ell 
i_\mu C_{rt}(\lambda+\mu+\partial)
\bullet_2
\ldb a_\lambda (C^{-1})_{tj}(\mu)\rdb
\,.
}
\end{array}
$$
Note that 
$i_\mu C(\lambda+\mu+\partial)$
is invertible in
$\Mat_{\ell\times\ell}\big(\mc V^{\otimes2}[\partial]((\lambda^{-1}))((\mu^{-1}))\big)$,
its inverse being 
$i_\mu C^{-1}(\lambda+\mu+\partial)$.
We then apply 
$i_\mu(C^{-1})_{ir}(\lambda+\mu+\partial)\bullet_2$ on the left to both 
sides of the above equation, and  using Lemma \ref{lemma:bullet-i}(a)
and summing over $r=1,\dots,\ell$, we get
$$
\begin{array}{l}
\displaystyle{
\sum_{t=1}^\ell \delta_{it} \ldb a_\lambda (C^{-1})_{tj}(\mu)\rdb_L
=-
\sum_{r,t=1}^\ell
i_\mu(C^{-1})_{ir}(\lambda+\mu+\partial)
\bullet_2
} \\
\displaystyle{
\qquad\qquad\qquad\qquad\qquad\qquad
\ldb a_\lambda C_{rt}(y)\rdb_L
\bullet_1
\Big(\Big|_{y=\mu+\partial}(C^{-1})_{tj}(\mu)\Big)\,,
}
\end{array}$$
proving equation \eqref{20111012:eq2aL}.
The proof of the equation \eqref{20111012:eq2aR} is done in a similar way
using the second equation in \eqref{20140305:eq1}.

For the third equation we have, by the right Leibniz rule (the third equation
in \eqref{20140305:eq1}):
$$
\begin{array}{l}
\displaystyle{
0=\sum_{t=1}^\ell
\ldb C_{rt}(\lambda+\partial)\bullet(C^{-1})_{tj}(\lambda)\,_{\lambda+\mu}a\rdb_L
} \\
\displaystyle{
=
\sum_{t=1}^\ell
{\ldb{C_{rt}}(x)_{\lambda+\mu+\partial}a\rdb_L}_\to
\bullet_3
\Big(\Big|_{x=\lambda+\partial}(C^{-1})_{tj}(\lambda)\Big)
} \\
\displaystyle{
+ \sum_{t=1}^\ell
{\ldb(C^{-1})_{tj}(\lambda)_{\lambda+\mu+\partial}a\rdb_L}_\to
\bullet_1\big(
i_\lambda C^\dagger_{tr}(\mu)\big)
\,.}
\end{array}
$$
We next replace in the above equation $\mu$ (placed at the right) by $\mu+\partial$,
and we apply the resulting differential operator to
$i_\lambda({C^\dagger}^{-1})_{ri}(\mu)$.
As a result we get, after summing over $r=1,\dots,\ell$:
$$
\begin{array}{l}
\displaystyle{
\sum_{t=1}^\ell 
{\ldb(C^{-1})_{tj}(\lambda)_{\lambda+\mu+\partial}a\rdb_L}_\to \delta_{ti}
=
-\sum_{r,t=1}^\ell
{\ldb{C_{rt}(x)}_{\lambda+\mu+\partial}a\rdb_L}_\to 
} \\
\displaystyle{
\qquad\qquad\qquad\qquad\qquad\qquad
\bullet_3\Big(\Big|_{x=\lambda+\partial}(C^{-1})_{tj}(\lambda)\Big)
\bullet_1\big(
i_\lambda({C^\dagger}^{-1})_{ri}(\mu)\big)
\,,
}
\end{array}
$$
proving equation \eqref{20111012:eq2b}.
\end{proof}
\begin{corollary}\label{20130514:cor}
Let $\ldb-_\lambda-\rdb:\,\mc V^{\otimes2}\to\mc V^{\otimes2}((\lambda^{-1}))$ 
be a $2$-fold $\lambda$-bracket on the differential algebra $\mc V$.
Let
$C(\partial)=\big(C_{ij}(\partial)\big)_{i,j=1}^\ell\in\Mat_{\ell\times\ell}\mc V^{\otimes2}((\partial^{-1}))$
be an invertible $\ell\times\ell$ matrix pseudodifferential operator with coefficients in $\mc V^{\otimes2}$,
and let
$C^{-1}(\partial)=\big((C^{-1})_{ij}(\partial)\big)_{i,j=1}^\ell\in\Mat_{\ell\times\ell}\mc V^{\otimes2}((\partial^{-1}))$
be its inverse.
Let $a\in\mc V$, and assume that
\begin{equation}\label{20130514:eq1}
\ldb a_\lambda C_{ij}(\mu)\rdb_L\,\in\mc V^{\otimes3}_{\lambda,\mu}
\,\,\text{ for all } i,j=1,\dots,\ell\,.
\end{equation}
(As before, we identify $\mc V^{\otimes3}_{\lambda,\mu}$ with its image 
$i_\mu(\mc V^{\otimes3}_{\lambda,\mu})\subset\mc V^{\otimes3}((\lambda^{-1}))((\mu^{-1}))$.)
Then, we have
$\ldb a_\lambda (C^{-1})_{ij}(\mu)\big\rdb_L,
\,\ldb(C^{-1})_{ij}(\lambda) _{\lambda+\mu} a\rdb_L\,\in\mc V^{\otimes3}_{\lambda,\mu}$.
In fact, 
the following identities hold in the space $\mc V^{\otimes3}_{\lambda,\mu}$:
\begin{equation}\label{20111012:eq2cL}
\begin{array}{l}
\displaystyle{
\ldb a_\lambda (C^{-1})_{ij}(\mu)\rdb_L
=
}
\\
\displaystyle{
-\sum_{r,t=1}^\ell
(C^{-1})_{ir}(\lambda+\mu+\partial)
\bullet_2
\ldb a_\lambda C_{rt}(y)\rdb_L
\bullet_1
\Big(\Big|_{y=\mu+\partial}(C^{-1})_{tj}(\mu)\Big)
\,,
}
\end{array}
\end{equation}
\begin{equation}\label{20111012:eq2cR}
\begin{array}{l}
\displaystyle{
\ldb a_\lambda (C^{-1})_{ij}(\mu)\rdb_R
=
} \\
\displaystyle{
-\sum_{r,t=1}^\ell
(C^{-1})_{ir}(\lambda+\mu+\partial)\bullet_2
\ldb a_\lambda C_{rt}(y)\rdb_R
\bullet_3
\Big(\Big|_{y=\mu+\partial}(C^{-1})_{tj}(\mu)
\Big)
\,,
}
\end{array}
\end{equation}
\begin{equation}\label{20111012:eq2d}
\begin{array}{l}
\displaystyle{
\ldb(C^{-1})_{ij}(\lambda) _{\lambda+\mu} a\rdb_L
=
} \\
\displaystyle{
-\sum_{r,t=1}^\ell
{\ldb C_{rt}(x)_{\lambda+\mu+\partial}a\rdb_L}_\to
\bullet_3
\Big(\Big|_{x=\lambda+\partial} (C^{-1})_{tj}(\lambda)\Big)
\bullet_1
\big(({C^\dagger}^{-1})_{ri}(\mu) \big)
\,.
}
\end{array}
\end{equation}
\end{corollary}
\begin{proof}
It is an immediate corollary of Lemmas \ref{20111006:lem} and \ref{20111012:lem}.
\end{proof}
\begin{remark}
Lemmas \ref{20111006:lem} and \ref{20111012:lem} are the double PVA analogues
of (respectively) \cite[Lemmas 2.3 and 3.9]{DSK13}. Corollary \ref{20130514:cor}
is the double PVA analogue of \cite[Corollary 1.4]{DSKV14b}. 
Naturally, if we put $\lambda =0$ in \eqref{dirac}, 
we obtain Dirac modification of the corresponding double Poisson algebra.
\end{remark}
\begin{proof}[Proof of Theorem \ref{prop:dirac}]
Both sesquilinearity conditions \eqref{20140702:eq4b} 
for the Dirac modified $2$-fold $\lambda$-bracket \eqref{dirac} are immediate to check.
%
%
The skewsymmetry condition \eqref{eq:skew2} 
for the Dirac modified $2$-fold $\lambda$-bracket \eqref{dirac}
can also be easily proved:
it follows by the skewsymmetry of the $2$-fold $\lambda$-bracket
$\ldb-_\lambda-\rdb$,
by the fact that the matrix $C(\partial)$ (hence $C^{-1}(\partial)$) is skewadjoint,
and by equation \eqref{20140710:eq4}.
%
%
Moreover, by equation \eqref{lemma:bullet}, it follows that
the Dirac modified $2$-fold $\lambda$-bracket
$\ldb-_\lambda-\rdb^D$
satisfies the Leibniz rules \eqref{20140702:eq6b}.


Next, we prove that the Dirac modified $2$-fold $\lambda$-bracket
$\ldb-_\lambda-\rdb^D$
is admissible, in the sense of equation \eqref{20110921:eq4}.
For this, we compute $\ldb a_\lambda\ldb b_\mu c\rdb^D\rdb^D_L$
using the definition \eqref{dirac},
the sesquilinearity conditions \eqref{20140702:eq4b},
the Leibniz rule \eqref{20140702:eq6b}
and equation \eqref{20140305:eq1}.
We get
\begin{eqnarray}
&&
\ldb a_\lambda\ldb b_\mu c\rdb^D\rdb_L^D
=
\ldb a_{\lambda}\ldb b_{\mu}c\rdb\rdb_L
\label{1a}\\
&&
-\sum_{\gamma,\delta=1}^m
\ldb
a_{\lambda}
\ldb{\theta_{\delta}}_{y}c\rdb
\rdb_L
\bullet_1\Big(\Big|_{y=\mu+\partial}
(C^{-1})_{\delta\gamma}(\mu+\partial)
\ldb b_{\mu}\theta_{\gamma}\rdb
\Big)
\label{2a}\\
&&
-\sum_{\gamma,\delta=1}^m
\ldb{\theta_{\delta}}_{\lambda+\mu+\partial}c\rdb_{\to}\bullet_2
\ldb
a_{\lambda}
(C^{-1})_{\delta\gamma}(y)
\rdb_L
\bullet_1\Big(\Big|_{y=\mu+\partial}
\ldb b_{\mu}\theta_{\gamma}\rdb
\Big)
\label{3a}\\
&&
-\sum_{\gamma,\delta=1}^m
\big(\ldb{\theta_{\delta}}_{\lambda+\mu+\partial}c\rdb_{\to}
(C^{-1})_{\delta\gamma}(\lambda+\mu+\partial)\big)\bullet_2
\ldb
a_{\lambda}
\ldb b_{\mu}\theta_{\gamma}\rdb
\rdb_L
\label{4a}\\
&&
-\sum_{\alpha,\beta=1}^m
{\ldb{\theta_{\beta}}_{\lambda+\partial}\ldb b_{\mu}c\rdb\rdb_L}_{\to}
\bullet_3
\left((C^{-1})_{\beta\alpha}(\lambda+\partial)
\ldb a_{\lambda}\theta_{\alpha}\rdb\right)
\label{5a}\\
&&
+\sum_{\alpha,\beta,\gamma,\delta=1}^m
\ldb
{\theta_{\beta}}_{x}
\ldb{\theta_{\delta}}_{y}c\rdb
\rdb_L
\bullet_3
\Big(\Big|_{x=\lambda+\partial}
(C^{-1})_{\beta\alpha}(\lambda+\partial)
\ldb a_{\lambda}\theta_{\alpha}\rdb
\Big)
\nonumber\\
&&
\qquad\qquad\qquad
\bullet_1\Big(\Big|_{y=\mu+\partial}
(C^{-1})_{\delta\gamma}(\mu+\partial)
\ldb b_{\mu}\theta_{\gamma}\rdb
\Big)
\label{6a}\\
&&
+\sum_{\alpha,\beta,\gamma,\delta=1}^m
\ldb{\theta_{\delta}}_{\lambda+\mu+\partial}c\rdb_{\to}\bullet_2
\ldb
{\theta_{\beta}}_{x}
(C^{-1})_{\delta\gamma}(y)
\rdb_L
\nonumber\\
&&
\qquad\qquad
\bullet_3\Big(\Big|_{x=\lambda+\partial}
(C^{-1})_{\beta\alpha}(\lambda+\partial)
\ldb a_{\lambda}\theta_{\alpha}\rdb
\Big)
\bullet_1\Big(\Big|_{y=\mu+\partial}
\ldb b_{\mu}\theta_{\gamma}\rdb
\Big)
\label{7a}\\
&&
+\sum_{\alpha,\beta,\gamma,\delta=1}^m
\big(\ldb{\theta_{\delta}}_{\lambda+\mu+\partial}c\rdb_{\to}
(C^{-1})_{\delta\gamma}(\lambda+\mu+\partial)\big)\bullet_2
{\ldb
{\theta_{\beta}}_{\lambda+\partial}
\ldb b_{\mu}\theta_{\gamma}\rdb
\rdb_L}_{\to}
\nonumber\\
&&
\qquad\qquad\qquad
\bullet_3\big((C^{-1})_{\beta\alpha}(\lambda+\partial)
\ldb a_{\lambda}\theta_{\alpha}\rdb\big)
\,.
\label{8a}
\end{eqnarray}
%
%
All the terms \eqref{1a}, \eqref{2a}, \eqref{4a}, \eqref{5a}, \eqref{6a}, and \eqref{8a},
lie in $\mc V_{\lambda,\mu}^{\otimes3}$
by the admissibility assumption on $\ldb-_\lambda-\rdb$
and Lemma \ref{20111006:lem}.
Moreover, by the admissibility of $\ldb-_\lambda-\rdb$
and the definition \eqref{C} of the matrix $C(\partial)$,
condition \eqref{20130514:eq1} holds.
Hence, we can use Corollary \ref{20130514:cor} and Lemma \ref{20111006:lem}
to deduce that
the terms \eqref{3a} and \eqref{7a} lie in $\mc V^{\otimes3}_{\lambda,\mu}$ as 
well. Therefore, 
$\ldb a_\lambda{\ldb b_\mu c\rdb^D}\rdb^D_L$ lies in $\mc V^{\otimes3}_{\lambda,\mu}$
for every $a,b,c\in\mc V$,
i.e. the Dirac modification $\ldb-_\lambda-\rdb^D$
is admissible.


In order to complete the proof of part (a) we are left to check the Jacobi identity 
\eqref{eq:jacobi2} for the Dirac modified $2$-fold $\lambda$-bracket.
We can use equation \eqref{20111012:eq2cL} in Corollary \ref{20130514:cor}
and Lemma \ref{lemma:bullet-i}(a,c)
to rewrite the terms \eqref{3a} and \eqref{7a}.
As a result, we get
\begin{eqnarray}
&&
\ldb a_\lambda\ldb b_\mu c\rdb^D\rdb_L^D
=
\ldb a_{\lambda}\ldb b_{\mu}c\rdb\rdb_L
\label{1b}\\
&&
-\sum_{\gamma,\delta=1}^m
\ldb
a_{\lambda}
\ldb{\theta_{\delta}}_{y}c\rdb
\rdb_L
\bullet_1\Big(\Big|_{y=\mu+\partial}
(C^{-1})_{\delta\gamma}(\mu+\partial)
\ldb b_{\mu}\theta_{\gamma}\rdb
\Big)
\label{2b}\\
&&
+\sum_{\gamma,\delta,\eta,\zeta=1}^m
\big(\ldb{\theta_{\delta}}_{\lambda+\mu+\partial}c\rdb_{\to}
(C^{-1})_{\delta\zeta}(\lambda+\mu+\partial)\big)
\nonumber\\
&&
\qquad\qquad\qquad
\bullet_2
\ldb
a_{\lambda}
\ldb \theta_{\eta}{}_y\theta_{\zeta}\rdb
\rdb_L
\bullet_1
\Big(\Big|_{y=\mu+\partial}
(C^{-1})_{\eta\gamma}(\mu+\partial)
\ldb b_{\mu}\theta_{\gamma}\rdb
\Big)
\label{3b}\\
&&
-\sum_{\gamma,\delta=1}^m
\big(\ldb{\theta_{\delta}}_{\lambda+\mu+\partial}c\rdb_{\to}
(C^{-1})_{\delta\gamma}(\lambda+\mu+\partial)\big)
\bullet_2
\ldb
a_{\lambda}
\ldb b_{\mu}\theta_{\gamma}\rdb
\rdb_L
\label{4b}\\
&&
-\sum_{\alpha,\beta=1}^m
{\ldb{\theta_{\beta}}_{\lambda+\partial}\ldb b_{\mu}c\rdb\rdb_L}_{\to}
\bullet_3
\left((C^{-1})_{\beta\alpha}(\lambda+\partial)
\ldb a_{\lambda}\theta_{\alpha}\rdb\right)
\label{5b}\\
&&
+\sum_{\alpha,\beta,\gamma,\delta=1}^m
\ldb
{\theta_{\beta}}_{x}
\ldb{\theta_{\delta}}_{y}c\rdb
\rdb_L
\bullet_3
\Big(\Big|_{x=\lambda+\partial}
(C^{-1})_{\beta\alpha}(\lambda+\partial)
\ldb a_{\lambda}\theta_{\alpha}\rdb
\Big)
\nonumber\\
&&
\qquad\qquad\qquad
\bullet_1
\Big(\Big|_{y=\mu+\partial}
(C^{-1})_{\delta\gamma}(\mu+\partial)
\ldb b_{\mu}\theta_{\gamma}\rdb
\Big)
\label{6b}\\
&&
+\sum_{\alpha,\beta,\gamma,\delta,\eta,\zeta=1}^m
\big(\ldb{\theta_{\delta}}_{\lambda+\mu+\partial}c\rdb_{\to}
(C^{-1})_{\delta\zeta}(\lambda+\mu+\partial)\big)
\nonumber\\
&&
\qquad\qquad\qquad
\bullet_2
\ldb
{\theta_{\beta}}_{x}
\ldb\theta_{\eta}{}_{y}\theta_{\zeta}\rdb
\rdb_L
\bullet_3
\Big(\Big|_{x=\lambda+\partial}
(C^{-1})_{\beta\alpha}(\lambda+\partial)
\ldb a_{\lambda}\theta_{\alpha}\rdb
\Big)
\nonumber\\
&&
\qquad\qquad\qquad
\bullet_3
\Big(\Big|_{y=\mu+\partial}
(C^{-1})_{\eta\gamma}(\mu+\partial)\ldb b_{\mu}\theta_{\gamma}\rdb
\Big)
\label{7b}\\
&&
+\sum_{\alpha,\beta,\gamma,\delta=1}^m
\big(\ldb{\theta_{\delta}}_{\lambda+\mu+\partial}c\rdb_{\to}
(C^{-1})_{\delta\gamma}(\lambda+\mu+\partial)\big)
\bullet_2
{\ldb
{\theta_{\beta}}_{\lambda+\partial}
\ldb b_{\mu}\theta_{\gamma}\rdb
\rdb_L}_{\to}
\nonumber\\
&&
\qquad\qquad\qquad
\bullet_3\big((C^{-1})_{\beta\alpha}(\lambda+\partial)
\ldb a_{\lambda}\theta_{\alpha}\rdb\big)
\,.
\label{8b}
\end{eqnarray}
Next, we compute the second term in the Jacobi identity using
the definition \eqref{dirac},
the sesquilinearity conditions \eqref{20140702:eq4b},
the Leibniz rules \eqref{20140702:eq6b},
equations \eqref{20140305:eq1} and \eqref{20111012:eq2cR}
and Lemma \ref{lemma:bullet-i}(a,c).
We get
\begin{eqnarray}
&&
\ldb b_\mu{\ldb a_\lambda c\rdb^D}\rdb^D_R
=
\ldb
b_{\mu}
\ldb a_{\lambda}c\rdb
\big\rdb_R
\label{1c}\\
&&
-\sum_{\alpha,\beta=1}^m
\ldb
b_{\mu}
\ldb{\theta_{\beta}}_{x}c\rdb
\rdb_R
\bullet_3
\Big(\Big|_{x=\lambda+\partial}
(C^{-1})_{\beta\alpha}(\lambda+\partial)
\ldb a_{\lambda}\theta_{\alpha}\rdb
\Big)
\label{2c}\\
&&
+\sum_{\alpha,\beta,\gamma,\delta}^m
\big(
\ldb{\theta_{\delta}}_{\lambda+\mu+\partial}c\rdb_{\to}
(C^{-1})_{\delta\gamma}(\lambda+\mu+\partial)
\big)
\bullet_2
\ldb
b_\mu {\ldb{\theta_\beta}_x\theta_\gamma\rdb}
\rdb_R
\nonumber\\
&&
\qquad\qquad\qquad
\bullet_3
\Big(\Big|_{x=\lambda+\partial}
(C^{-1})_{\beta\alpha}(\lambda+\partial)
\ldb a_{\lambda}\theta_{\alpha}\rdb
\Big)
\label{3c}\\
&&
-\sum_{\gamma,\delta=1}^m
\big(
\ldb{\theta_{\delta}}_{\lambda+\mu+\partial}c\rdb_{\to}
(C^{-1})_{\delta\gamma}(\lambda+\mu+\partial)\big)
\bullet_2
\ldb
b_{\mu}
\ldb a_{\lambda}\theta_{\gamma}\rdb
\rdb_R
\label{4c}\\
&&
-\sum_{\gamma,\delta=1}^m
{\ldb
{\theta_{\delta}}_{\mu+\partial}
\ldb a_{\lambda}c\rdb
\rdb_R}_{\to}
\bullet_1
\big((C^{-1})_{\delta\gamma}(\mu+\partial)
\ldb b_{\mu}\theta_{\gamma}\rdb\big)
\label{5c}\\
&&
+\sum_{\alpha,\beta,\gamma,\delta=1}^m
\ldb
{\theta_{\delta}}_{y}
\ldb{\theta_{\beta}}_{x}c\rdb
\rdb_R
\bullet_1
\Big(\Big|_{y=\mu+\partial}
(C^{-1})_{\delta\gamma}(\mu+\partial)
\ldb b_{\mu}\theta_{\gamma}\rdb
\Big)
\nonumber\\
&&
\qquad\qquad\qquad
\bullet_3
\Big(\Big|_{x=\lambda+\partial}
(C^{-1})_{\beta\alpha}(\lambda+\partial)
\ldb a_{\lambda}\theta_{\alpha}\rdb
\Big)
\label{6c}\\
&&
-\sum_{\alpha,\beta,\gamma,\delta,\eta,\zeta=1}^m
\big(\ldb{\theta_{\delta}}_{\lambda+\mu+\partial}c\}_{\to}
(C^{-1})_{\delta\zeta}(\lambda+\mu+\partial)\big)
\nonumber\\
&&
\qquad\qquad\qquad
\bullet_2
\ldb
{\theta_{\eta}}_y {\ldb{\theta_\beta}_x\theta_\zeta\rdb}
\rdb_R
\bullet_1
\Big(\Big|_{y=\mu+\partial}
(C^{-1})_{\eta\gamma}(\mu+\partial)
\ldb b_{\lambda}\theta_{\gamma}\rdb
\Big)
\nonumber\\
&&
\qquad\qquad\qquad
\bullet_3
\Big(\Big|_{x=\lambda+\partial}
(C^{-1})_{\beta\alpha}(\lambda+\partial)
\ldb a_{\lambda}\theta_{\alpha}\rdb
\Big)
\label{7c}\\
&&
+\sum_{\gamma,\delta,\eta,\zeta=1}^m
\big(
\ldb{\theta_{\delta}}_{\lambda+\mu+\partial}c\rdb_{\to}
(C^{-1})_{\delta\zeta}(\lambda+\mu+\partial)
\big)
\nonumber\\
&&
\qquad\qquad
\bullet_2
{\ldb
{\theta_{\eta}}_{\mu+\partial}
\ldb a_{\lambda}\theta_{\zeta}\rdb
\rdb_R}_{\to}
\bullet_1
\big((C^{-1})_{\eta\gamma}(\mu+\partial)
\ldb b_{\mu}\theta_{\gamma}\rdb\big)
\,.\label{8c}
\end{eqnarray}
In a similar way we compute the RHS of the Jacobi identity
for the Dirac modified $2$-fold $\lambda$-bracket,
using the definition \eqref{dirac},
the sesquilinearity \eqref{20140702:eq4b},
the right Leibniz rule \eqref{20140702:eq6b},
equation \eqref{20111012:eq2d}
(recall that the matrix $C$ is skewadjoint), and Lemma \ref{lemma:bullet-i}(a,c).
We get
\begin{eqnarray}
&&
\ldb{\ldb a_\lambda b\rdb^D}_{\lambda+\mu}c\rdb^D_L
=
\ldb
{\ldb a_{\lambda}b\rdb}_{\lambda+\mu}c
\rdb_L
\label{1d}\\
&&
-\sum_{\alpha,\beta=1}^m
{\ldb
\ldb{\theta_{\beta}}_{x}b\rdb
_{\lambda+\mu+\partial}
c\rdb_L}_\to
\bullet_3
\Big(\Big|_{x=\lambda+\partial}
(C^{-1})_{\beta\alpha}(\lambda+\partial)
\ldb a_{\lambda}\theta_{\alpha}\rdb
\Big)
\label{2d}\\
&&
-\sum_{\alpha,\beta,\gamma,\delta=1}^m
{\ldb
{\ldb{\theta_\beta}_x\theta_\gamma\rdb}
_{\lambda+\mu+\partial}c
\rdb_L}_\to
\bullet_3
\Big(\Big|_{x=\lambda+\partial}
(C^{-1})_{\beta\alpha}(\lambda+\partial)
\ldb a_{\lambda}\theta_{\alpha}\rdb
\Big)
\nonumber\\
&&
\qquad\qquad\qquad\qquad
\bullet_1
\big(
(C^{-1})_{\delta\gamma}(\mu+\partial) 
\ldb{\theta_{\gamma}}_{-\mu-\partial}b\rdb^\sigma
\big)
\label{3d}\\
&&
-\sum_{\gamma,\delta=1}^m
{{\ldb a_{\lambda}\theta_{\delta}\rdb}_{\lambda+\mu+\partial}c
\rdb_L}_\to
\bullet_1
\big(
({C^\dagger}^{-1})_{\delta\gamma}(\mu+\partial)
\ldb{\theta_{\gamma}}_{-\mu-\partial}b\rdb^\sigma
\big)
\label{4d}\\
&&
-\sum_{\gamma,\delta=1}^m
\big(
\ldb{\theta_{\delta}}_{\lambda+\mu+\partial}c\rdb_{\to}
(C^{-1})_{\delta\gamma}(\lambda+\mu+\partial)
\big)
\bullet_2
\ldb{\ldb a_{\lambda}b\rdb}_{\lambda+\mu}\theta_{\gamma}
\rdb_L
\label{5d}\\
&&
+\sum_{\alpha,\beta,\gamma,\delta=1}^m
\big(
\ldb{\theta_{\delta}}_{\lambda+\mu+\partial}c\rdb_{\to}
(C^{-1})_{\delta\gamma}(\lambda+\mu+\partial)
\big)
\nonumber\\
&&
\qquad
\bullet_2
\ldb{\ldb{\theta_{\beta}}_{x}b\rdb_{\lambda+\mu+\partial}\theta_{\gamma}
\rdb_L}_\to
\bullet_3
\Big(\Big|_{x=\lambda+\partial}
(C^{-1})_{\beta\alpha}(\lambda+\partial)
\ldb a_{\lambda}\theta_{\alpha}\rdb
\Big)
\label{6d}\\
&&
+\sum_{\alpha,\beta,\gamma,\delta,\eta,\zeta=1}^m
\big(
\ldb{\theta_{\delta}}_{\lambda+\mu+\partial}c\rdb_{\to}
(C^{-1})_{\delta\gamma}(\lambda+\mu+\partial)
\big)
\nonumber\\
&&
\qquad\qquad
\bullet_2
{\ldb
{\ldb{\theta_\beta}_x\theta_\eta\rdb}_{\lambda+\mu+\partial}\theta_{\zeta}
\rdb_L}_\to
\bullet_3
\Big(\Big|_{x=\lambda+\partial}
(C^{-1})_{\beta\alpha}(\lambda+\partial)
\ldb a_{\lambda}\theta_{\alpha}\rdb
\Big)
\nonumber\\
&&
\qquad\qquad
\bullet_1
\big((C^{-1})_{\eta\gamma}(\mu+\partial)
\ldb{\theta_{\gamma}}_{-\mu-\partial}b\rdb^\sigma
\big)
\label{7d}\\
&&
+\sum_{\gamma,\delta,\eta,\zeta=1}^m
\big(
\ldb{\theta_{\delta}}_{\lambda+\mu+\partial}c\rdb_{\to}
(C^{-1})_{\delta\zeta}(\lambda+\mu+\partial)
\big)
\nonumber\\
&&
\bullet_2
{\ldb \ldb a_{\lambda}\zeta_{\alpha}\rdb_{\lambda+\mu+\partial}\theta_{\eta}
\rdb_L}_\to
\bullet_1
\big(({C^\dagger}^{-1})_{\eta\gamma}(\mu+\partial)
\ldb{\theta_{\gamma}}_{-\mu-\partial}b\rdb^\sigma\big)
\,.
\label{8d}
\end{eqnarray}
The following equations hold
due to the skewsymmetry \eqref{eq:skew2},
the Jacobi identity \eqref{eq:jacobi2},
and the fact that the matrix $C$ is skewadjoint:
$$
\begin{array}{l}
\text{RHS}\eqref{1b}-\text{RHS}\eqref{1c}=\text{RHS}\eqref{1d}
\,\,,\,\,\,\,
\eqref{2b}-\eqref{5c}=\eqref{4d}
\,,\\
\eqref{5b}-\eqref{2c}=\eqref{2d}
\,\,,\,\,\,\,
\eqref{4b}-\eqref{4c}=\eqref{5d}
\,\\
\eqref{3b}-\eqref{8c}=\eqref{8d}
\,,\,\,\,\,
\eqref{8b}-\eqref{3c}=\eqref{6d}
\,.
\end{array}
$$
Moreover, using Lemma \ref{lemma:bullet-i}(c),
the skewsymmetry \eqref{eq:skew2} and
the Jacobi identity \eqref{eq:jacobi2}, we also get
$$
\eqref{6b}-\eqref{6c}=\eqref{3d}
\,,\,\,\,\,
\eqref{7b}-\eqref{7c}=\eqref{7d}
\,.
$$
This concludes the proof of the Jacobi identity for the Dirac modified $2$-fold
$\lambda$-bracket, and of part (a).


Note that the identities $C(\partial)C^{-1}(\partial)=C^{-1}(\partial)C(\partial)=1\otimes1$ read, 
in terms of the symbols of the pseudodifferential operators $C(\partial)$ and $C^{-1}(\partial)$,
as
$$
\sum_{\beta=1}^m
\!
\ldb{\theta_\beta}_{\lambda+\partial}\theta_\alpha\rdb_\to
\!\bullet\!
(C^{-1})_{\beta\gamma}(\lambda)
\!=\!
\delta_{\alpha\gamma}(1\otimes1)
\,,\,\,
\sum_{\beta=1}^m
(C^{-1})_{\alpha\beta}(\lambda+\partial)
\!\bullet\!
\ldb{\theta_\gamma}_\lambda\theta_\beta\rdb
\!=\!
\delta_{\alpha\gamma}(1\otimes1)
\,.
$$
Part (b) is an immediate consequence of these identities and the definition \eqref{dirac}
of the Dirac modified $2$-fold $\lambda$-bracket.
\end{proof}
\begin{definition}
Let $\mc I\subset\mc V$ be a two-sided differential ideal. We say that $\mc I$ is
a \emph{double Poisson vertex algebra ideal} if
$$
\ldb \mc V_\lambda\mc I\rdb,\ldb\mc I_\lambda \mc V\rdb
\in(\mc V\otimes\mc I+\mc I\otimes\mc V)[\lambda]\,.
$$
If $\mc I\subset\mc V$ is a double PVA ideal then we have an induced
double PVA structure on the quotient differential algebra $\quot{\mc V}{\mc I}$.
\end{definition}
\begin{corollary}
The two-sided differential ideal $\mc I=\langle\theta_1,\dots,\theta_m\rangle_{\mc V}\subset\mc V$,
generated by $\theta_1,\dots,\theta_m$,
is a double PVA ideal with respect to the Dirac modified double
$\lambda$-bracket $\ldb \cdot\,_\lambda\,\cdot\rdb^D$.
Hence, the quotient space $\quot{\mc V}{\mc I}$ is a (non-local) double PVA,
with $2$-fold $\lambda$-bracket induced by $\ldb-_\lambda-\rdb^D$,
which we call the \emph{double Dirac reduction} of $\mc V$ 
by the constraints $\theta_1,\dots,\theta_m$.
\end{corollary}
\begin{proof}
The statement follows by the sesquilinearity conditions
and the left and right Leibniz rules
for the Dirac modified $2$-fold $\lambda$-bracket.
\end{proof}
\begin{example}\label{exa:5.9}
Let us consider the non-commutative algebra of differential polynomials
$\mc V=\mc R_I=\mb F<u_i^{(n)}\mid i\in I,n\in\mb Z_+>$ (the index set may be either
$I=\{i\in\mb Z|i\geq-N\}$ or $I_-=\{-N,-N+1,\dots,-1\}$)
with the double Poisson vertex algebra structure defined by
equation \eqref{eq:H}.
Let us denote $C(\lambda)=\ldb u_{-N}{}_\lambda u_{-N}\rdb_H$. By Lemma
\ref{20140311:lem1} we have that
$$
C(\lambda)=1\otimes u_{-N}-u_{-N}\otimes1-(1\otimes1)N\lambda\,,
$$
thus $C(\partial)\in\mc V^{\otimes2}((\partial^{-1}))$ is an invertible
pseudodifferential operator.
Denote by $\mc I$ the two-sided differential ideal of $\mc V$ generated by $u_{-N}$.
Note that $\quot{\mc V}{\mc I}\simeq\mb F<u_i^{(n)}\mid i\in I\setminus\{-N\},n\in\mb Z_+>$.
Using the explicit expression for $C(\lambda)$, the formula for the Dirac 
reduction given by \eqref{dirac}, and the expression for the $2$-fold
$\lambda$-bracket $\ldb-_\lambda-\rdb_H$ given by \eqref{eq:H} and Lemma \ref{20140311:lem1},
the Dirac reduced $2$-fold $\lambda$-bracket on the quotient space
$\mb F<u_i^{(n)}\mid i\in I\setminus\{-N\},n\in\mb Z_+>$, which we denote
by $\ldb-_\lambda-\rdb_{H^D}$, becomes
\begin{equation}\label{20140311:eq1}
\begin{array}{l}
\displaystyle{
\ldb L(z)_{\lambda}L(w)\rdb_{H^D}
=L(z)\otimes i_z(z-w-\lambda-\partial)^{-1}L(w)
}
\\
\displaystyle{
-L(w+\lambda+\partial)\otimes i_z(z-w-\lambda-\partial)^{-1}L^*(-z+\lambda)
}
\\
\displaystyle{
-\frac1N
L(w+\lambda+\partial)\otimes(\lambda+\partial)^{-1}L^*(-z+\lambda)
-\frac1N
\big((\lambda+\partial)^{-1}L(z)\big)\otimes L(w)
}
\\
\displaystyle{
+\frac1N L(w+\lambda+\partial)(\lambda+\partial)^{-1}L(z)\otimes1
+\frac1N \otimes\big((\lambda+\partial)^{-1}L^*(-z+\lambda)\big)L(w)\,.
}
\end{array}
\end{equation}
Using equations \eqref{20140311:eq1} and \eqref{20130917:eq1},
the corresponding $\lambda$-bracket on $\left(\quot{\mc V}{\mc I}\right)_m
=\mb F<u_{ab,i}^{n}\mid i\in I\setminus\{-N\},a,b=1,\dots,m,n\in\mb Z_+>$
is given by the following generating series
\begin{equation}\label{eq:H_dirac_mat}
\begin{array}{l}
\vphantom{\Big(}
\displaystyle{
\{L_{ab}(z)_{\lambda}L_{cd}(w)\}_{H^D}
=
L_{cb}(z)i_z(z-w-\lambda-\partial)^{-1}L_{ad}(w)
} \\
\vphantom{\Big(}
\displaystyle{
-L_{cb}(w+\lambda+\partial)i_z(z-w-\lambda-\partial)^{-1}L_{ad}^*(-z+\lambda)
} \\
\vphantom{\Big(}
\displaystyle{
-\frac1NL_{cb}(w+\lambda+\partial)(\lambda+\partial)^{-1}L^*_{ad}(-z+\lambda)
-\frac1NL_{ad}(w)(\lambda+\partial)^{-1}L_{cb}(z)
} \\
\vphantom{\Big(}
\displaystyle{
+\frac{1}{N}
\sum_{k=1}^m
\delta_{ad}L_{ck}(w+\lambda+\partial)(\lambda+\partial)^{-1}L_{kb}(z)
} \\
\vphantom{\Big(}
\displaystyle{
+\frac{1}{N}
\sum_{k=1}^m
\delta_{cb}L_{kd}(w)(\lambda+\partial)^{-1}L^*_{ak}(-z+\lambda)
\,.}
\end{array}
\end{equation}
This is the same as equation (4.11) in \cite{DSKV14a}.
\end{example}
\begin{example}
For $N=2$, we have
$\quot{\mc R_2}{\mc I}=\mb F\langle u^{(n)}\mid n\in\mb Z_+\rangle$
(where $u$ is the image of $u_{-1}$).
The two compatible double PVA structures
$(H^D,K)$ given by the
equations \eqref{eq:H_dirac_mat} and \eqref{eq:K}
(using that $L(z)=z^2+u$)
are
\begin{align*}
&\ldb u_\lambda u\rdb_{H^D}=(1\otimes1)\frac{\lambda^3}{2}
+\frac12(2\lambda+\partial)(u\otimes1+1\otimes u)
-\frac12u\otimes(\lambda+\partial)^{-1}u\\
&-\frac12\left((\lambda+\partial)^{-1}u\right)\otimes u
+\frac12u(\lambda+\partial)^{-1}u\otimes1
+\frac12\otimes\left((\lambda+\partial)^{-1}u\right)u\,,\\
&\ldb u_\lambda u\rdb_{K}=2(1\otimes1)\lambda\,.
\end{align*}
The above double Poisson structures are (up to a scalar factor) the same as that
appearing in formula (6.15) in \cite{OS98}. Note that the Olver and Sokolov 
formula is given in terms of operators on $\quot{\mc R_2}{\mc I}$, rather than as
an element of
$\left(\quot{\mc R_2}{\mc I}\otimes\quot{\mc R_2}{\mc I}\right)[\lambda]$.
The rule to recover their formula is the following: we should replace an element
$p\otimes q\in\quot{\mc R_2}{\mc I}\otimes\quot{\mc R_2}{\mc I}$
by the operator $\mc L_p\circ\mc R_q$, where $\mc L_p$ (respectively
$\mc R_q$) denotes the left (respectively right) multiplication by $p$
(respectively $q$).
\end{example}
\begin{example}
For $N=3$, we have
$\quot{\mc R_3}{\mc I}=\mb F\langle u^{(n)},v^{(n)}\mid n\in\mb Z_+\rangle$
(where $u$ is the image of $u_{-2}$ and $v$ is the image of $u_{-1}$).
The two compatible double PVA structures
$(H^D,K)$ given by the
equations \eqref{eq:H_dirac_mat} and \eqref{eq:K}
(using that $L(z)=z^3+uz+v$)
are
\begin{align*}
&\ldb u_\lambda u\rdb_{H^D}=
\frac{1}{3}\left(
u(\lambda+\partial)^{-1}u\otimes1
+1\otimes\left((\lambda+\partial)^{-1}u\right)u
\right)
\\
&-\frac13\left(u\otimes(\lambda+\partial)^{-1}u+
\left((\lambda+\partial)^{-1}u\right)\otimes u
\right)
+1\otimes v-v\otimes1
\\
&+(1\otimes u)\lambda+(\lambda+\partial)u\otimes1
+2(1\otimes1)\lambda^3\,,
\\
&\ldb u_\lambda v\rdb_{H^D}=
\frac13\left(
v(\lambda+\partial)^{-1}u\otimes1
+1\otimes\left((\lambda+\partial)^{-1}u\right)v
\right)
\\
&-\frac13\left(v\otimes(\lambda+\partial)^{-1}u
+\left((\lambda+\partial)^{-1}u\right)\otimes v
\right)
+\frac13\left(u^2\otimes1-u\otimes u
\right)\\
&+1\otimes(2\lambda+\partial)v+(v\otimes1)\lambda
+\frac13(\lambda+\partial)^2\left(u\otimes1-1\otimes u
\right)
+(u\otimes1)\lambda^2\\
&+(1\otimes1)\lambda^4\,,
\\
&\ldb v_\lambda v\rdb_{H^D}
=\frac13\left(v(\lambda+\partial)^{-1}v\otimes1
+1\otimes\left((\lambda+\partial)^{-1}v\right)v
\right)
+\frac23\left(u\otimes v-v\otimes u\right)
\\
&+\frac13\left(uv\otimes1-1\otimes uv
\right)
-\frac23\left(
u\otimes(\lambda+\partial)u
\right)
+\frac23\left((\lambda+\partial)^2(1\otimes v)-(v\otimes1)\lambda^2
\right)
\\
&+\frac13\left((\lambda+\partial)^2(v\otimes 1)-(1\otimes v)\lambda^2
\right)
-\frac23\left((\lambda+\partial)^3(1\otimes u)-(u\otimes 1)\lambda^3
\right)
\\
&-\frac25(1\otimes1)\lambda^5\,,
\\
&\ldb u_\lambda u\rdb_{K}=0\,,
\qquad
\ldb u_\lambda v\rdb_{K}=3(1\otimes1)\lambda\,,
\qquad
\ldb v_\lambda v\rdb_{K}=u\otimes1-1\otimes u\,.
\end{align*}
\end{example}

\section{Adler-Gelfand-Dickey non-commutative integrable hierarchies}\label{sec:hierarchies}
In this section we want to show how to apply the Lenard-Magri scheme
of integrability (see Section \ref{sec:lenard-magri})
in order to obtain integrable hierarchies for
the compatible pair of PVAs we constructed in Section \ref{sec:AGD}.

First we state a technical lemma.
\begin{lemma}\label{lem:15032013}
Let $\mc V$ be an arbitrary differential algebra endowed with a $2$-fold
$\lambda$-bracket
$\ldb-_{\lambda}-\rdb$.
Let $L(\partial)\in\mc V((\partial^{-1}))$ be a monic pseudodifferential operator
of order $N>0$.
Then, for all $k\geq1$, the following identity holds in $\mc V((w^{-1}))$:
\begin{equation}\label{eq:lenard1}
\begin{array}{l}
\displaystyle{
\res_z
\mult\left(\ldb L^{\frac kN}(z)_{\lambda}L(w)\rdb\big|_{\lambda=0}\right)
}\\
\displaystyle{
=\frac kN
\res_z\mult\left(\ldb L(z+x)_x L(w)\rdb\star_1\big(\big|_{x=\partial}L^{\frac kN-1}(z)\big)\right)\,.
}
\end{array}
\end{equation}
\end{lemma}
\begin{proof}
Since,
$L^{\frac kN}(z)
=L^{\frac 1N}(z+\partial)L^{\frac 1N}(z+\partial)\dots L^{\frac 1N}(z)$ 
($k$ times), we have, 
by sesquilinearity and the right Leibniz rule,
\begin{equation}\label{eq:15032013_1}
\begin{array}{l}
\displaystyle{
\ldb L^{\frac kN}(z)_{\lambda}L(w)\rdb
=\sum_{l=1}^k
\big(\big|_{y=\partial}(L^*)^{\frac{l-1}{N}}(-z+\lambda)\big)
}
\\
\displaystyle{
\star_1
\ldb L^{\frac1N}(z+x)_{\lambda+x+y}L(w)\rdb\star_1
\big(\big|_{x=\partial}L^{\frac{k-l}{N}}(z)\big)
\,.
}
\end{array}
\end{equation}
Since $\mult(a\star_1X\star_1 b)=\mult(X\star_1 ba)$, for
all $a,b\in\mc V$ and $X\in\mc V^{\otimes2}$, we have
\begin{equation}\label{eq:15032013_1bis}
\begin{array}{l}
\displaystyle{
\mult\ldb L^{\frac kN}(z)_{\lambda}L(w)\rdb
=\sum_{l=1}^k
\mult\Big(
\ldb L^{\frac1N}(z+x)_{\lambda+x+y}L(w)\rdb
}
\\
\displaystyle{
\star_1\big(\big|_{x=\partial}L^{\frac{k-l}{N}}(z)\big)
\big(\big|_{y=\partial}(L^*)^{\frac{l-1}{N}}(-z+\lambda)\big)
\Big)
\,.
}
\end{array}
\end{equation}
Taking the residue of both sides of equation \eqref{eq:15032013_1bis} and
using \eqref{star} with $t=\lambda+y$, we get
$$
\begin{array}{l}
\displaystyle{
\res_z\mult\ldb L^{\frac kN}(z)_{\lambda}L(w)\rdb
=\res_z\sum_{l=1}^k
\mult\Big(
\ldb L^{\frac1N}(z+\lambda+x+y)_{\lambda+x+y}L(w)\rdb
}
\\
\displaystyle{
\star_1\big(\big|_{x=\partial}L^{\frac{k-l}{N}}(z+\lambda+y)\big)
\big(\big|_{y=\partial}L^{\frac{l-1}{N}}(z)\big)
\Big)\,,
}
\end{array}
$$
and setting $\lambda=0$ we get
\begin{equation}\label{eq:15032013_1e}
\begin{array}{l}
\displaystyle{
\res_z\mult\left(\ldb L^{\frac kN}(z)_{\lambda}L(w)\rdb\Big|_{\lambda=0}\right)
=
}
\\
\displaystyle{
k\res_z\mult\left(\ldb L^{\frac1N}(z+x)_{x}L(w)\rdb\star_1
\big(\big|_{x=\partial}L^{\frac{k-1}{N}}(z))\right)
\,.
}
\end{array}
\end{equation}
On the other hand, letting $k=N$ in \eqref{eq:15032013_1}, we have
\begin{equation}\label{eq:15032013_1b}
\begin{array}{l}
\displaystyle{
\ldb L(z)_{\lambda}L(w)\rdb
=\sum_{l=1}^N
\big(\big|_{x=\partial}L^{\frac{N-l}{N}}(z)\big)
}
\\
\displaystyle{
\star_1\ldb L^{\frac1N}(z+x)_{\lambda+x+y}L(w)\rdb\star_1
\big(\big|_{y=\partial}(L^*)^{\frac{l-1}{N}}(-z+\lambda)\big)
\,.
}
\end{array}
\end{equation}
If we replace in equation \eqref{eq:15032013_1b},
$z$ by $z+\partial$ and $\lambda$ by $\lambda+\partial$
acting on $L^{\frac kN-1}(z)$, and we apply the multiplication map $\mult$
we get
\begin{equation}\label{eq:15032013_1c2}
\begin{array}{l}
\displaystyle{
\mult\left(\ldb L(z+x)_{\lambda+x}L(w)\rdb
\star_1\big(\big|_{x=\partial}L^{\frac kN-1}(z)\big)
\right)
} \\
\displaystyle{
=\sum_{l=1}^N
\mult\left(
\ldb L^{\frac1N}(z+x)_{\lambda+x+y}L(w)\rdb\star_1
\big(\big|_{x=\partial}L^{\frac{k-l}{N}}(z)\big)
\big(\big|_{y=\partial}(L^*)^{\frac{l-1}{N}}(-z+\lambda)\big)
\right)
\,.
}
\end{array}
\end{equation}
Taking, as before, residues of both sides of equation \eqref{eq:15032013_1c2}
and using \eqref{star} with $t=\lambda+y$, 
we get, after setting $\lambda=0$,
\begin{equation}\label{eq:15032013_1d}
\begin{array}{l}
\displaystyle{
\res_z\mult\left(\ldb L(z+x)_{x}L(w)\rdb\star_1\big(\big|_{x=\partial}L^{\frac kN-1}(z)\big)\right)
} \\
\displaystyle{
=N\res_z\mult\left(\ldb L^{\frac1N}(z+x)_{x}L(w)\rdb
\big(\big|_{x=\partial}L^{\frac{k-1}{N}}(z)\big)\right)
\,.}
\end{array}
\end{equation}
Equation \eqref{eq:lenard1} follows from equations \eqref{eq:15032013_1e} and \eqref{eq:15032013_1d}.
\end{proof}
As in Section \ref{sec:AGD}, let $\mc V$ be the non-commutative algebra of 
differential polynomials in the variables $u_i$, $i\in I$, namely
$$
\mc V=\mc R_I=\mb F\langle u_i^{(n)}\mid i\in I,n\in\mb Z_+\rangle\,.
$$
(The index set $I$ may be either $I=\{i\in\mb Z\mid i\geq-N\}$
or $I_-=\{-N,-N+1,\dots,-1\}$).
Let us collect the differential generators of $\mc V$ into
the generating series
$$
L(z)=z^N+\sum_{i\in I}u_{i}z^{-i-1}\in\mc V((z^{-1}))\,,
$$
and let us consider the compatible double PVA structures $(H,K)$ on $\mc V$
defined by  Corollary \ref{cor:bi_dPVA}.
Let, for $k\geq1$,
\begin{equation}\label{hk}
h_k=\frac Nk\res_z L^{\frac kN}(z)\in\mc V\,,
\end{equation}
where $L^\frac1N(\partial)\in\mc V((\partial^{-1}))$
is uniquely defined (see Proposition 1.1 in \cite{DSKV14a}).
\begin{lemma}\label{cor:17032013}
For the compatible pair of double PVA structures
$(H,K)$ on $\mc V$, we have for $k\geq1$:
\begin{enumerate}[(a)]
\item
$\mult\ldb h_k{}_\lambda L(w)\rdb_H\big|_{\lambda=0}
=L^{\frac kN}(w+\partial)_+L(w)-L(w+\partial)L^{\frac kN}(w)_+$;
\item
$\mult\ldb h_k{}_\lambda L(w)\rdb_K\big|_{\lambda=0}
=L^{\frac{k}{N}-1}(w+\partial)_+L(w)-L(w+\partial)L^{\frac{k}{N}-1}(w)_+$.
\end{enumerate}
\end{lemma}
\begin{proof}
By Lemma \ref{lem:15032013} and equation \eqref{eq:H} we have
\begin{equation}\label{20131106:eq1}
\begin{array}{l}
\displaystyle{
\mult\left.\ldb h_k{}_\lambda L(w)\rdb_H\right|_{\lambda=0}
=\res_z
L^{\frac kN}(z)i_z(z-w-\partial)^{-1}
L(w)
}
\\
\displaystyle{
-L(w+\partial)
\res_z L^{\frac{k}{N}-1}(z)i_z(z-w-\partial)^{-1}L^*(-z)
\,.}
\end{array}
\end{equation}
Using \eqref{valentinesday}, we have
\begin{equation}\label{20131106:eq2}
\res_z
L^{\frac kN}(z)i_z(z-w-\partial)^{-1}
=L^{\frac kN}(w+\partial)_+
\,,
\end{equation}
while, using \eqref{star} and equation \eqref{valentinesday},
we have
\begin{equation}\label{20131106:eq3}
\begin{array}{l}
\vphantom{\Big(}
\displaystyle{
\res_z L^{\frac{k}{N}-1}(z)i_z(z-w-\partial)^{-1}L^*(-z)
=
\res_z L^{\frac{k}{N}-1}(z+\partial)i_z(z-w)^{-1}L(z)
} \\
\vphantom{\Big(}
\displaystyle{
=
\res_z L^{\frac{k}{N}}(z)i_z(z-w)^{-1}
=
L^{\frac{k}{N}}(w)_+
\,.}
\end{array}
\end{equation}
Combining equations \eqref{20131106:eq1}, \eqref{20131106:eq2} and \eqref{20131106:eq3},
we get part (a).
Similarly, for part (b), 
we use Lemma \ref{lem:15032013} and equation \eqref{eq:K} to get
\begin{equation}\label{20131106:eq4}
\begin{array}{l}
\vphantom{\Big(}
\displaystyle{
\mult\left.\ldb h_k{}_\lambda L(w)\rdb_K\right|_{\lambda=0}
=
\res_zi_z(z-w)^{-1}
\big(L(z+\partial)-L(w+\partial)\big)L^{\frac kN-1}(z)
} \\
\vphantom{\Big(}
\displaystyle{
+
\res_z
L^{\frac kN-1}(z)
i_z(z-w-\partial)^{-1}
\big(L(w)-L^*(-z)\big)
\,.}
\end{array}
\end{equation}
By equations \eqref{valentinesday} and \eqref{star},
we have
\begin{equation}\label{20131106:eq5}
\begin{array}{l}
\displaystyle{
\res_zi_z(z-w)^{-1}
L(z+\partial)L^{\frac kN-1}(z)
=
L^{\frac kN}(w)_+
} \\
\displaystyle{
=
\res_z
L^{\frac kN-1}(z)
i_z(z-w-\partial)^{-1}
L^*(-z)
\,.}
\end{array}
\end{equation}
Moreover, by equation \eqref{valentinesday} we also have
\begin{equation}\label{20131106:eq6}
\res_zi_z(z-w)^{-1}
L(w+\partial)
L^{\frac kN-1}(z)
=
L(w+\partial)L^{\frac kN-1}(w)_+\,,
\end{equation}
and
\begin{equation}\label{20131106:eq7}
\res_z
L^{\frac kN-1}(z)
i_z(z-w-\partial)^{-1}L(w)
=
L^{\frac kN-1}(w+\partial)_+
L(w)
\,.
\end{equation}
Combining equations \eqref{20131106:eq4}, 
\eqref{20131106:eq5}, \eqref{20131106:eq6}, and \eqref{20131106:eq7}, 
we get the claim.
\end{proof}
The following Theorem says that the Lenard-Magri scheme of integrability works
for the bi-Poisson structure $(H,K)$,
see Remark \ref{rem:20140723}:
\begin{theorem}\label{lem:lenard_works}
\begin{enumerate}[(a)]
\item
For every $\varepsilon\in\{1,\dots,N\}$, 
we have 
\begin{equation}\label{leneq2}
\big\{\tint {h_{\varepsilon}} , u \big\}_K=0
\,,\,\,
\text{ for all } u\in\mc V
\,.
\end{equation}
\item
For every $k\geq1$, 
we have the Lenard-Magri recursion
\begin{equation}\label{leneq}
\big\{ \tint h_k , u\big\}_H
=
\big\{ \tint h_{k+N} , u\big\}_K
\,,\,\,
\text{ for all } u\in\mc V
\,.
\end{equation}
\end{enumerate}
\end{theorem}
\begin{proof}
For $1\leq\varepsilon<N$, we have $L^{\frac{\varepsilon}{N}-1}(w)_+=0$,
and therefore, recalling \eqref{20140707:eq3b-lie},
equation \eqref{leneq2} holds by Lemma \ref{cor:17032013}(b).
Moreover, 
$\mult\ldb h_N{}_\lambda L(w)\rdb_K\big|_{\lambda=0}
=L(w)-L(w+\partial)\cdot1=0$.
This proves part (a).
For part (b),
by Lemma \ref{cor:17032013},
the recursion \eqref{leneq} holds for $u=L(w)$,
the generating series of the generators of $\mc V$.
Hence, \eqref{leneq} holds for all $u\in\mc V$  by the left Leibniz rule.
\end{proof}
\begin{remark}\label{20132507:rem1}
It follows from Lemma \ref{cor:17032013} 
that the Hamiltonian equation corresponding to the Hamiltonian functional
$\tint h_k$, $k\geq1$, can be written as (in terms of generating series)
\begin{equation}\label{laxpair}
\frac{dL(w)}{dt_k}=[(L^{\frac kN})_+,L](w)\,,
\end{equation}
where on the RHS we have to take the symbol of the usual commutator
of pseudodifferential operators.
This equation is the symbol of the usual \emph{Lax pair} representation
of the AGD hierarchies of Hamiltonian equations.
\end{remark}
\begin{example}\label{ex:KP}
On $\mb F\langle u_i^{(n)}\mid i,n\in\mb Z_+\rangle$,
we have
$L(\partial)=\partial+\sum_{i\in\mb Z_+}u_i\partial^{-i-1}$.
By an explicit computation we get
\begin{align*}
&L^2(\partial)
=\partial^2+2u_0+(2u_1-u_0')\partial^{-1}
+(2u_2+u_1'+u_0^2)\partial^{-2}+\dots\,,\\
& L^3(\partial)
=\partial^3+3u_0\partial+3(u_1+u_0')
+(3u_2+3u_1'+3u_0^2+u_0'')\partial^{-2}+\dots\,.
\end{align*}
Hence, the first few integrals of motion are
$$
\tint h_1=\tint \tr(u_0)
\,,\,\,
\tint h_2=\tint \tr(u_1)
\,,\,\,
\tint h_3=\tint \tr(u_2+u_0^2)
\,,\dots
$$
To find the corresponding bi-Hamiltonian equations,
we use Lemma \ref{cor:17032013}.
We have
$L(w)_+=w$,
$L^2(w)_+=w^2+2u_0$,
$L^3(w)_+=w^3+3u_0w+3(u_0'+u_1)$.
Hence, 
\begin{equation}\label{KP}
\begin{array}{l}
\displaystyle{
\frac{dL(w)}{dt_1}=\partial L(w)
\,\,,\,\,\,\,
\frac{dL(w)}{dt_2}
=\partial^2L(w)+2w\partial L(w)+2(u_0L(w)-L(w+\partial)u_0)
\,,} \\
\displaystyle{
\frac{dL(w)}{dt_3}=
\partial^3L(w)+3w\partial^2L(w)+3w^2\partial L(w)+3u_0\partial L(w)
\,,} \\
\displaystyle{
\,\,\,\,\,\,\,\,\,\,\,\,\,\,\,\,\,\,\,\,\,\,\,\,\,\,\,\,\,\,\,\,\,\,\,\,\,\,\,\,\,\,\,\,\,
+3\big(((w+\partial)u_0+u_1)L(w)-L(w+\partial)((w+\partial)u_0+u_1)\big)
\,\dots}
\end{array}
\end{equation}
Consider the first two equations in the second system of the hierarchy \eqref{KP},
and the first equation in the third system of \eqref{KP}. Namely,
$$
\left\{
\begin{array}{l}
\displaystyle{
\frac{du_0}{dt_2}=u_0''+2u_1'\,,
}\\
\displaystyle{
\frac{du_1}{dt_2}=u_1''+2u_2'+2u_0u_0'+2u_0u_1-2u_1u_0\,,
}\\
\displaystyle{
\frac{du_0}{dt_3}=u_0'''+3u_1''+3u_2'+3u_0u_0'+3u_0'u_0\,.
}
\end{array}
\right.
$$
We can eliminate the variable $u_2$ from this system.
After relabeling $t_2=y$, $t_3=t$, $u=2u_0$ and $w=4u_1+2u_0'$, we get
the system
\begin{equation}\label{20140107:eq2}
\left\{
\begin{array}{l}
\displaystyle{
u_y=w'\,,
}\\
\displaystyle{
3w_y=4u_t-u'''-3(u^2)'+3[u,w]\,,
}
\end{array}\right.
\end{equation}
We call the system \eqref{20140107:eq2} the non-commutative Kadomtsev-Petviashvili (KP) equation.
\end{example}
\begin{example}\label{ex:KdV}
On $\mb F\langle u^{(n)}\mid n\in\mb Z_+\rangle$,
we have
$L(\partial)=\partial^2+u$.
By an explicit computation we get
\begin{align*}
&L^{\frac12}(\partial)
=\partial+\frac u2\partial^{-1}
-\frac{u'}{4}\partial^{-2}+\frac18(u''-u^2)\partial^{-3}\dots\,,\\
& L^{\frac32}(\partial)
=\partial^3+\frac32u\partial+\frac34u'
+\frac{1}{8}(3u^2+u'')\partial^{-1}+\dots\,.
\end{align*}
Hence, the first two integrals of motion are
$\tint h_1=\tint \tr(u)$ and $\tint h_3=\tint \frac14\tr(u^2)$.
To find the corresponding bi-Hamiltonian equations,
we use Lemma \ref{cor:17032013}.
We have
$L^{\frac12}(w)_+=w$ and
$L^\frac32(w)_+=w^3+\frac32uw+\frac34 u'$.
Hence, $\frac{du}{dt_1}=u'$ and
\begin{equation}\label{KdV}
\frac{du}{dt_3}=\frac14\left(
u'''+3uu'+3u'u
\right)\,.
\end{equation}
The equation \eqref{KdV} is the non-commutative KdV equation.
\end{example}
\begin{example}\label{ex:Boussinesq}
On $\mb F\langle u^{(n)},v^{(n)}\mid n\in\mb Z_+\rangle$,
we have
$L(\partial)=\partial^3+u\partial+v$.
By an explicit computation we get
\begin{align*}
&L^{\frac13}(\partial)
=\partial+\frac u3\partial^{-1}
+\frac{1}{3}(v-u')\partial^{-2}+\frac19(2u''-u^2-3v')\partial^{-3}\dots\,,\\
& L^{\frac23}(\partial)
=\partial^2+\frac23u
+\frac13(2v-u')\partial^{-1}+\frac{1}{9}(3u''-u^2-3v')\partial^{-2}+\dots\,.
\end{align*}
Hence, the first two integrals of motion are
$\tint h_1=\tint \tr(u)$ and $\tint h_2=\tint \tr(v)$.
To find the corresponding bi-Hamiltonian equations,
we use Lemma \ref{cor:17032013}.
We have
$L^{\frac13}(w)_+=w$ and
$L^\frac23(w)_+=w^2+\frac23u$.
Hence, the corresponding Hamiltonian equations are
$\frac{du}{dt_1}=u'$, $\frac{dv}{dt_1}=v'$ and
\begin{equation}\label{Boussinesq}
\left\{
\begin{array}{l}
\frac{du}{dt_2}=
2v'-u''
\\
\frac{dv}{dt_2}=v''-\frac23u'''-\frac23uu'+\frac23[u,v]\,.
\end{array}\right.
\end{equation}
We call the equation \eqref{Boussinesq} the non-commutative
Boussinesq equation (since it reduces to the classical Boussinesq equation
if $[u,v]=0$).
\end{example}
\end{document}